\documentclass[12pt]{article} 
\usepackage[sectionbib]{natbib}
\usepackage{array,epsfig,fancyhdr,rotating}

\usepackage{xr-hyper}
\usepackage[]{hyperref}  

\usepackage[font={scriptsize}]{caption}
\usepackage{subcaption}
\usepackage[linesnumbered,ruled,vlined]{algorithm2e}

\SetAlFnt{\small}
\graphicspath{ {./images/} }

\usepackage{tikz,amssymb}
\usetikzlibrary{decorations.pathreplacing}
\usetikzlibrary{calc}
\usepackage{fp}
\newcommand\bound{10} 
\usepackage{sectsty, secdot}
\renewcommand{\theequation}{\thesection\arabic{equation}}

\usepackage[margin=1in]{geometry}

\usepackage{graphicx}
\usepackage{framed}
\usepackage{amsmath}
\usepackage{amssymb}
\usepackage{amsfonts}
\usepackage{multirow}
\usepackage{amsthm}
\usepackage{arydshln,leftidx,mathtools}
\usepackage{cleveref}
\usepackage{xcolor}
\usepackage{color}

\newtheorem{theorem}{Theorem}
\newtheorem{lemma}{Lemma}

\newtheorem{proposition}{Proposition}
\theoremstyle{definition}

\newtheorem{remark}{Remark}
\usepackage{bm}
\usepackage{bbm}
\DeclareMathOperator*{\argmin}{arg\,min}

\SetKwInput{KwInput}{Input}
\SetKwInput{KwOutput}{Output}
\SetKwInput{KwInit}{Initialization}

\newcommand{\E}{\mathbb{E}}
\newcommand{\R}{\mathbb{R}}

\makeatletter
\newcommand*{\addFileDependency}[1]{
  \typeout{(#1)}
  \@addtofilelist{#1}
  \IfFileExists{#1}{}{\typeout{No file #1.}}
}
\makeatother

\begin{document}


\renewcommand{\baselinestretch}{1}



$\ $\par


\fontsize{12}{14pt plus.8pt minus .6pt}\selectfont \vspace{0.8pc}
\centerline{\large\bf \uppercase{A unified framework for change point} }
\vspace{2pt}  
\centerline{\large\bf \uppercase{ detection in high-dimensional linear models} }
\vspace{.4cm} 
\centerline{Yue Bai \quad \quad  Abolfazl Safikhani} 
\vspace{.4cm} 
\centerline{\it Department of Statistics, University of Florida}
 \vspace{.55cm} \fontsize{9}{11.5pt plus.8pt minus.6pt}\selectfont


\begin{quotation}
\noindent {\it Abstract:}
In recent years, change point detection for high dimensional data has become increasingly important in many scientific fields. 
Most literature develop a variety of separate methods designed for specified models (e.g. mean shift model, vector auto-regressive model, graphical model). In this paper, we provide a unified framework for structural break detection which is suitable for a large class of models.
Moreover, the proposed algorithm automatically achieves consistent parameter estimates during the change point detection process, without the need for refitting the model.  
Specifically, we introduce a three-step procedure.  The first step utilizes the block segmentation strategy combined with a fused lasso based estimation criterion, leads to significant computational gains without compromising the statistical accuracy in identifying the number and location of the structural breaks. This procedure is further coupled with hard-thresholding and exhaustive search steps to consistently estimate the number and location of the break points.
The strong guarantees are proved 
on both the number of estimated change points and the rates of convergence of their locations.
The consistent estimates of model parameters are also provided. 
The numerical studies provide further support of the theory and validate its competitive performance for a wide range of models.
The developed algorithm is implemented in the \textbf{R} package \textbf{LinearDetect}. 

\vspace{9pt}
\par
\end{quotation}\par

\def\thefigure{\arabic{figure}}
\def\thetable{\arabic{table}}

\renewcommand{\theequation}{\thesection.\arabic{equation}}


\section{Introduction}\label{sec:intro}


Developing methods to detect change points (break points) in dynamical systems have become increasingly important due to the wide range of applications in many real life problems, including quality control \citep{qiu2013introduction}, neuroscience \citep{OmbaoVonSachsGuo_2005}, economics and finance \citep{frisen2008financial}, and social network analysis \citep{savage2014anomaly}, just to name a selected few. A change point represents a discontinuity in the parameters of the data
generating process. The literature has investigated both the {\em offline} and {\em online} versions of the problem \citep{Basseville1993detection,
csorgo1997limit}. In the former case, one is given a sequence of observations and questions of interest include: (i) whether there exist change (break) points and (ii) if there exist change points, identify their locations, as well as estimate the parameters of the data generating process. 
In the online case, one sequentially obtains new observations and the main interest is in {\em quickest detection} of the change point \citep{wang2015large,chan2021self}.



Fused lasso \citep{rinaldo2009properties} is among computationally attractive offline change point detection methods due to its linear computation time with respect to sample size \citep{bleakley2011group}. In this method, first the parameter space is expanded to allow model parameters to change at all time points while parameters' consecutive differences are fused (forced to zero) to reduce the parameter space dimension. It is known that fused lasso over-estimates the number of change points, i.e. it has a non-vanishing false positive rate \citep{harchaoui2010multiple}, while there is no unified result in deriving upper bounds for the total positive rate of fused lasso. As a result, additional steps are typically combined with fused lasso in order to consistently estimate the number of change points, see e.g. the screening step in \cite{safikhani2020joint}. These additional steps typically include several hyper-parameters and the finite sample detection performance can be sensitive to small changes in these hyper-parameters. Further, the theoretical rates of such hyper-parameters depend on the model and need to be derived separately for each statistical model under consideration. Note that despite these issues, fused lasso is among attractive detection algorithms due to its computational speed compared to more exhaustive search methods such as dynamic programming which has at least quadratic computation time with respect to sample size which makes it not scalable to large scale (and high-dimensional) data sets.  


In this paper, we propose a new detection algorithm called Threshold Block Fused Lasso (TBFL) which is motivated by fused lasso while the fused lasso issues mentioned are mitigated by specific modifications developed in the new algorithm. Unlike fused lasso, TBFL can consistently estimate the number of change points in a single step while its computational complexity is similar to fused lasso (or better sometimes, see details in Remark~\ref{rk:1}). Further, location of change points are estimated consistently by developing a local exhaustive search step. The proposed algorithm is flexible and can handle break detection in a wide range of statistical models. In this paper, we focus on detection of break points and model parameter estimation for general sparse multivariate regression models with high-dimensional covariates \citep{rothman2010sparse}. In this model (model~\ref{eq:model_1}), both response variable and covariates are multivariate and their dimensions can potentially be much larger than the sample size. Moreover, unlike typical regression models, independence among covariates in different samples is not assumed (see more details in Sections~\ref{sec:model} and \ref{sec:thm}). This makes the model flexible enough to include a wide range of models (with possible temporal and/or spatial correlations) including mean shift models \citep{harchaoui2010multiple}, multiple linear regression model \citep{leonardi2016computationally}, vector auto-regressive models \citep{lutkepohl2005new}, Gaussian graphical models \citep{yuan2007model}, and network auto-regressive model \citep{zhu2017network}.

TBFL starts with partitioning the time domain into certain blocks while assuming the model parameters remain fixed within each block and change among neighboring blocks. The block sizes ($b_n$ with $n$ as the sample size) are selected carefully to control false positive rates while not missing any true break point. Then, model parameters among all blocks are estimated simultaneously using regularized estimation procedures motivated by fused lasso and further,
{Frobenius}
norm of differences between estimated model parameters in consecutive blocks are computed which are called ``jumps". Intuitively, a large magnitude of jump implies that there exists a true break point inside the neighboring blocks while a small jump can potentially be due to finite sample estimation error. Thus, jumps are thresholded using a certain data-driven threshold and only block ends corresponding to jumps above the threshold are regarded as ``candidate" change points. Note that the hard-thresholding technique has been used in lasso regularization to reduce the false positive rate \citep{van2011adaptive}, while thresholding hasn't been fully investigated for fused lasso. It is verified (Theorem~\ref{thm_selection_block}) that under certain conditions, this procedure leads to a set of ``clusters" of candidate change points while the number of clusters matches with the true number of break points in the model (denoted by $m_0$) with high probability. As a bi-product of this result, it can be seen that the total number of candidate change points is at most $2m_0$ with high probability converging to one as the sample size diverges. This can be interpreted as an upper bound to control the false positive rate, a result not available for fused lasso for such a general linear model. Moreover, a simple exhaustive search within each estimated cluster gives the final estimation for location of break points. Non-asymptotic consistency rates of final estimates of break point locations are derived (Theorem~\ref{thm:final_consistency_rate}) where it can be seen that change point estimates are optimal up to a logarithmic factor (see more details in Section~\ref{sec:thm}). Model parameter estimates after break detection has not received much attention while having consistent estimators for model parameters before and after break points can reveal the main drives of breaks in the system. This could provide valuable insights to scientists to decipher the main features which contributed to the shock/break in the system (for example, see the application of TBFL on an EEG data set in Section~\ref{sec:eeg}). Interestingly, estimated parameters within TBFL can be utilized to develop model parameter estimates between any two consecutive break points \textit{without} refitting and their consistency is derived as well (see Theorem~\ref{thm:estimation}). Steps of the TBFL algorithm are illustrated in Figure~\ref{fig:jumps_step}.
A random realization from model~\ref{eq:model_1} is generated with sample size $n=1000$, $p_x=20$, $p_y=1$, two true change points at $333$ and $666$ (solid red lines) with the block size of $b_n=30$. In Figure~\ref{fig:jumps_step}, 
{square of jump sizes (i.e., square of
{Frobenius}
norm of differences between estimated model parameters in consecutive blocks) at block ends ($30,60, 90, \ldots$) are plotted in all panels (see more details about the model settings in supplementary material~\ref{sec:plot}).} It can be seen from the left panel that there are large jumps close to two break points while there are some small jumps far from any true break point. Thresholding (green horizontal dashed line) can help removing those small jumps. The middle panel depicts clusters of candidate change points in neighborhoods of true break points. It can be seen that $3$ candidate break points remain after thresholding which matches with Theorem~\ref{thm_selection_block} which states that there should be at most $2m_0=4$ candidate break points. Finally, the right panel illustrates the final estimated break points as blue vertical dashed lines using local search within each cluster.

\begin{figure}[!ht]
\begin{center}
\includegraphics[width=0.3\linewidth, clip=TRUE, trim=0mm 0mm 00mm 10mm]{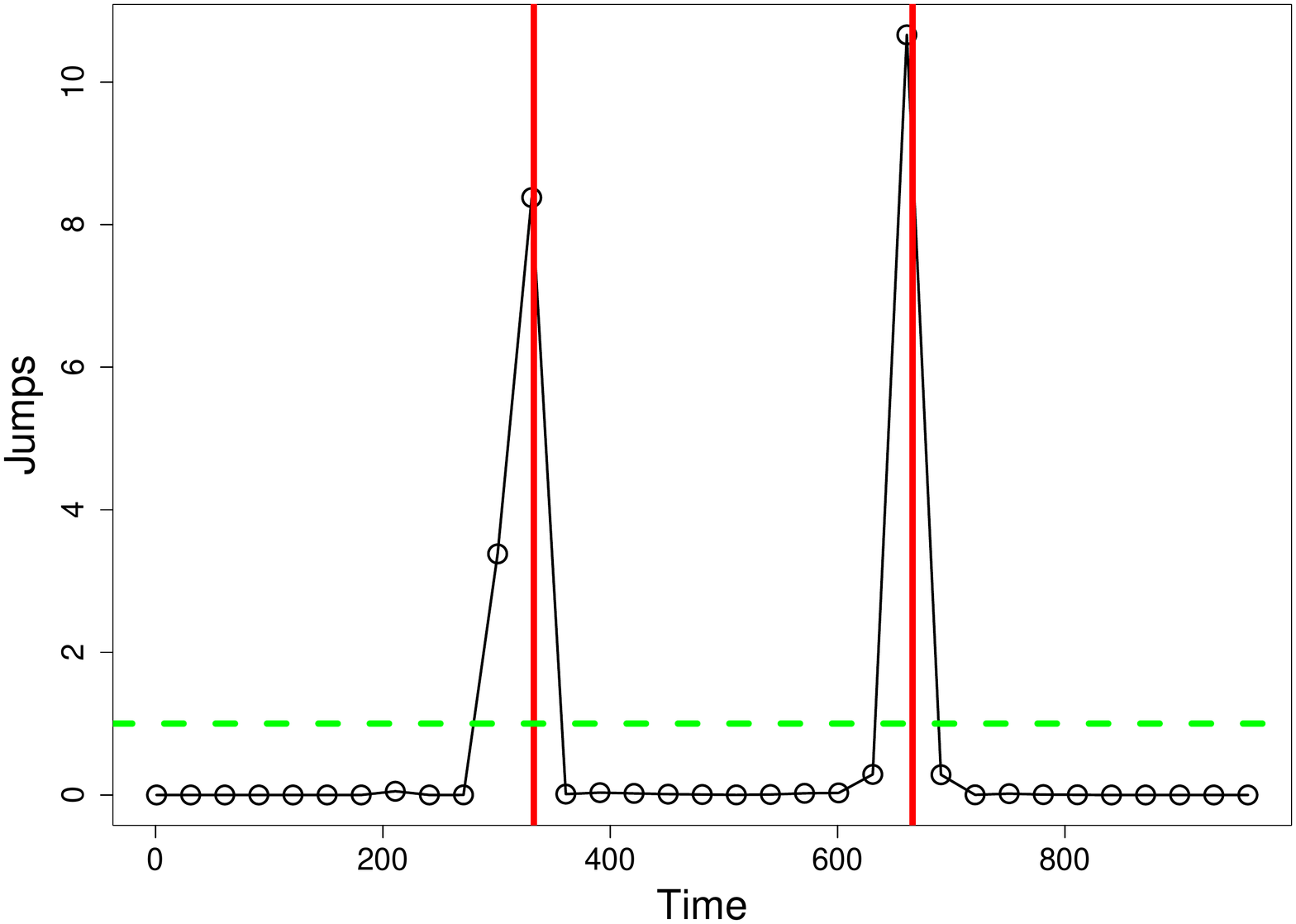}
\includegraphics[width=0.3\linewidth, clip=TRUE, trim=0mm 0mm 00mm 10mm]{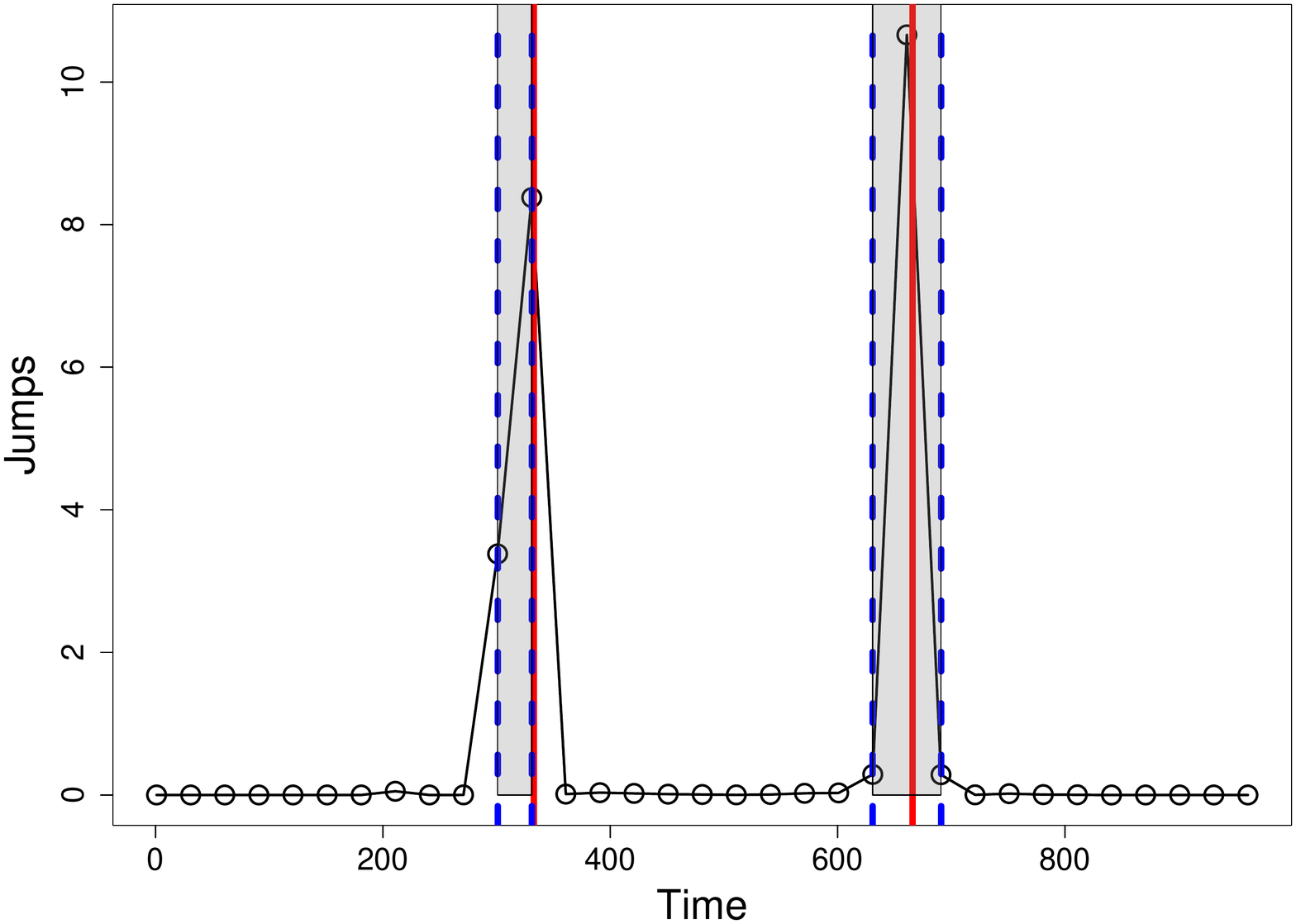}
\includegraphics[width=0.3\linewidth, clip=TRUE, trim=0mm 0mm 00mm 10mm]{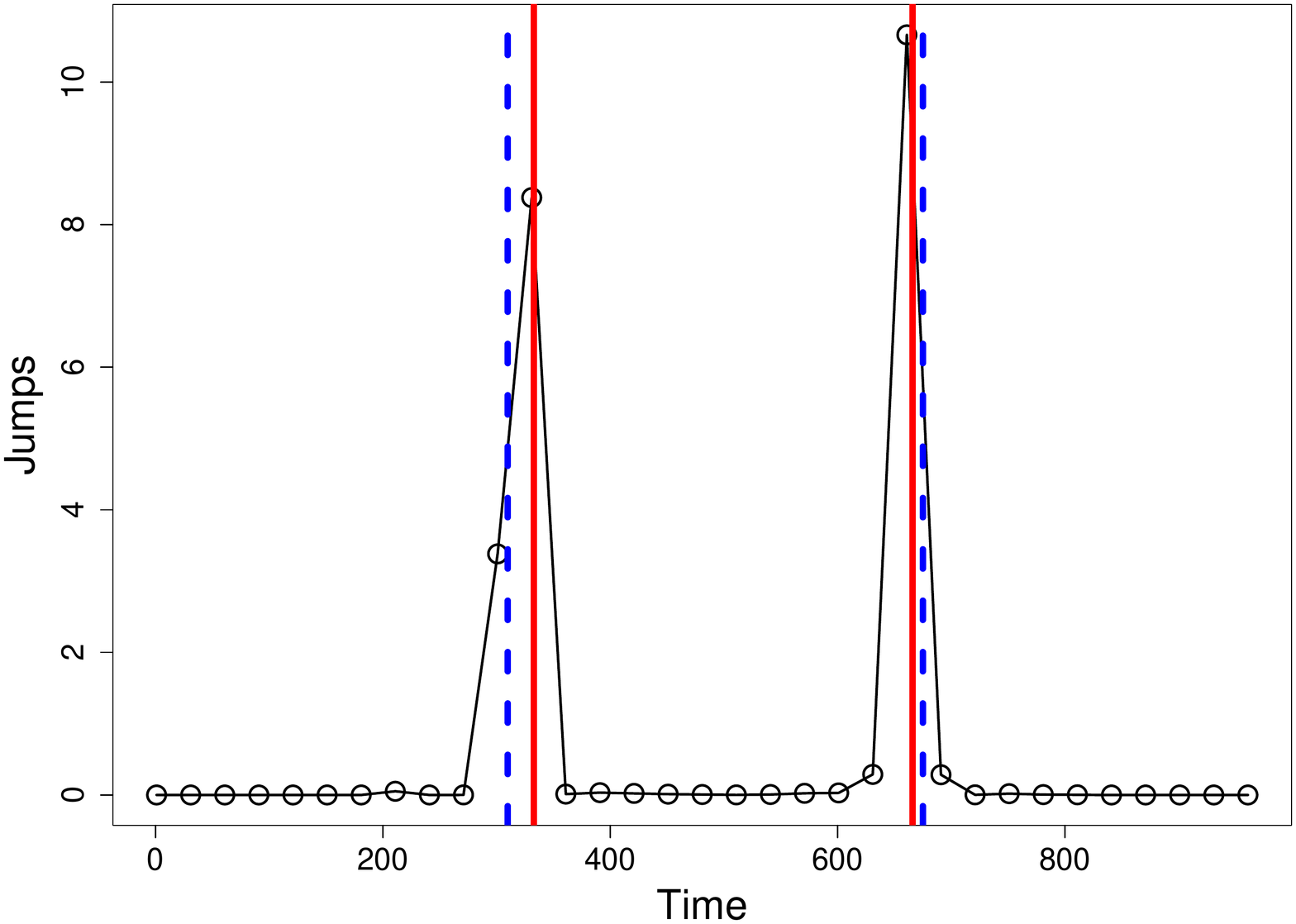}
\caption{Illustration of the TBFL algorithm.}
\label{fig:jumps_step}
\end{center}
\end{figure}

{In summary, the main contributions include (a) proposing a detection algorithm which can handle a wide range of linear  models- including change-in-mean model, multiple linear regression model, Vector Auto-Regressive (VAR) model, and Gaussian graphical model- in both high-dimensional and fixed-dimensional cases; (b) providing theoretical guarantees in terms of consistency rate of change point detection and parameter estimation; (c) providing consistent model parameter estimates during the change point detection process, \textit{without} the need for refitting the model; (d) providing data-driven methods to select all hyper-parameters in the algorithm. Further, the algorithm is implemented in the \textbf{R} package
\textbf{LinearDetect} \citep{lineardetect}. Next, a brief overview of existing detection methods is provided.}

\vspace{-0.5cm}

\subsection{Related Works}\label{sec:related work}

There exists a large body of literature addressing the problem of change point detection in the offline version mostly focusing on fix-dimensional regimes. The literature can be categorized into three groups with respect to the dimension of coefficient parameters considered by the model - i.e. univariate, multivariate, and high-dimensional. There exist several works focusing on different types of models in the univariate case.
For example,  \cite{davis2006structural}  utilize minimum description length principle to locate change points in piecewise univariate auto-regressive models while
\cite{killick2012optimal} propose a Pruned Exact Linear Time (PELT) method  using the Optimal Partitioning approach of \cite{jackson2005algorithm}, together with a pruning step within the dynamic program to detect the structural breaks. \cite{fryzlewicz2017tail} applies a tail-greedy Haar transformation to consistently estimate the number and locations of multiple change points in the univariate piecewise-constant model while
\cite{aue2017detecting} develop a method based on the (scaled) functional cumulative sum (CUSUM) statistic for detecting shifts in the mean of functional data models.
In the multivariate case with the number of the model parameters $p$ fixed,
\cite{OmbaoVonSachsGuo_2005} develop a spectral representation to locate the break points, a self-normalized technique is developed in \cite{zhang2018unsupervised} to test for change points, \cite{matteson2014nonparametric} propose a non-parametric approach based on Euclidean distances between sample observations. There has been an increasing interest recently in the high-dimensional case in which the number of model parameter $p$ is much larger than the number of observations $n$ \citep{hastie2009elements}. \cite{cho2015multiple} and \cite{cho2016change} employed Binary Segmentation for locating break points in high-dimensional data while \cite{wang2016high} proposed a high-dimensional change point detection method using a sparse projection to project the high-dimensional into a univariate case. The algorithm for estimating a single change point can be combined with the wild binary segmentation scheme of \cite{frick2014multiscale} to locate sequentially multiple change points in high-dimensional time series. In \cite{wang2019localizing}, $l_0$-optimization is developed for change point detection in Vector auto-regressive (VAR) models while \cite{roy2017change} developed a likelihood-based method for locating a single break point for high-dimensional Markov random fields and provide the rate of estimating the change point, as well as the model parameters. Further, a U-statistic-based cumulative sum statistic is developed in \cite{liu2020unified} to test for existence of a single change point while \cite{safikhani2020joint}, \cite{bai2020multiple} and \cite{safikhani2021fast} utilize fused lasso \citep{tibshirani2005sparsity} and a screening step to estimate multiple break points in a VAR model and also establish consistency results for both the break points and the model parameters. Moreover, \cite{kolar2012estimating} consider a fused lasso regularization together with a neighborhood selection approach to detect the change points in the Gaussian graphical model while \cite{bybee2018change} introduce majorize-minimize algorithm plus Simulated Annealing (SA) algorithm for computing change points in large graphical models. Finally, \cite{gibberd2017multiple} utilize Group-Fused Graphical Lasso
(GFGL) to detect multiple change points in high-dimensional setting. We refer to two recent review papers \cite{aue2013structural,yu2020review} for a more comprehensive review.

The remainder of the paper is organized as follows. In Section~\ref{sec:model}, the general model formulation is introduced while in Section~\ref{sec:algo}, we provide a detailed description of the proposed TBFL algorithm.  Asymptotic properties including the consistency of the number of change points and their locations are established in Section~\ref{sec:thm} while examples of models are provided in Sections~\ref{sec:example} and \ref{sec:MLR} (supplementary materials). Further, the optimal block size selection method is discussed in Section~\ref{sec:optimalblock}.
The comparison with other methods together with
numerical performance of the proposed TBFL in various simulation settings
are provided in Sections \ref{sec:comparison}, \ref{Sim_Scenarios} and \ref{sec:sa} (supplementary materials). 
Finally, the real data application of electroencephalograms (EEGs) recorded during eyes-closed and eyes-open resting conditions is presented in Section~\ref{sec:eeg} while Section~\ref{sec:conclustion} includes some concluding remarks.


\textbf{Notation}:
Denote the indicator function of a subset $S$ as $\mathbbm{1}_{S}$. 
For any vector $v \in \mathbb{R}^p$,
we use $\Vert v \Vert_{\infty}$  to denote  $\max_{1 \leq i\leq p}\{\vert v_i\vert\}$.
For any matrix $\bm{A}$, the $\ell_1$, $\ell_2$ and $\ell_{\infty}$  norms of the vectorized form of $\bm{A}$ are denoted by $\Vert \bm{A} \Vert_1 = \Vert \text{vec}(\bm{A}) \Vert_1$, $\Vert \bm{A} \Vert_F = \Vert \text{vec}(\bm{A}) \Vert_2$ and $\Vert \bm{A} \Vert_{\infty} = \Vert \text{vec}(\bm{A}) \Vert_{\infty}$. 
The transpose of a matrix $\bm{A}$ is denoted by $\bm{A}^\prime$. Let $\Lambda_{\text{max}}(\Sigma)$ and $\Lambda_{\min}(\Sigma) $  denote the maximum and minimum eigenvalues of the symmetric  matrix $\Sigma$.
Denote the tensor product of two matrices as $\otimes $.
For functions $f (n)$ and $g(n)$, we write $f(n) = \Omega(g(n))$ if and only if for some constants $c \in (0, \infty) $ and $n_0 >0$,  $f(n) \geq cg(n)$ for all $n \geq n_0$; we write $f(n) = O(g(n))$ if and only if for some constants $c \in (0, \infty) $ and $n_0 >0$,  $f(n) \leq cg(n)$ for all $n \geq n_0$.
We define the Hausdorff distance between two countable sets on the real line as $d_H (A, B) = \max \left\{ \max_{b\in B} \min_{a \in A} \vert b- a \vert, \max_{a\in A} \min_{b \in B} \vert b- a \vert  \right\}.$
{For scalars $a$ and $b$, define $a \wedge b = \min(a, b)$ and  $a \vee b = \max (a, b)$.}


\section{Model Formulation}\label{sec:model}
We consider a multivariate  regression model \citep{rothman2010sparse} with structural break such that the values of coefficient matrix change over time in a piece-wise constant manner. Specifically, suppose there exist $m_0$ change points $\left \{t_1 , \dots , t_{m_0} \right\}$ such that $1 =t_0 <t_1 <\dots<t_{m_0} <t_{m_0+1} =n+1$, then the structural break multivariate regression model is given by 
\begin{equation}\label{eq:model_1}
    \bm{y}_t = \sum_{j = 1}^{m_0+1}\left (\bm{B}_j^\star \bm{\bm{x}}_t + \bm{\varepsilon}_{j,t} \right ) \mathbbm{1}_{\{t_{j-1}\leq t <t_j \}}, \ t =  1, \dots, n, 
\end{equation}
where  $\bm{y}_t \in \R^{p_y}$ is the response vector at time $t$; $\bm{B}_j^\star \in \R^{p_y \times p_x}$ is the true coefficient matrix during the $j$th segment;
{
$\bm{x}_t \in \R^{p_x}$ is the  predictor vector at time $t$;
and $\bm{\varepsilon}_{j,t} \in \R^{p_y}$ is a multivariate white noise during the $j$th segment at time $t$. }
All parameters in the model are considered fixed during each segment, while the coefficient matrices $\bm{B}_j^\star$ are allowed to vary over segments. 
The multivariate regression model requires one to estimate $p_xp_y$ parameters within each segment which will be challenging when either the number of predictors $p_x$ or the number of response $p_y$ become large. We work under the high-dimensional setting in which we allow the number of predictors $p_x$ and the number of response $p_y$  to grow by sample size and possibly exceed the sample size $n$, i.e. $p_x \gg n$ and/or $p_y \gg n$. As a result, we assume sparsity of coefficient matrices $\bm{B}_{j}^\star$'s. Specifically, denote the number of nonzero elements 
in $\bm{B}_{j}^\star$ by $d_{j}$, 
$j = 1, 2, \dots , m_0+1$.
Let $d_n^\star =  \max_{1 \leq j \leq m_0+1} d_{j}$ to be the maximum sparsity of the model. We assume that $d_n^\star$ is much smaller than $p_x$ and $p_y$, see more details in Section~\ref{sec:thm}.

\section{Proposed Algorithm}\label{sec:algo}

In this section, we introduce a three-step estimation procedure denoted by Threshold Block Fused Lasso (TBFL).
The first step aims to select candidate change points among blocks and estimate each segment's coefficient matrix by solving a block fused lasso problem.
A hard-thresholding step is then added to reduce the over-selection problem from the fused lasso step. 
In the third step, a local exhaustive search examines every time point inside a neighborhood region based on the cluster of candidate change points estimated in the previous step. Moreover, a consistent model parameter estimate is also obtained during the block fused lasso step. Each step is described in details next.




\noindent
{\bf (Step I) Block Fused Lasso.} Define a sequence of time points $1 = r_0 < r_1 < \dots < r_{k_n} = n + 1 $ for block segmentation such that $r_{i} - r_{i-1} \approx b_n $ for $i =1 ,\dots, k_n-1$, where $k_n = \lfloor {n}/{b_n} \rfloor$ is the total number of blocks. To simplify notation and without loss of generality, throughout the rest of the paper, we assume that $n$ is divisible by $b_n$ such that $r_{i} - r_{i-1} = b_n$ for all $i =1 ,\dots, k_n$.
By partitioning
the observations into blocks of size $b_n$ and fixing the model parameters within each block, 
we set $\bm{\Theta}_1 = \bm{B}_1^\star$ and $\bm{\Theta}_i=\bm{B}_{j+1}^\star -\bm{B}_j^\star$ when $t_j \in [r_{i-1}, r_{i})$ for some $j$, and $\bm{\Theta}_i=0$ otherwise, for $i = 2,3,\dots, k_n$. Note that $\bm{\Theta}_i \neq \bm{0}$ for $i \geq 2$ means that $\bm{\Theta}_i$ has at least one non-zero entry and implies a change in the coefficients. We now formulate the following linear regression model in terms of $\bm{\Theta}(k_n) = (\bm{\Theta}_1, \dots, \bm{\Theta}_{k_n})^\prime $:

{\scriptsize
\begin{equation}\label{block_model_1}
\underbrace{\left(\begin{array}{c}
\bm{y}_{(1: r_1-1)}\\
\bm{y}_{(r_1: r_2-1)}\\
\vdots \\
\bm{y}_{(r_{k_n-1}:r_{k_n}-1)}\\
\end{array}\right)}_{\mathcal{Y}}= 
\underbrace{\left(\begin{array}{cccc}
\bm{x}_{(1: r_1-1)}  & \mathbf{0} &\dots & \mathbf{0}\\
\bm{x}_{(r_1: r_2-1)} & \bm{x}_{(r_1: r_2-1)} &\dots &\mathbf{0}\\
\vdots & \vdots & \ddots & \vdots \\
\bm{x}_{(r_{k_n-1}:r_{k_n}-1)} & \bm{x}_{(r_{k_n-1}:r_{k_n}-1)}& \dots  &\bm{x}_{(r_{k_n-1}:r_{k_n}-1)}\\
\end{array}\right)}_{\mathcal{X}}
\underbrace{
\begin{pmatrix}
{\bm{\Theta}_{1}}^\prime \\
{\bm{\Theta}_{2}}^\prime\\
\vdots \\
{\bm{\Theta}_{k_n}}^\prime\\
\end{pmatrix}}_{\bm{\Theta}(k_n)}
+\underbrace{
\left(\begin{array}{c}
{\bm{\zeta}_{(1: r_1-1)}}\\
{\bm{\zeta}_{(r_1: r_2-1)}}\\
\vdots \\
{\bm{\zeta}_{(r_{k_n-1}:r_{k_n}-1)}}\\
\end{array}\right)}_{E},   
\end{equation}}
where $\bm{y}_{(a:b)}:= (\bm{y}_a, \dots, \bm{y}_{b})^\prime$, $\bm{x}_{(a:b)}:= (\bm{x}_a, \dots, \bm{x}_{b})^\prime$,
${\bm{\zeta}_{(a:b)}:= (\bm{\zeta}_a, \dots, \bm{\zeta}_{b})^\prime}$; 
$\mathcal{Y} \in \R^{n \times p_y} $, $\mathcal{X} \in \R^{n \times k_n p_x} $, $\bm{\Theta}(k_n) \in \R^{k_np_x \times  p_y}$ and $E \in \R^{n \times p_y} $. Letting $\pi_n =k_np_xp_y $, $\mathbf{y} = \text{vec}(\mathcal{Y})$,  $\mathbf{Z} = I_{p_y} \otimes  \mathcal{X} $, $\bm{\theta} = \text{vec}(\mathbf{\Theta}(k_n))$, the regression model~\eqref{block_model_1} into vector form can be written as $\mathbf{y} =\mathbf{Z}\bm{\theta} + \text{vec}(E)$, where  $\mathbf{y} \in \R^{np_y }$, $\mathbf{Z} \in \R^{np_y \times \pi_n} $, $\bm{\theta} \in \R^{\pi_n }$ and $\text{vec}(E) \in \R^{np_y}$. Due to sparsity of parameter $\bm{\theta}$, one can estimate it by using an $\ell_1$-penalized least squares regression of the form:
 \begin{equation}\label{eq:estimation_block}
\widehat{\bm{\theta}} = \argmin_{\bm{\theta} \in \R^{\pi_n}} \left\{ \frac{1}{n} \| \mathbf{y} -\mathbf{Z}\bm{\theta} \|_2^2 + \lambda_{1,n} \|\bm{\theta} \|_1  + \lambda_{2,n} \sum_{i=1}^{k_n} \left\| \sum_{j=1}^{i} \bm{\Theta}_j \right\|_1  \right \} ,
\end{equation}
which uses a fused lasso penalty to control the number of change points and a lasso penalty to control the sparsity of the coefficient parameter in the model. Denote the sets of indices of  blocks with non-zero jumps  and estimated change points obtained from solving \eqref{eq:estimation_block} by 
{\small \[\widehat{I}_n =\left \{\widehat{i}_1, \widehat{i}_2, \dots, \widehat{i}_{\widehat{m}}\right\}= \left \{i :{\left\| \widehat{\bm{\Theta}}_i\right \|_F} \neq 0,\  i = 2,\dots, k_n \right\} \text{ and } \widehat{\mathcal{A}}_n = \left \{\widehat{t}_1, \dots, \widehat{t}_{\widehat{m}}\right\}= \left \{r_{i -1} : i \in  \widehat{I}_n\right\}, \]}
where $ \widehat{m} = \left| \widehat{\mathcal{A}}_n \right| $ and $\widehat{\mathcal{A}}_n \subset \left\lbrace {r}_1, \dots,{r}_{k_n-1}   \right\rbrace $. A data-driven method to select the optimal block size is introduced in Section~\ref{sec:optimalblock}.
\noindent
{\bf (Step II) Hard-thresholding Procedure.} The estimated change points obtained from solving \eqref{eq:estimation_block}
in the block fused lasso step include all block-end points with non-zero $\widehat{\bm{\Theta}}$,   
which lead to an over-estimation of the number of true change points in the model.
To remedy this issue, a hard-thresholding step to ``thin out'' redundant change points with small changes in the estimated coefficients is introduced. Intuitively, we keep estimated change points from the first step whose jumps are large enough (above a threshold). Specifically, denote the sets of indices of candidate blocks and estimated change points after hard-thresholding by 
 \[\widetilde{I}_n = \left \{i :\left\|\widehat{\bm{\Theta}}_i\right \|_F  > \omega_n,\  i = 2,\dots, k_n \right\} \text{ and } \widetilde{A}_n  = \left\{\widetilde{t}_1 ,\dots,  \widetilde{t}_{\widetilde{m} } \right\} = \left \{r_{i-1} : i \in  \widetilde{I}_n\right\}, \]
where $\omega_n$ is proportional to the minimum jump sizes
$\nu_n = \min_{1\leq j \leq m_0} \left\lVert \bm{B}_{j+1}^\star-\bm{B}_{j}^\star \right\rVert_F $.
Given the fact that the $\nu_n $ is unknown, we introduce a data-driven procedure to select a threshold value $\omega_n$ (see details in supplementary material~\ref{subsec:hard-thresholding-tuning}).

\noindent
{\bf (Step III) Exhaustive Search Procedure.} After hard-thresholding, the candidate change points that are located far from any true change points have been removed. However, there may be more than one selected change points remaining in the set $\widetilde{\mathcal{A}}_n$ in the neighborhoods of each true change point. Thus, we cluster the remaining estimated change points based on how close they are to each other. The idea is that the number of clusters is a reasonable estimate for $m_0$, the number of true change points. We consider a block clustering step which is based on a data-driven procedure to partition the $\widetilde{m}$ candidate change points into $\widetilde{m}^f$ clusters. In particular, we select the optimal number of clusters based on the idea of Gap statistics \citep{tibshirani2001estimating} (see details in supplementary material \ref{subsec:blockcluster}). For a set $A$, define cluster $(A)$ to be the partition of $A$ based on the clustering algorithm. Denote the subset in $ \mbox{cluster} \left( \widetilde{\mathcal{A}}_n \right)$ by $ \mbox{cluster} \left( \widetilde{\mathcal{A}}_n\right) = \left\lbrace R_1,  \ldots, R_{\widetilde{m}^f} \right\rbrace $, where $\widetilde{m}^f = \left\vert  \mbox{cluster} \left( \widetilde{\mathcal{A}}_n \right) \right\vert$. Denote the set of corresponding indices clusters by $ \mbox{cluster} \left( \widetilde{I}_n\right) = \left\lbrace J_1, J_2 \ldots, J_{\widetilde{m}^f} \right\rbrace $.

Next, we describe the local exhaustive search procedure to estimate location of change points. First, define the following local coefficient parameter estimates for each segment:
\begin{equation}\label{eq:beta_est_local}
{
 {\widehat{\bm{B}}}_{j} = \sum_{i=1}^{{\lfloor\frac{1}{2}\left(\max(J_{j-1})+\min(J_{j})\right) \rfloor} }\widehat{\bm{\Theta}}_i, \text{ for } j = 1, \dots, \widetilde{m}^f + 1,}
\end{equation}
where  $J_0 = \{1\}$, $J_{\widetilde{m}^f+1} = \{k_n\}$, and $\left \{{\widehat{\bm{\Theta}}}_{i}, i = 1,\dots, k_n \right  \}$ are matrix form parameters estimated from~\eqref{eq:estimation_block}. Define $l_j = (\min (R_j)-b_n) \mathbbm{1}_{\{\vert R_j\vert = 1\}} + \min (R_j) \mathbbm{1}_{\{\vert R_j\vert > 1\}}$ and $u_j = (\max (R_j)+b_n) \mathbbm{1}_{\{\vert R_j\vert = 1\}} + \max (R_j) \mathbbm{1}_{\{\vert R_j\vert > 1\}}$.
Now, given a subset $R_j$,  we apply the exhaustive search method for each time point $s$ in the interval $(l_j, u_j)$ to the data set truncated by the two end points in time, i.e. only consider the data within the interval $ [\min(R_j)-b_n , \max(R_j)+b_n) $. Specifically, define the final estimated change point $ \widetilde{t}_j$ as
\begin{equation}\label{eq:Exhaust:final}
\widetilde{t}_j^f = \argmin_{s \in (l_j, u_j)} \Bigg \{ \sum_{t= \min(R_j)-b_n}^{s-1} \left\| \bm{y}_{t}  - \widehat{\bm{B}}_{j}\bm{x}_{t} \right\|_2^2 +  \sum_{t= s}^{\max (R_j)+b_n-1} \left\| \bm{y}_{t}  - \widehat{\bm{B}}_{j+1}\bm{x}_{t} \right\|_2^2 \Bigg \},
\end{equation}
for $ j = 1, \ldots, \widetilde{m}^f $. Denote the set of final estimated change points from \eqref{eq:Exhaust:final} by $ \widetilde{\mathcal{A}}_n^f = \left\lbrace \widetilde{t}_1^f, \ldots, \widetilde{t}_{\widetilde{m}^f}^f  \right\rbrace $. Note that the local model parameter estimates ${\widehat{\bm{B}}}_{j}$'s defined in \eqref{eq:beta_est_local} can serve as estimation for parameters $\bm{B}_{j}^\star$'s. Thus, as mentioned in Section~\ref{sec:intro}, TBFL can estimate model parameters in parallel to change point detection without any refitting. To enhance the variable selection properties of model parameter estimates, we propose to hard-threshold ${\widehat{\bm{B}}}_{j}$'s. Specifically, define the thresholded estimate $\widetilde{\bm{B}}_{j}$ as  
\begin{equation}\label{est_coef_threshold}
{\widetilde{\bm{B}}}_{j} = \widehat{\bm{B}}_{j}\mathbbm{1}_{ \left\{ \vert \widehat{\bm{B}}_{j} \vert > \eta_{n, j}  \right\}}, \text{ for } j = 1, \dots, \widetilde{m}^f+1, 
\end{equation}
which is element-wise thresholding such that $ \widehat{\bm{B}}_{j,hl}  = 0$ if $\left\vert \widehat{\bm{B}}_{j,hl}  \right\vert  \leq \eta_{n,j} $ and unchanged otherwise, for all $j = 1, \dots, {\widetilde{m}^f}+1, h = 1,\dots, p_y, l = 1, \dots, p_x$. 
The thresholding parameter $\eta_{n, j}$ is selected using BIC criterion (see details in supplementary material~\ref{subsec:thresholding-tuning}).

\begin{remark}\label{rk:hard}
{
Note that the hard-thresholding (Step II) is only used for selecting potential change point locations with large changes in their estimated coefficients.
To guarantee a consistent estimation of
segment-specific model parameters $\bm{B}_j$'s,
those $\widehat{\Theta}_i$'s with smaller norm values are still kept in the local coefficient parameter estimates \eqref{eq:beta_est_local};
 See more discussion in supplementary material~\ref{subsec:para_est}.}
\end{remark}

\begin{remark}\label{rk:1}
The approximate computational complexity of TBFL method is $O \left( n/b_n+ b_n \right) $ for fixed $p_y$, $p_x$ and finite $m_0$.
The computational time  is of order $ O \left( n/b_n  \right) $ in the first step \citep{bleakley2011group}
and of order $ O \left( 2 b_n \right) $ in the   exhaustive search step.
 $b_n$ can essentially be selected as $n^{\epsilon}$ such that $0 \leq \epsilon < 1$.  Setting $\epsilon = 0$ (i.e., selecting $b_n$ as a constant)  yields to linear computational complexity (which matches the complexity in fused lasso). When $0 <\epsilon <1$, the computational complexity is of order $O \left( n^{\max(\epsilon, 1-\epsilon)} \right)$ which is \textit{sub-linear} with respect to the sample size.
\end{remark}

\section{Theoretical Properties}\label{sec:thm}

In this section, we provide asymptotic properties of TBFL in terms of both detection accuracy and model parameter estimation consistency. The following assumptions are needed:



\begin{itemize}
    \item[(A1.)] Lower restricted eigenvalue condition (Lower-RE condition). 
     There exist constants $c_1, c_2 > 0 $, a sequence $\delta_n \rightarrow + \infty$,
     {
     $a_n = \Omega( \log (p_xp_y \vee n)) $
     and parameters $\alpha_1 >0$ and $\tau = c_0\alpha_1 (a_n)^{-1} \log (p_xp_y \vee n)>0$ }
     such that with probability at least $ 1 - c_1 \exp(-c_2 \delta_n )  $ for all $ v \in \mathbb{R}^{p_xp_y}$, 
     \begin{equation}\label{eq:RE_bound}
\inf_{1 \leq j \leq {m_0+1}, t_j > u > l \geq t_{j-1}, | u - l  |  > a_n  }  v^\prime I_{p_y} \otimes \left({(l-u)}^{-1}  \sum_{t=l}^{u - 1} \bm{x}_{t}\bm{x}_{t}^\prime \right) v  \geq
\alpha_1 \Vert v \Vert_2^2 - \tau \Vert v \Vert_1^2.
\end{equation}

Upper restricted eigenvalue condition (Upper-RE condition). 
     There exist constants $c_1, c_2 > 0 $, a sequence $\delta_n \rightarrow + \infty$,
     {
     $a_n = \Omega( \log (p_xp_y \vee n) ) $ and parameters $\alpha_2 >0$ and $\tau = c_0\alpha_2{(a_n)}^{-1}{\log (p_xp_y \vee n)} >0$ }
     such that with probability at least $ 1 - c_1 \exp(-c_2 \delta_n )  $ for all $ v \in \mathbb{R}^{p_xp_y}$, 
     \begin{equation}\label{eq:upper-RE_bound}
\sup_{1 \leq j \leq {m_0+1}, t_j > u > l \geq t_{j-1}, | u - l  |  > a_n  }  v^\prime I_{p_y} \otimes \left({(l-u)}^{-1}  \sum_{t=l}^{u - 1} \bm{x}_{t}\bm{x}_{t}^\prime \right) v  \leq
\alpha_2 \Vert v \Vert_2^2 + \tau \Vert v \Vert_1^2.
\end{equation}
    

    \item[(A2.)]Deviation bound condition. 
    There exist constants $c_1, c_2 > 0 $ and a sequence $\delta_n \rightarrow + \infty$ such that with probability at least $ 1 - c_1 \exp(-\delta_n) $, for any sequence $a_n$,
\begin{equation}\label{eq:dev_bound}
{
\sup_{1 \leq j \leq {m_0+1},  t_j > u > l \geq t_{j-1}, | u - l  |  > a_n }  \left|\left| {(l-u)}^{-1}  \sum_{t=l}^{u - 1} \bm{x}_{t} \bm{\varepsilon}_t^\prime    \right|\right|_\infty \leq c_2 \sqrt{\frac{\log (p_xp_y \vee n)  }{a_n}}.}
\end{equation}
\item[(A3.)]The matrices $\bm{B}_{j}^\star$  are $d_j$-sparse. More specifically, for all $j = 1, \dots, m_0 + 1$, $d_{j} \ll p_x p_y$, i.e., $d_{j}/ (p_x p_y) = o(1)$. Moreover, there exists a positive constant $M_{\bm{B}}  > 0$ such that
    \[\max_{1 \leq j \leq m_0 +1}\left\lVert \bm{B}_{j}^\star \right\rVert_{\infty} \leq M_{\bm{B}} .\]
    
\item[(A4.)]
    Let $\nu_n = \min_{1\leq j \leq m_0} \left\lVert \bm{B}_{j+1}^\star-\bm{B}_{j}^\star \right\rVert_F $ and  $ \Delta_n = \min_{1 \leq j \leq m_0} |t_{j+1} - t_j| $.  
     There exists a  positive sequence $b_n$ such that, as $n \rightarrow \infty$,
\[{\frac{\Delta_n}{b_n }\rightarrow +\infty,   
\ d^\star_n \frac{\log( p_x p_y \vee n)}{b_n} \rightarrow 0,
\nu_n = \Omega\left(\sqrt{\frac{d^\star_n\log( p_x p_y \vee n)}{b_n}}\right)  
.}\]
\item[(A5.)] {The regularization
parameters $\lambda_{1,n}$ and $\lambda_{2,n}$  satisfy
$\lambda_{1,n} = C_1 \sqrt{{ \log( p_x p_y \vee n)}/{n}} \sqrt{{b_n}/{n}}$, and 
$\lambda_{2,n} = C_2\sqrt{{\log( p_x p_y \vee n)}/{n}}\sqrt{{b_n}/{n} }$ 
for some large constant $C_1, C_2 > 0$.}
\end{itemize}

Assumptions~A1 and A2 (known as restricted eigenvalue condition and deviation bound condition) are common assumptions in high-dimensional linear regression models \citep{loh2012} and hold for a wide range of models with possible temporal dependence \citep{Basu_2015}. These assumptions should hold uniformly over all $(m_0+1)$ segments due to changes in the model parameters. Assumption~A3 is related to the sparsity of the model which is needed due to the high-dimensionality of model (i.e. $p_x \gg n$ and $p_y \gg n$). Further, it puts an upper bound on the entries of coefficient matrices, which is a common assumption in change point detection literature (e.g. see Assumption A2 in \cite{safikhani2020joint}). 
Assumption~A4 connects several important quantities together including the minimum jump size required for coefficient matrices to make the change point detectable, the block size used in the TBFL algorithm, total sparsity allowed in the model, and the minimum spacing between consecutive change points. {Specifically, block size should be selected significantly smaller than $\Delta_n$ in order for TBFL not to miss any true break points (i.e., to ensure that there is at most one true change point in each block). The method can handle the case of diverging number of change points as well (i.e. $m_0 \to \infty$) as the sample size $n$ diverges.}
On the other hand, total sparsity allowed in the model, $d_n^*$, can increase proportionally to the block size $b_n$. Note that in the case of no change points, one can set $b_n=n$, thus the constrain on the model sparsity becomes similar to high-dimensional linear regression models with no change points \citep{loh2012}. Also, the higher the $b_n$, the smaller the jump size $\nu_n$ can be while the TBFL can still detect all change points in the model consistently. Finally, Assumption A5 specifies the rate of the tuning parameters $\lambda_{1,n}$ and $\lambda_{2,n}$ in the block fused lasso problem in \eqref{eq:estimation_block}. Note that again, in the case of no change points, one can pick $b_n=n$, and the rates in Assumption~A5 become the typical rates of tuning parameters in high-dimensional regression models \citep{loh2012}.

The first theorem is one of the main results about the false positive rate of the first step of TBFL as well as consistency of number of change points in the second step of TBFL.

\begin{theorem}\label{thm_selection_block}
Suppose A1-A5 hold. Then, as $n \rightarrow + \infty$
\[\mathbb{P} \left( d_H \left(\widetilde{\mathcal{A}}_n, \mathcal{A}_n \right)< b_n , m_0 \leq \left\vert \widetilde{\mathcal{A}}_n \right\vert \leq 2m_0 \mbox{ and } \widetilde{m}^f = m_0\right) \to 1.\]
     

\end{theorem}

Theorem~\ref{thm_selection_block} states that the number of clusters obtained in the second step of TBFL is a consistent estimator for the number of true change points $m_0$, despite the fact that the total number of estimated change points in this step can be larger than $m_0$. Note that although the number of candidate change points in the second step of TBFL can be larger than $m_0$, but Theorem~\ref{thm_selection_block} states that it can be at most $2m_0$ with high probability. Moreover, all candidate change points in the second step of TBFL are within $b_n$-radius neighborhood of a true change point with high probability. In other words, none of candidate estimated change points are far from true change points, a result not true for fused lasso \citep{safikhani2020joint}.




By utilizing the exhaustive search procedure (third step  of TBFL), one can remove additional candidate estimated break points within estimated clusters in the second step of TBFL. The next theorem states the main result on accuracy for locating break points in TBFL.

\begin{theorem}\label{thm:final_consistency_rate}
    Suppose the Assumptions A1-A5 hold. Then as $n \to +\infty$, there exists a large enough constant $K>0$ such that 
    \begin{equation*}
    {
        \mathbb{P}\bigg( \max_{1\leq j \leq m_0} \left|\widetilde{t}^{f}_j - t_j \right| \leq \frac{K {d_n^\star} \log (p_x  p_y \vee n)}{\nu_n^2}  \bigg) \to 1.}
    \end{equation*}
\end{theorem}

Theorem~\ref{thm:final_consistency_rate} states the localization error of TBFL algorithm uniformly over all $m_0$ change points. It scales logarithmically with respect to the model dimensions $p_x$ and $p_y$. Moreover, small jump sizes can potentially worsen the consistency rate for locating break points since the localization error scales proportionally with respect to the reciprocal of $\nu_n^2$. Note that the rate stated in Theorem~\ref{thm:final_consistency_rate} is optimal up to a logarithm factor \citep{csorgo1997limit}.  


Finally, consistent estimation of segment-specific model parameters can be achieved, as stated in the following theorem.

\begin{theorem}\label{thm:estimation}
    Suppose the Assumptions A1-A5  hold. 
    Then solution $\widehat{\bm{B}}_j$ from \eqref{eq:beta_est_local} satisfies 
    \begin{equation*}
    \max_{ 1 \leq j \leq m_0 + 1 } \left \Vert  \widehat{\bm{B}}_j - \bm{B}_j^\star \right\Vert_F =  O_p \left(\sqrt{\frac{d^\star_n {\log (p_xp_y \vee n) }}{b_n}}\right).
    \end{equation*}
Further, 
    {
     if {$\eta_{n,j} = {C_j}\sqrt{\frac{{\log (p_x p_y \vee n)} }{b_n}}$ for some positive constant $C_j$}, the thresholded variant $\widetilde{\bm{B}}_j$ in \eqref{est_coef_threshold}  satisfies
\begin{equation*}
   \max_{ 1 \leq j \leq m_0 + 1 }  \left \vert  \text{supp}(\widetilde{\bm{B}}_j) \backslash  \text{supp}( {\bm{B}}_j^\star) \right\vert =  O_p \left(d^\star_n\right).
    \end{equation*} }
\end{theorem}

Theorem~\ref{thm:estimation} states that the estimator $\widehat{\bm{B}}_j$ of model parameter has proper consistency results while its thresholded version $\widetilde{\bm{B}}_j$ satisfies satisfactory variable selection property. Note that in the case of no break points, by selecting $b_n=n$, the rates stated in Theorem~\ref{thm:estimation} match the typical consistency rates in high-dimensional regression models \citep{loh2012,Basu_2015}. Thus, the $b_n$ in the denominator of consistency rate in Theorem~\ref{thm:estimation} serves as a proxy for the sample size in each segment.

\section{Examples of Models}\label{sec:example}

In this section, we list two examples of specific well-known models which would fit into the modeling framework~\eqref{eq:model_1}. A third example on high-dimensional regression model is presented in supplementary materials, Section~\ref{sec:MLR}, due to space limitations.

\subsection{Mean Shift Model}
We consider a simple regression model that the values of mean change over time. 
In this case,  setting the parameters $\bm{x}_t = 1$, $\bm{B}_{j}^\star = \bm{\mu}_{j}^\star $, $p_x = 1$, $p_y = p$ in  the model representation  in \eqref{eq:model_1}, the  structural break mean shift model is given by 
\begin{equation}\label{eq:model_mean}
    \bm{y}_t = \sum_{j = 1}^{m_0+1}\left (\bm{\mu}_{j}^\star + \bm{\varepsilon}_{j,t} \right ) \mathbbm{1}_{\{t_{j-1}\leq t <t_j \}} , \   t =  1,\dots, n,
\end{equation}
where $\bm{\bm{y}}_t \in \R^{p}$ is the observation vector at time $t$; $\bm{\mu}_{j}^\star \in \R^{p}$ is the sparse mean vector during the $j$th segment;
{and  $\bm{\varepsilon}_{j,t} \in \R^{p}$ is
multivariate white noise during the $j$th segment at time $t$.}
Define $\bm{\Theta}_1 = \bm{\mu}_{1}^\star$,  $\bm{\Theta}_i = \bm{\mu}_{j+1}^\star -\bm{\mu}_{j}^\star$ when $t_j \in [r_{i-1}, r_{i})$ and $\bm{\Theta}_i = \bm{0}$ otherwise,
for $i = 2,3,\dots, k_n$.
In this case, the linear regression model in terms of ${\Theta}$ can be written as
{\scriptsize
\begin{equation}\label{block_model_mean}
\underbrace{\left(\begin{array}{c}
\bm{y}_{(1: r_1-1)}\\
\bm{y}_{(r_1: r_2-1)}\\
\vdots \\
\bm{y}_{(r_{k_n-1}:r_{k_n}-1)}\\
\end{array}\right)}_{\mathcal{Y}}= 
\underbrace{\left(\begin{array}{c cccc}
\mathbf{1}_{(1: r_1-1)}  & \mathbf{0} &\dots & \mathbf{0}\\
\mathbf{1}_{(r_1: r_2-1)} & \mathbf{1}_{(r_1: r_2-1)} &\dots &\mathbf{0}\\
\vdots & \vdots & \ddots & \vdots \\
\mathbf{1}_{(r_{k_n-1}:r_{k_n}-1)} & \mathbf{1}_{(r_{k_n-1}:r_{k_n}-1)}& \dots  &\mathbf{1}_{(r_{k_n-1}:r_{k_n}-1)}\\
\end{array}\right)}_{\mathcal{X}}
\underbrace{
\begin{pmatrix}
{\bm{\Theta}_{1}}^\prime \\
{\bm{\Theta}_{2}}^\prime\\
\vdots \\
{\bm{\Theta}_{k_n}}^\prime\\
\end{pmatrix}}_{{\Theta}}
+\underbrace{
\left(\begin{array}{c}
{\bm{\zeta}_{(1: r_1-1)}}\\
{\bm{\zeta}_{(r_1: r_2-1)}}\\
\vdots \\
{\bm{\zeta}_{(r_{k_n-1}:r_{k_n}-1)}}\\
\end{array}\right)}_{E},   
\end{equation}}
where $\bm{y}_{(a:b)}:= (\bm{y}_a, \dots, \bm{y}_{b})^\prime$, 
${\bm{\zeta}_{(a:b)}:= (\bm{\zeta}_a, \dots, \bm{\zeta}_{b})^\prime}$; 
where  $\mathbf{1}_{(a:b)} \in \mathbb{R}^{b-a+1}$ is all-ones vector, $\mathcal{Y} \in \R^{n \times p} $, $\mathcal{X} \in \R^{n \times k_n} $, ${\Theta} \in \R^{k_n \times  p}$  and $E \in \R^{n \times p} $. The TFBL algorithm can be applied to this specific model, while following the algorithm described in Section~\ref{sec:algo}, the estimated coefficient parameters is given by 
\begin{equation}\label{est_coef_mean}
{
{\widehat{\bm{\mu}}}_{j} = \sum_{i=1}^{{\left \lfloor\frac{1}{2}\left(\max(J_{j-1})+\min(J_{j})\right)\right\rfloor}} \widehat{\bm{\Theta}}_i, \text{ for } j = 1, \dots, \widetilde{m}^f+1,} 
\end{equation}
and its thresholded variant estimate $\widetilde{\bm{\mu}}_{j}$ can be defined as  
\begin{equation}\label{est_coef_threshold_mean}
{\widetilde{\bm{\mu}}}_{j} = \widehat{\bm{\mu}}_{j}\mathbbm{1}_{ \left\{ \vert \widehat{\bm{\mu}}_{j} \vert > \eta_{n,j}  \right\}}, \text{ for } j = 1, \dots, \widetilde{m}^f+1. 
\end{equation}


To  establish  consistency  properties  of  the  detection/estimation  procedure,  the  following  assumptions  are needed:

\begin{itemize}
 \item[(B1.)] For the $j$-th segment, where $j = 1, 2, \dots, m_0 +1$, the process can be written as $\bm{y}_{j,t} = {\bm{\mu}_{j}^\star} + \bm{\varepsilon}_{j,t}$ where
 the error $\{\bm{\varepsilon}_{j,t}\}$ is a sub-Gaussian random vector with parameter $(\Sigma_{j}, \sigma_{j}^2)$ (see the details of sub-Gaussian definition in Appendix~\ref{sec:subGaussian}). Furthermore,
    {\footnotesize\[1/C_1 \leq  \min_{1\leq j \leq m_0+1}\Lambda_{\min}(\Sigma_{j}) \leq \max_{1\leq j \leq m_0+1}  \Lambda_{\text{max}}(\Sigma_{j}) \leq C_1, \text{ and } 
     1/C_2  < \min_{1\leq j \leq m_0+1} \sigma_{j}^2 \leq \max_{1\leq j \leq m_0+1} \sigma_{j}^2 <C_2,\]}
where $C_1$  and $C_2$ are positive constants.

    \item[(B2.)]The mean vectors  $\bm{\mu}_{j}^\star$  are sparse. More specifically, for all $j = 1,2,\dots, m_0+1$, $d_{j} \ll p$, i.e., $d_{j}/ p = o(1)$. Moreover, there exists a positive constant $M_{\bm{\mu}}  > 0$ such that
    \[\text{max}_{1 \leq j \leq m_0 +1}\lVert \bm{\mu}_{j}^\star \rVert_{\infty} \leq M_{\bm{\mu}} .\]
    
    \item[(B3.)]
    Let $\nu_n = \min_{1\leq j \leq m_0} \left\lVert \bm{\mu}_{j+1}^\star-\bm{\mu}_{j}^\star \right\rVert_2 $.  
     There exists a  positive sequence $b_n$ such that, as $n \rightarrow \infty$,
\[\frac{\min_{1\leq j \leq m_0+1} \vert t_{j} - t_{j-1}\vert}{b_n }\rightarrow +\infty,   
{\ d^\star_n \frac{\log (p\vee n)}{b_n} \rightarrow 0 \text{ and }
\nu_n = \Omega\left(\sqrt{\frac{d^\star_n \log (p\vee n)}{b_n}}\right).}
\]
\item[(B4.)] {The regularization
parameters $\lambda_{1,n}$ and $\lambda_{2,n}$  satisfy
$\lambda_{1,n} = C_1 \sqrt{{ \log (p\vee n)}/ {n}} \sqrt{{b_n} /{n}}$
, and 
$\lambda_{2,n} = C_2\sqrt{{\log (p\vee n)}/{n}}\sqrt{{b_n}/{n} }$ 
for some large constant $C_1, C_2 > 0$.}
\end{itemize}

The Assumption~B1 is a standard assumption in mean shift models and allows one to obtain necessary concentration inequalities needed in high dimensions including restricted eigenvalue condition and deviation bound condition (see \cite{loh2012}). Assumptions~B2-B4 are special cases of Assumptions A3-A5 in Section~\ref{sec:thm}.
The next theorem states the detection and estimation consistency of TBFL in mean shift model.

\begin{theorem}[Results for mean shift model]\label{thm:mean} 
 Suppose the Assumptions B1-B4 hold,  then
 there exists a large enough constant $K>0$ such that  as $n \to +\infty$,

  \[\mathbb{P} \left( \widetilde{m}^f = m_0 , \max_{1\leq j \leq m_0} \left|\widetilde{t}_j - t_j \right| \leq \frac{K {d_n^\star}  {\log (p\vee n)} }{\nu_n^2} \right) \to 1.\]
Also, the solution $\widehat{\bm{\mu}}_j$ from \eqref{est_coef_mean} satisfies 
    \begin{equation*}
    \max_{1\leq j \leq m_0 +1} \left \Vert  \widehat{\bm{\mu}}_j - \bm{\mu}_j^\star \right\Vert_F =  O_p \left(\sqrt{\frac{d^\star_n {{\log( p\vee n)} }}{b_n}}\right).
    \end{equation*}
Further, if {$\eta_{n,j} = {C_j}\sqrt{\frac{{\log (p \vee n)} }{b_n}}$ for some positive constant $C_j$},
the thresholded variant $\widetilde{\bm{\mu}}_j$ from \eqref{est_coef_threshold_mean}  satisfies
\begin{equation*}
   \max_{ 1 \leq j \leq m_0 + 1 }  \left\vert  \text{supp}(\widetilde{\bm{\mu}}_j) \backslash  \text{supp}( {\bm{\mu}}_j^\star) \right\vert = { O_p \left(d^\star_n\right).}
    \end{equation*} 
\end{theorem}

The localization error rate obtained in Theorem~\ref{thm:mean} for mean shift model is superior compared to some competing methods, e.g. the Sparsified Binary Segmentation (SBS) algorithm developed in \cite{cho2015multiple} and the Inspect algorithm \citep{wang2016high}. Note that our rate of consistency for estimating the break point locations is of order {${{d_n^\star} \log (p \vee n)}/{\nu_n^2} $, which could be as low as $ \left(\log (p\vee n)\right)^{1 + v}$  if we set a constant $\nu_n$ and $d_n^\star = \left(\log (p\vee n) \right)^{\nu}$}. \cite{cho2015multiple}  can achieve a similar rate when $\Delta_n$ is of order $n$. However, when $\Delta_n$ is smaller and is of order  $n^{\psi}$ for some $\psi \in (6/7, 1) $ , Cho \& Fryzlewicz’s rate of consistency will be of order $n^{2-2\psi}$, which is larger than our logarithmic rate. Moreover, \cite{wang2016high} proposed a two-stage procedure called ``Inspect'' for estimation of the change points. The Inspect method guarantees the recovery of the correct number of change points with high probability. Translating to our notation, their best localization error is at least of order $ m_0^4 { (\log n + \log p )}/{\nu_n^2}$  (see Theorem 5 in \cite{wang2016high}), where $m_0$ is the number of change points. This rate can be larger than the rate stated in Theorem~\ref{thm:mean}, specially when $m_0$ is large. We also compared the performance of these three methods (TBFL, SBS, and Inspect) numerically, see more details in Section~\ref{sec:comparison}.

\subsection{ Gaussian Graphical Model}

In this section, we consider a Gaussian graphical model with possible changes in its covariance (precision) matrix. Specifically, suppose there exist $m_0$ change points $\left \{t_1 , \dots , t_{m_0} \right\}$ such that $1 =t_0 <t_1 <\dots<t_{m_0} <t_{m_0+1} =n+1$, then  
\begin{equation}\label{eq:ggm_model_1}
\bm{x}_t \sim \sum_{j = 1}^{m_0+1} \mathcal{N}_p \left(\bm{0}, \Sigma_{j}\right) \mathbbm{1}_{\{t_{j-1}\leq t <t_j \}}, \ t =  1, \dots, n, 
\end{equation}
such that observations $\bm{x}_t \in \mathbb{R}^p$ are p-dimensional realizations
of a multivariate normal distribution with zero mean and covariance matrix $\Sigma_{j}$ during the $j$-th segment. Let $\Omega_{j}:= \Sigma_{j}^{-1}$ denote the precision matrix during the $j$-th segment, with elements $\left (\Omega_{j}({l,k}) \right)$, $1 \leq l, k \leq p$. We study the problem of estimating both the change points and the non-zero elements of the precision matrices. 
Setting the parameters  $\bm{x}_t = \bm{y}_t$, $p_x =p_y = p$ in  the model representation  in \eqref{eq:model_1}, the model \eqref{eq:ggm_model_1} can be equivalently expressed as the following regression equation (utilizing the neighborhood selection method developed in \cite{meinshausen2006high}):
\begin{equation}\label{eq:model_ggm}
    \bm{x}_t = \sum_{j = 1}^{m_0+1}\left (\bm{A}_{j}^\star \bm{\bm{x}}_t + \bm{\varepsilon}_{j,t} \right ) \mathbbm{1}_{\{t_{j-1}\leq t <t_j \}}, \ t =  1, \dots, n, 
\end{equation}
where $\bm{x}_t$ is the $p$-vector of observation at time $t$; 
$\bm{A}_{j}^\star \in \R^{p \times p}$ 
is the sparse coefficient matrix with zero diagonal during the $j$-th segment, such that  the off-diagonal elements
${\bm{A}^\star_{j}}({l,-l}) = \Sigma_{j}({l,-l})\left({\Sigma_{j}({-l,-l})} \right)^{-1} = -\left(\Omega_{j}({ l,l}) \right)^{-1}\Omega_{j}({ l,-l}) $,
where $\Sigma({-l,-k})$ is  the sub-matrix of $\Sigma$ with its $l$-th row and $k$-th column removed; $\Sigma({l,k})$ is  the entry of matrix $\Sigma$ that lies in the $l$-th row and $k$-th column; $\bm{\varepsilon}_{j,t}$ is a multivariate Gaussian white noise, such that $\bm{\varepsilon}_{j,t} \sim \mathcal{N}\left(0, \left(I_p - \bm{A}_{j}^\star\right) \Sigma_{j} \left(I_p - \bm{A}_{j}^\star\right)^\prime \right)$ 
where the variance of the $l$-th component in the error term $\text{Var} (\varepsilon_{j, t}(l)  ) = \Sigma_j({l,l}) -  \Sigma_j({l,-l}) \left(\Sigma_j({-l,-l}) \right)^{-1} \Sigma_j({-l,l})  $. Therefore, we have 
$\Omega_j({l,l}) =  \left(\text{Var} (\varepsilon_{j,t}(l) )\right)^{-1}, \text{ and }
    \Omega_j({l,-l}) =  -\left(\text{Var} (\varepsilon_{j,t}(l) )\right)^{-1}   { \bm{A}_{j}^\star}({l,-l}),
$
where $l=1,\dots ,p$, $j=1,\dots,m_0+1$, $t =1,  \dots, n$.
 The sparsity in the entries of $\Omega_{j}$ can be matched into sparsity in regression coefficient matrix
$\bm{A}_{j}^\star$'s.

Define $\bm{\Theta}_1 = \bm{A}_{j}^\star$,  $\bm{\Theta}_i =\bm{A}_{j+1}^\star-\bm{A}_{j}^\star$ when $t_j \in [r_{i-1}, r_{i})\  \text{for some}\ j$ and  $\bm{\Theta}_i = \bm{0}$ otherwise, 
for $i = 2,3,\dots, k_n$.
In this case, the linear regression model in terms of $\Theta$ can be written as
{\scriptsize
\begin{equation}\label{block_model_ggm}
\underbrace{\left(\begin{array}{c}
\bm{x}_{(1: r_1-1)}\\
\bm{x}_{(r_1: r_2-1)}\\
\vdots \\
\bm{x}_{(r_{k_n-1}:r_{k_n}-1)}\\
\end{array}\right)}_{\mathcal{Y}}= 
\underbrace{\left(\begin{array}{c cccc}
\bm{x}_{(1: r_1-1)}  & \mathbf{0} &\dots & \mathbf{0}\\
\bm{x}_{(r_1: r_2-1)} & \bm{x}_{(r_1: r_2-1)} &\dots &\mathbf{0}\\
\vdots & \vdots & \ddots & \vdots \\
\bm{x}_{(r_{k_n-1}:r_{k_n}-1)} & \bm{x}_{(r_{k_n-1}:r_{k_n}-1)}& \dots  &\bm{x}_{(r_{k_n-1}:r_{k_n}-1)}\\
\end{array}\right)}_{\mathcal{X}}
\underbrace{
\begin{pmatrix}
{\bm{\Theta}_{1}}^\prime \\
{\bm{\Theta}_{2}}^\prime\\
\vdots \\
{\bm{\Theta}_{k_n}}^\prime\\
\end{pmatrix}}_{{\Theta}}
+\underbrace{
\left(\begin{array}{c}
{\bm{\zeta}_{(1: r_1-1)}}\\
{\bm{\zeta}_{(r_1: r_2-1)}}\\
\vdots \\
{\bm{\zeta}_{(r_{k_n-1}:r_{k_n}-1)}}\\
\end{array}\right)}_{E},   
\end{equation}}
where  $\bm{x}_{(a:b)}:= (\bm{x}_a, \dots, \bm{x}_{b})^\prime$,
${\bm{\zeta}_{(a:b)}:= (\bm{\zeta}_a, \dots, \bm{\zeta}_{b})^\prime}$; 
$\mathcal{Y} \in \R^{n \times p} $, $\mathcal{X} \in \R^{n \times k_n p} $, ${\Theta} \in \R^{k_np \times  p}$ and $E \in \R^{n \times p} $. The TBFL algorithm can be applied to detect change points while the estimated coefficient parameters is given by 
\begin{equation}\label{est_coef_ggm}
{\widehat{\bm{A}}}_{j} = \sum_{i=1}^{{\left \lfloor\frac{1}{2}\left(\max(J_{j-1})+\min(J_{j})\right) \right \rfloor}} \widehat{\bm{\Theta}}_i, \text{ for } j= 1, \dots, \widetilde{m}^f+1,
\end{equation}
and its thresholded variant estimate $\widetilde{\bm{A}}_{i}$ as  
\begin{equation}\label{est_coef_threshold_ggm}
{\widetilde{\bm{A}}}_{j} = \widehat{\bm{A}}_{j}\mathbbm{1}_{ \left\{ \vert \widehat{\bm{A}}_{j} \vert > \eta_{n,j}  \right\}}, \text{ for } j = 1, \dots, \widetilde{m}^f+1. 
\end{equation}


To establish consistency properties of the detection/estimation procedure, the following assumptions are needed:
\begin{itemize}
\item[(D1.)] For each $j = 1, 2, \dots, m_0 +1$, the process follows the \eqref{eq:model_ggm} such that
$\bm{x}_{j,t} \sim \mathcal{N} (\bm{0}, \Sigma_{j})$ and $\bm{\varepsilon}_{j,t} \sim \mathcal{N}\left(0, \left(I_p - \bm{A}_{j}^\star\right) \Sigma_{j} \left(I_p - \bm{A}_{j}^\star\right)^\prime \right)$. 
Further,
{\footnotesize\begin{align*}
    1/C_1 \leq  \min_{1\leq j \leq m_0+1}\Lambda_{\min}(\Sigma_{j}) \leq \max_{1\leq j \leq m_0+1}  \Lambda_{\text{max}}(\Sigma_{j}) \leq C_1, \text{ and } 1/C_2 \leq  \min_{1\leq j \leq m_0+1, 1\leq l \leq p} (\Omega_j(l,l))^{-1},
\end{align*}}
where $C_1$  and $C_2$ are positive constants.

\item[(D2.)]The coefficient vectors  $\bm{A}_{j}^\star$  are sparse. More specifically, for all $j = 1,2,\dots, m_0+1$, $d_{j} \ll p^2$, i.e., $d_{j}/ p^2 = o(1)$. Moreover, there exists a positive constant $M_{\bm{A}}  > 0$ such that
    \[\text{max}_{1 \leq j \leq m_0 +1}\lVert \bm{A}_{j}^\star \rVert_{\infty} \leq M_{\bm{A}} .\]
    
    \item[(D3.)]
    Let $\nu_n = \min_{1\leq j \leq m_0} \lVert \bm{A}_{j+1}^\star-\bm{A}_{j}^\star \rVert_F $.  
     There exists a  positive sequence $b_n$ such that, as $n \rightarrow \infty$,
\[\frac{\min_{1\leq j \leq m_0+1} \vert t_{j} - t_{j-1}\vert}{b_n }\rightarrow +\infty,   
\ {d^\star_n \frac{\log (p\vee n )}{b_n} \rightarrow 0 \text{ and }
\nu_n = \Omega\left(\sqrt{\frac{d^\star_n\log ( p\vee n )}{b_n}}\right).}
\]

\item[(D4.)] The regularization
parameters $\lambda_{1,n}$ and $\lambda_{2,n}$  satisfy
{$\lambda_{1,n} = C_1 \sqrt{{ \log( p\vee n)}/{n}} \sqrt{{b_n}/{n}}$
, and 
$\lambda_{2,n} = C_2\sqrt{{\log( p \vee n) }/{n}}\sqrt{{b_n}/{n} }$ }
for some large constant $C_1, C_2 > 0$.

\end{itemize}

One can exclude singular or nearly singular covariance matrices based on Assumption D1, thus guaranteeing the uniqueness of $\bm{\Theta}$ \citep{wang2016precision, meinshausen2006high}. 
The RE condition (A1) and deviation bound (A2) holds under Assumption D1 (see more details in  Section 4 in  \cite{bickel2009simultaneous} and Lemma 12 in \cite{zhou2011high}).
Assumptions~D2-D4 are special cases of Assumptions A3-A5 in Section~\ref{sec:thm}.

The next theorem is about the detection and estimation consistency of TBFL applied to Gaussian graphical model with breaks.

\begin{theorem}[Results for Gaussian graphical model]\label{thm:ggm} 
   Suppose the Assumptions D1-D4 hold. Then
   there exists a large enough constant $K>0$ such that, 
   as $n \to +\infty$,
   \[\mathbb{P} \left( \widetilde{m}^f = m_0, \max_{1\leq j \leq m_0} \left|\widetilde{t}_j - t_j \right| \leq \frac{K {d_n^\star} {\log (p\vee n) }}{\nu_n^2} \right) \to 1.\]
Also, the solution $\widehat{\bm{A}}_j$ from \eqref{est_coef_ggm} satisfies 
    \begin{equation*}
    \max_{1 \leq j \leq m_0 + 1 }\left\Vert  \widehat{\bm{A}}_j - \bm{A}_j^\star \right\Vert_F =  O_p \left(\sqrt{\frac{d^\star_n{\log( p \vee n)} }{b_n}}\right).
    \end{equation*}
Further, 
if {$\eta_{n,j} = {C_j}\sqrt{\frac{{\log (p  \vee n)} }{b_n}}$ for some positive constant $C_j$},
the thresholded variant  $\widetilde{\bm{A}}_j$ from \eqref{est_coef_threshold_ggm}  satisfies
\begin{equation*}
   \max_{ 1 \leq j \leq m_0 + 1 }  \left\vert  \text{supp}\left(\widetilde{\bm{A}}_j \right) \backslash  \text{supp} \left( {\bm{A}}_j^\star\right) \right\vert = { O_p \left(d^\star_n\right).}
    \end{equation*} 
\end{theorem}

The localization error stated in Theorem~\ref{thm:ggm} is optimal up to a logarithmic factor. This rate is an improvement over the consistency rate of Group-Fused Graphical Lasso (GFGL) method developed in \cite{gibberd2017multiple} in which the localization error is of order $O(p^2 \log p/v_n^2)$  (as shown in Theorem 3.2 in \cite{gibberd2017multiple}). Moreover, TBFL achieves a better consistency rate compared with the $O(p \log n/v_n^2)$ localization error rate established in \cite{kolar2012estimating}. Finally, it achieves a similar consistency rate in terms of the localization error compared with the method in \cite{bybee2018change} for a single change point while there is no theoretical results for consistency of number of change points in the detection method developed in \cite{bybee2018change}. We investigated a numerical comparison between TBFL and the method developed in \cite{bybee2018change} in which TBFL outperforms the latter method both in terms of estimated number of change points and their location accuracy, see more details on this numerical comparison in Section~\ref{sec:sa} of the supplementary materials.

\section{Optimal Block Size Selection}\label{sec:optimalblock}

In this section, we develop a data-driven method to select the optimal block size. 
If the true number of change points $m_0$ is relatively small, the proposed TBFL algorithm is robust to changes in the block size $b_n$ (see more details in Section~\ref{Sim_Scenarios}), but for a large $m_0$, we propose to select the optimal block size by minimizing the High-dimensional Bayesian Information Criterion (HBIC) developed in \cite{wang2011consistent}  over a grid search domain. Specifically, we select the optimal $b_n$ as
\begin{align*}
    \widehat{b}_n = \argmin_{b_n \in S} \mbox{HBIC} (b_n) = \argmin_{b_n \in S} \left( n \log \left (\frac{1}{n}\mbox{RSS}(b_n)\right) + 2 \gamma \log(p_xp_y) \left\vert M(b_n) \right\vert \right)
\end{align*}
where 
$ \mbox{RSS}(b_n)= \sum_{j = 1}^{\widetilde{m}^f(b_n) + 1}\sum_{t= \widetilde{t}_{j-1}^f(b_n)}^{\widetilde{t}_{j}^f(b_n)-1} \left\| \bm{y}_{t}  - \widetilde{\bm{B}}_{j}(b_n)\bm{x}_{t} \right\|_2^2 $
is the residual sum of squares, $\widetilde{\bm{B}}_{j}(b_n)$, $\widetilde{m}^f(b_n)$ and $\widetilde{t}_{j}^f(b_n)$ are estimated parameters, number of change points and location of change points using block size $b_n$;
$|M(b_n)|= \sum_{j=1}^{\widetilde{m}^f(b_n) + 1} \widetilde{d}_j(b_n)$, where $\widetilde{d}_j(b_n)$ is the number of non-zero elements in the coefficient parameter $\widetilde{B}_j$ in \eqref{est_coef_threshold} while using the block size $b_n$. 
We follow \cite{wang2011consistent}'s suggestion for $\gamma$ selection.
Note that the detection and estimation results are robust with respect to changes in $\gamma$ as investigated in Section~\ref{sec:add_sim} in  supplementary materials. 
The details for selection of the search domain $S$ are provided in Section~\ref{subsec:domain} in supplementary materials.  

\section{Numerical Performance Evaluation}\label{sec:comparison}

In this section, we compare the empirical performance of our method (TBFL) with selected competing methods. For mean shift model, we compare our method with 
SBS \citep{cho2015multiple} and Inspect \citep{wang2016high}.
For Gaussian graphical model, we compare our method with the Simulated Annealing (SA) algorithm \citep{bybee2018change}.
We also evaluate the performance of the TBFL method with respect to both structural break detection and parameter estimation over several simulation scenarios. Due to space limitations, in this section, we only provide details of comparisons with SBS and Inspect, while details on comparison with the method developed in \citep{bybee2018change} and details of empirical performance of TBFL over several simulation scenarios are provided in supplementary material, Sections~\ref{sec:sa} and \ref{Sim_Scenarios}, respectively.

\begin{figure}[!ht]
\begin{center}
\includegraphics[width=0.24\linewidth, clip=TRUE, trim=0mm 0mm 00mm 10mm]{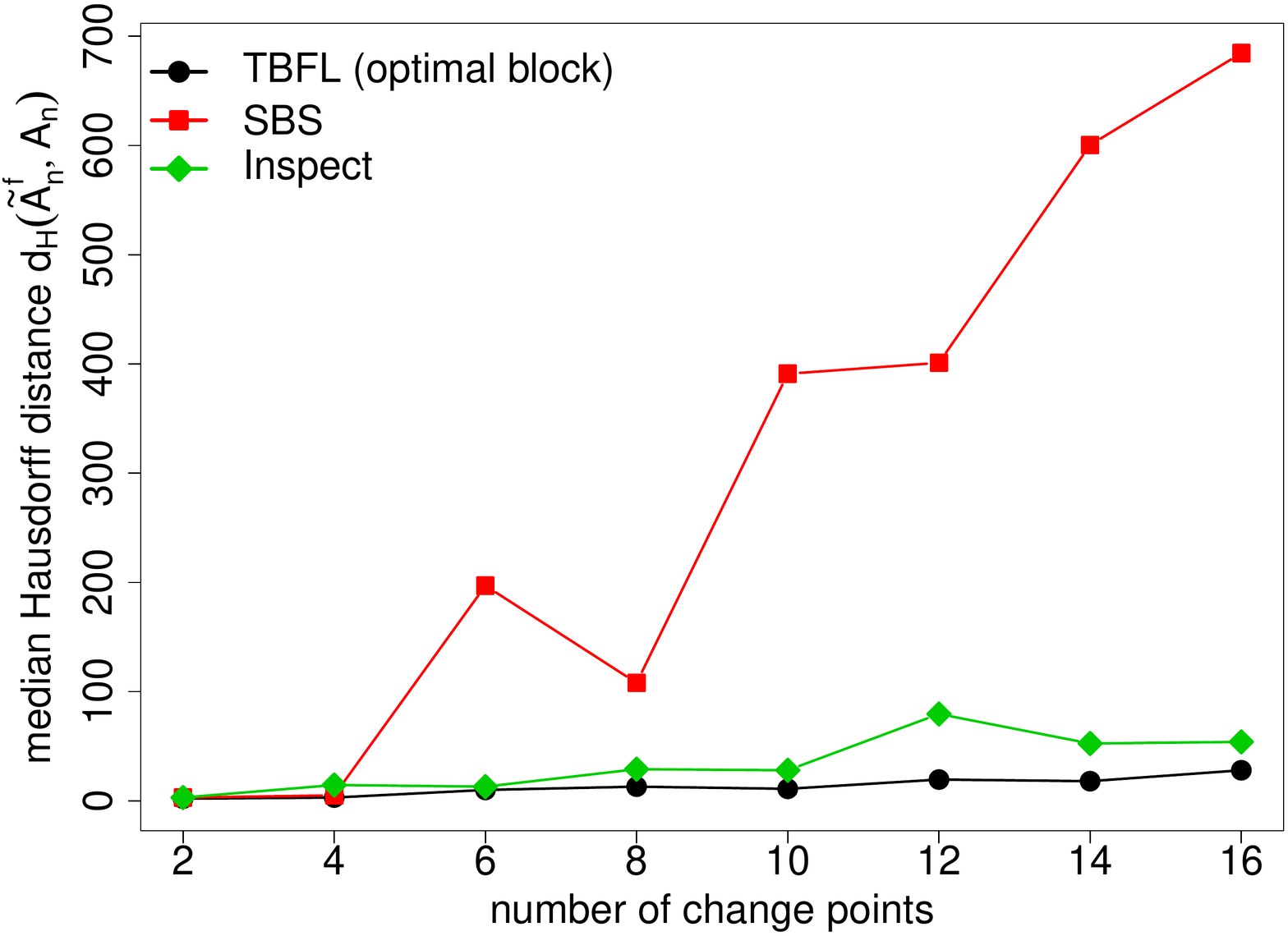}
\includegraphics[width=0.24\linewidth, clip=TRUE, trim=0mm 0mm 00mm 10mm]{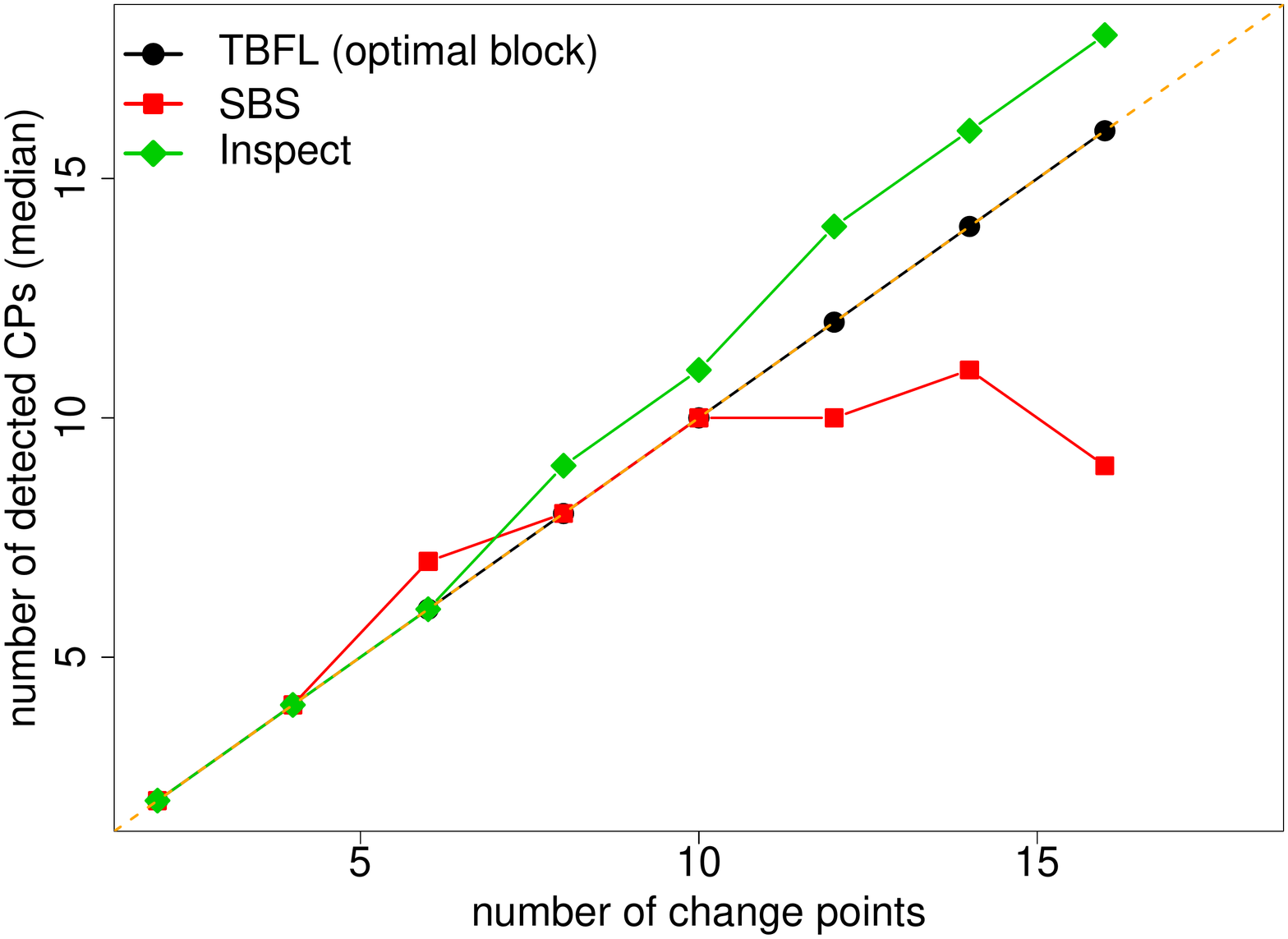}
\includegraphics[width=0.24\linewidth, clip=TRUE, trim=0mm 0mm 00mm 10mm]{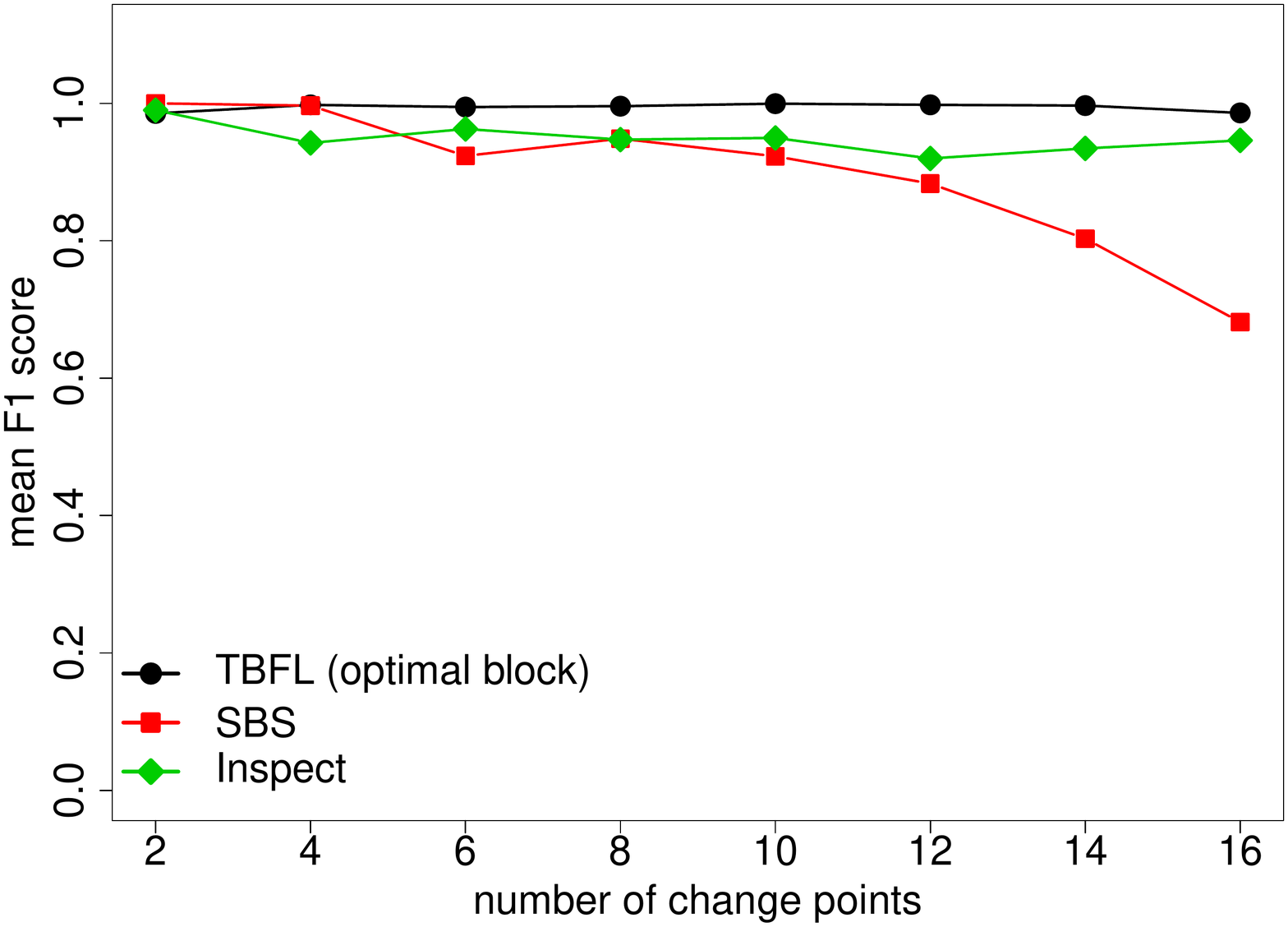}
\caption{
(a)  Hausdorff distance $d_H \left(  \widetilde{\mathcal{A}}_n^f, \mathcal{A}_n \right)$ for the TBFL, SBS and Inspect methods;
(b) median of number of detected change points for three methods; (c) F1 score.}
\label{fig:comp_performance_mean}
\end{center}
\end{figure}

Before describing the simulation settings, we need to list/define certain measures to compare detection performance among competing methods. First, the Hausdorff distance $d_H \left(  \widetilde{\mathcal{A}}_n^f, \mathcal{A}_n \right)$ is used as the measure for estimation accuracy of the location of break points. Moreover, following \cite{hushchyn2020online}, we define a set of correctly detected change-points as True Positive Change Points (TPCP): 
\[\text{TPCP} =   \left \{t_j  \Big| \exists \widetilde{t}_{j'}^f \text{ such that }\widetilde{t}_{j'}^f \in \left [ t_j - \frac{t_{j} - t_{j-1}}{5}, t_j + \frac{t_{j+1} - t_{j}}{5}  \right] , j = 1, \dots, m_0 \right\}. \]
Further, the Precision, Recall and F1-score are calculated as follows:
\begin{align*}
\mbox{Precision} = \frac{\vert \text{TPCP} \vert }{\widetilde{m}^f}, \quad \mbox{Recall} =\frac{\vert \text{TPCP} \vert }{m_0}, \quad \mbox{F1} = \frac{2 \cdot \mbox{Precision} \cdot \mbox{Recall}}{ \mbox{Precision} + \mbox{Recall}  },
\end{align*}
where $\vert \text{TPCP} \vert$ is the cardinality of set $\text{TPCP}$. The highest possible value of an F1 score is 1, indicating perfect precision and recall, and the lowest possible value is 0, if either the precision or the recall is zero. We select F1 score as another quantitative measurement to evaluate detection performances. Next, details of the simulation setting are explained.


\textit{ Setting A (Mean Shift Model)}. In the setting A, $n=5000$, $p=20$,  with the number of non-zero elements in $j$th segments $d_j= 2$, for all $j= 1, \dots, m_0 + 1$. 
The mean coefficient $\mu$ are chosen to be multivariate with  random sparse structure and random entries sampled from $\mbox{Uniform}(-1,-0.5)\mathbbm{1}_{\{j \text{ is odd}\}} + \mbox{Uniform}(0.5,1)\mathbbm{1}_{\{j \text{ is even}\}}$, for each $j = 1, \dots, m_0 + 1$.  We consider different setting of $m_0$ starting from 2 to 16.

\begin{table}[!ht]
\caption{\label{table_sim_E}Results of difference between $\widetilde{m}^f$ and $m_0$ for TBFL, SBS and Inspect methods in simulation scenario A. }
\tiny
\centering
\begin{tabular}{lccccccccc} 
  \hline
   method&$\left\vert \widetilde{m}^f -m_0\right\vert$ &  $m_0= 2$ & $m_0= 4$& $m_0= 6$& $m_0= 8$& $m_0= 10$& $m_0= 12$& $m_0= 14$& $m_0= 16$ \\
   \hline
   \multirow{4}{*}{TBFL} & $0$ & 94& \textbf{98}& \textbf{93}& \textbf{93}& \textbf{99} & \textbf{97}& \textbf{95} & \textbf{74}\\
   & $1$& 4& 2 & 7 &7 & 1 & 3 & 5 & 16\\
   & $2$& 2 & 0 & 0 & 0 &0 & 0& 0 & 6\\
   & $>2$& 0 & 0 &0 &0 &0 &0 &0 & 4\\
   \hline
   \multirow{4}{*}{SBS} & $0$ &\textbf{100}& 96& 30& 57 & 45& 8 & 1 & 0\\
   & $1$& 0 & 3 & 60 & 40 & 28 & 25& 11 & 2\\
   & $2$& 0 & 0 & 10 & 3 & 26 & 40 & 18 & 1\\
   & $>2$&0 & 1 & 0 & 0 & 1 & 27 & 70 & 97\\
   \hline
   \multirow{4}{*}{Inspect} & $0$ &95& 56 & 64& 38& 31 & 3 & 7 & 11\\
   & $1$& 5& 35& 24 & 37 & 41 & 28 & 22 & 32\\
   & $2$& 0 & 8 &9 & 19 & 18 & 36 & 42 & 30\\
   & $>2$&0 & 1 & 3 & 6 & 10 & 33 & 29 & 27\\
  \hline
  \end{tabular}
  \end{table}    


The detection results of three methods TBFL, SBS, and Inspect are summarized in Figure~\ref{fig:comp_performance_mean} and Table~\ref{table_sim_E}. As shown in Figure~\ref{fig:comp_performance_mean} (left panel), the Hausdorff distance between the set of estimated change points and true change points increases significantly for the SBS method when $m_0$ increases while TBFL and Inspect seem to be more stable. In the middle panel, median of number of detected change points is plotted for all three methods. It can be seen from this plot that Inspect (SBS) over-estimates (under-estimates) the true number of change points while TBFL correctly identifies $m_0$. Further, the right panel of Figure~\ref{fig:comp_performance_mean} depicts the F1 score in which it can be seen that for small $m_0$, all models perform reasonably well while TBFL outperforms SBS and Inspect for larger $m_0$. Overall, TBFL performs better than these two competing methods both in terms of estimating the number of change points and their locations. Finally, as shown in Table~\ref{table_sim_E}, among 100 replicates, our method can correctly estimate $m_0$ over 90\% replicates when $m_0 = 2$ to 14, while the SBS (Inspect) tends to underestimate (overestimate) the $m_0$ starting from $m_0 = 6$. Note that for $m_0=16$, TBFL only selects the true number of change points in $74\%$ replicates which implies that with model specifications in this simulation setting, TBFL has reached its detection limit.

\section{An Application to Electroencephalogram (EEG) Data}\label{sec:eeg}




In this section, 
we apply TBFL method and SA method \citep{bybee2018change} to an EEG data set analyzed in \cite{ANS9Q1_2019}. 
In this database, EEG signals from active electrodes for 72 channels are recorded at a sampling frequency of 256Hz,
 for a total of $\sim 3$ min. 
 The stimulus procedure tested on the selected subject comprised of three 1-min duration interleaved sessions with eyes open and closed.
To speed up the computations, we construct a subset of the EEG data observation by selecting one in every $16$ record. After de-trending and scaling the data, the final total time points is reduced to $n= 2,922$. The data is also pre-processed to remove the temporal structure pattern (more details are provided in supplementary material~\ref{sec:eeg_more}).

\begin{figure}[!ht]
\begin{center}
\includegraphics[width=0.32\linewidth, clip=TRUE, trim=0mm 0mm 00mm 0mm]{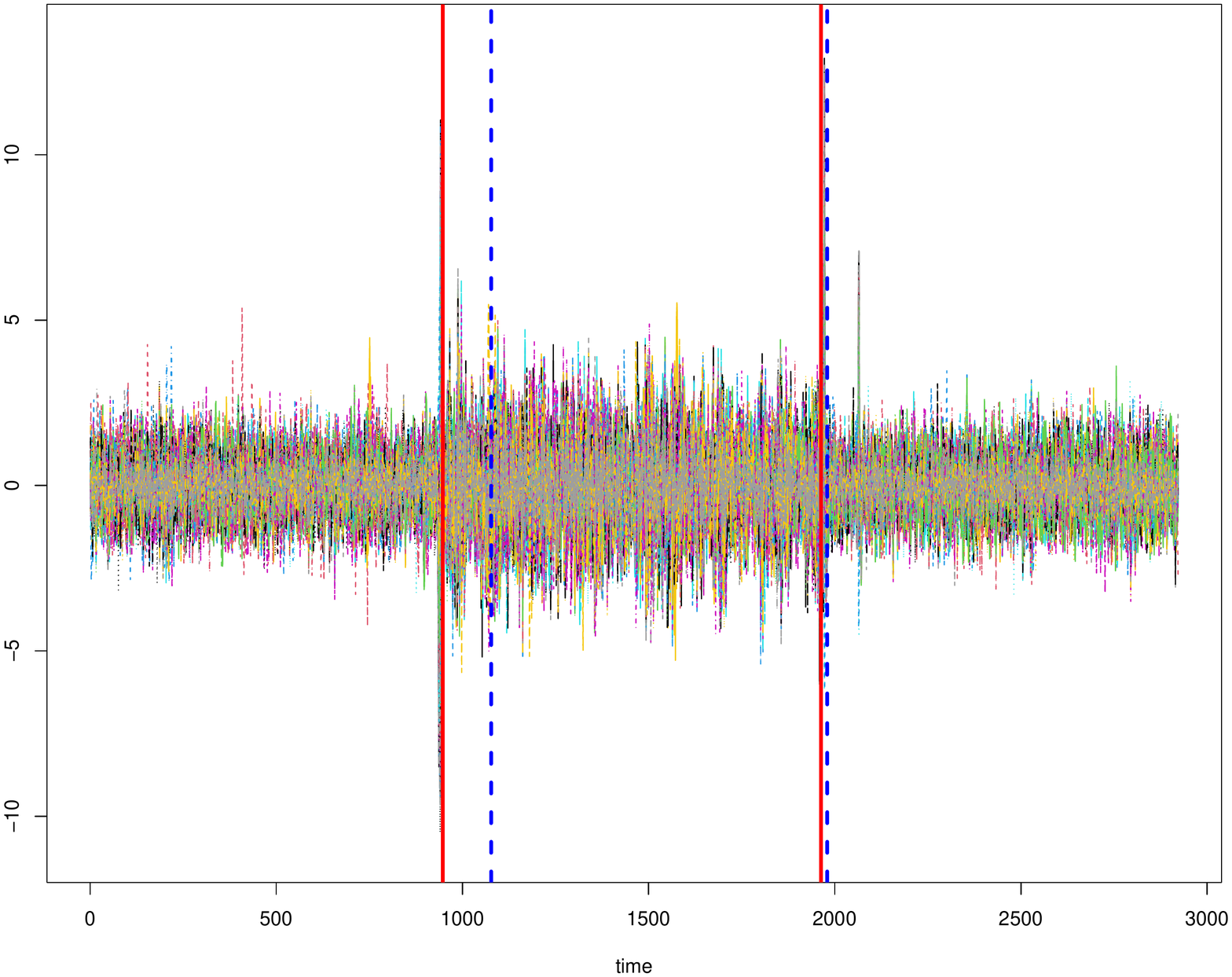}
\includegraphics[width=0.32\linewidth, clip=TRUE, trim=0mm 0mm 00mm 10mm]{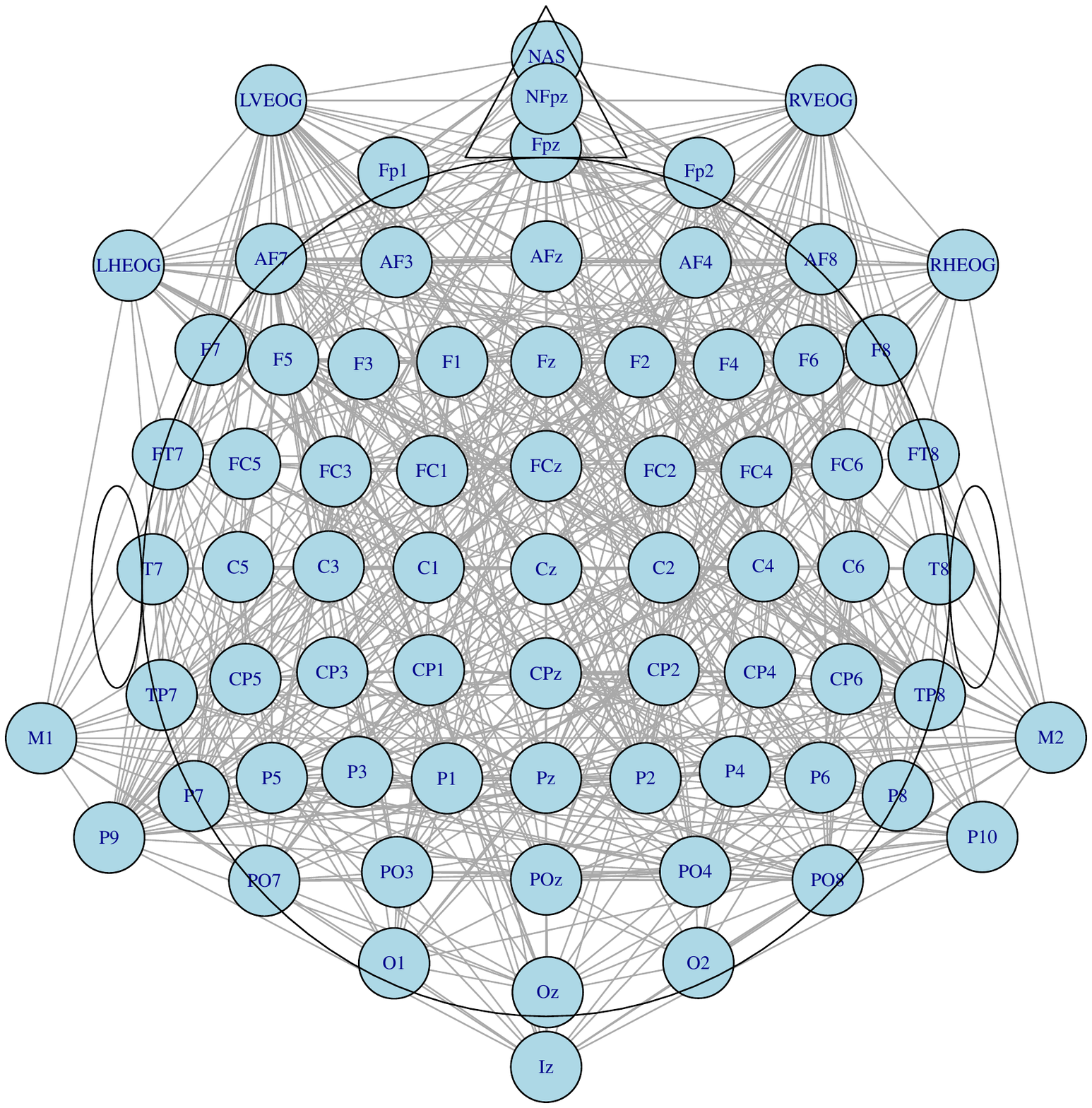}
\includegraphics[width=0.32\linewidth, clip=TRUE, trim=0mm 0mm 00mm 10mm]{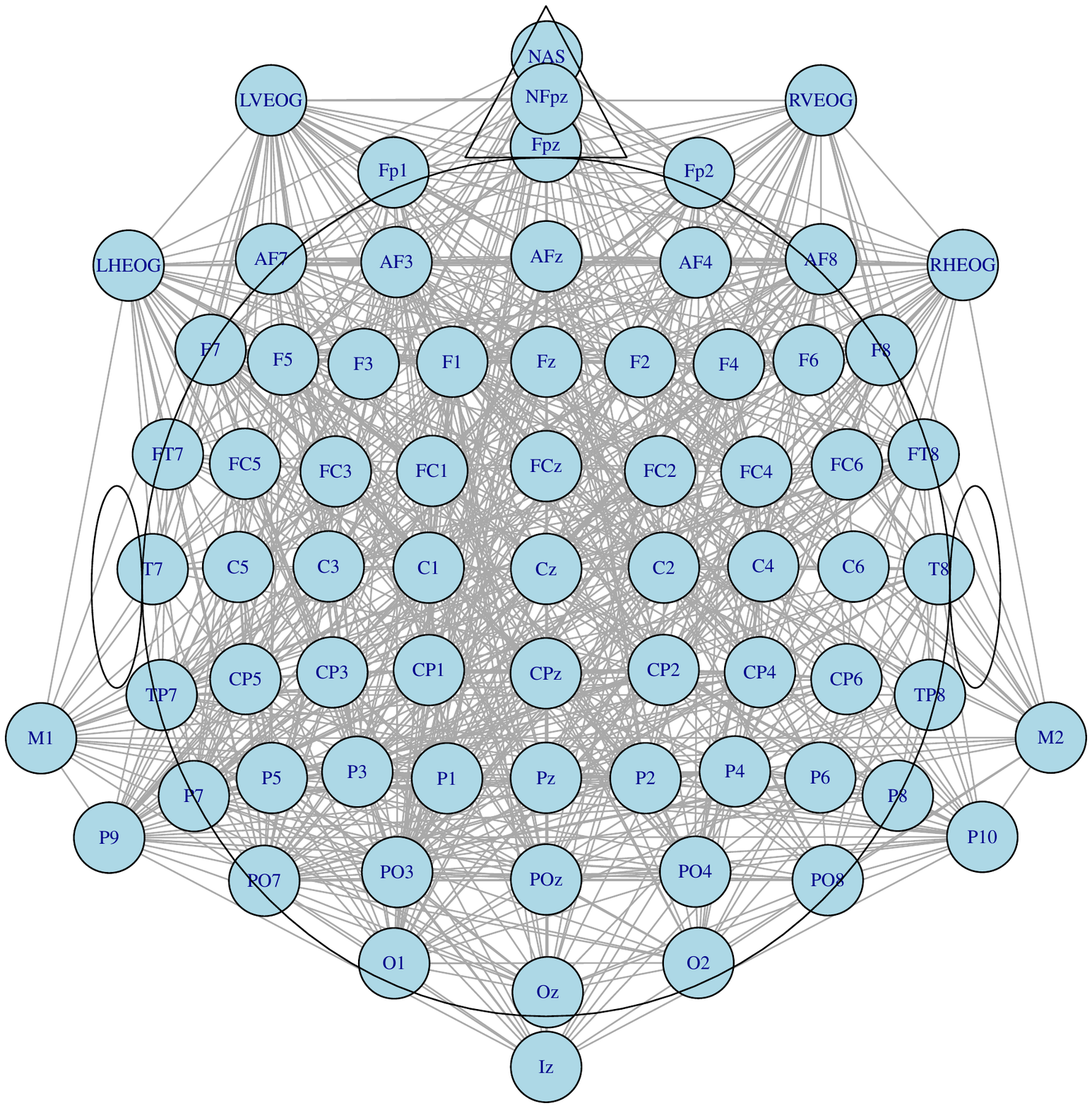}
\caption{(left) The de-trended EEG Data  with 72 channels. The blue dashed line indicates the detected change point while the red vertical line indicates the true change point. The first and third part correspond to eyes open status while the second part corresponds to eyes closed status; (middle and right)  Directed graph of EEG channels connectivity before (middle panel) and after (right panel) the first change point.  
}
\label{fig:eeg_3_truth}
\end{center}
\end{figure}

{
{We considered the Gaussian graphical model with breaks for this data set and applied TBFL with optimal block size procedure under search domain $b_n = 80, 90, 100, 110, 120$ to detect change points and estimate model parameters. The selected optimal block size  is $b_n= 90$. 
As shown in left panel of Figure~\ref{fig:eeg_3_truth}, our method detects two break points at 
$\widetilde{t}_1^f = 1077$ and $\widetilde{t}_2^f = 1980$,
which are close to the open-eye and closed-eye times identified by neurologists ($t_1 = 947$ and $t_2 = 1963$).} 
We also applied the Simulated Annealing (SA) method \citep{bybee2018change} on this EEG dataset. The SA method only detects one change point close to boundaries ($30$) 
 for which there are no recorded stimuli, but no estimated change points close to true change points.
 To demonstrate the changes between eye-open and eye-closed segments, we focus on the first two segments and estimated the model parameters in both segments using the thresholded estimator defined in \eqref{est_coef_threshold_ggm}, i.e. ${\widetilde{\bm{A}}}_{1}$ (segment 1, open-eye) and ${\widetilde{\bm{A}}}_{2}$ (segment 2, closed-eye). Network edges corresponding to non-zero coefficients in these two estimated parameters are depicted in middle and right panel of Figure~\ref{fig:eeg_3_truth}. It can be seen that during the second segment (which is the eyes-closed state), the overall network connectivity has increased. 
Specifically, the total number of edges in the eyes-open (EO) state is 724, while the total number of edges in the  eyes-closed (EC) state is 857.
Among channels which had the most connectivity changes, i.e. their degree (number of edges) between the two segments have changed the most, there are 6 EEG channels PO4, POz, PO3,  Pz,  P3 and CP2, which are located in the visual cortex in the brain \citep{EEGvisual}.
This result confirms the satisfactory variable selection performance of model parameter estimation as stated in Theorem~\ref{thm:estimation}, after detecting break points in the TBFL procedure. Such estimations can produce insights to scientists to study further the channels which have been affected the most by the stimulus procedure. 
}

\section{Concluding Remarks}\label{sec:conclustion}

In this paper, we introduced a novel unified framework that can consistently identify structural breaks and estimate model parameters for  general sparse multivariate linear models with high-dimensional covariates.
We developed a regularized estimation procedure to simultaneously detect the structural break points, and estimate the model parameters. 
Key technical developments include the calibration of the block size and the introduction of hard-thresholding for screening out redundant candidate change points. 
Note that our method could also handle Vector auto-regressive (VAR) model.
Extension of the current framework to nonlinear model constitutes an interesting future research direction.

\appendix

\section{Details about Sub-Gaussian}\label{sec:subGaussian}

We first introduce the definitions of sub-Gaussian random variable and sub-Gaussian random vector similar to that in \cite{vershynin2010introduction} and \cite{loh2012}. 
A random variable $X$ that satisfies
$(\mathbb{E} \vert X \vert^p)^{1/p} \leq K \sqrt{p}$ for all $p\geq 1$
is called a sub-Gaussian random variable with parameter $\Vert X \Vert_{\psi_2} := \sup_{p \geq 1} p^{-1/2}(\mathbb{E} \vert X \vert^p)^{1/p}  $. 
A random vector  $\bm{X} \in \mathbb{R}^{p}$ is said to be sub-Gaussian with parameters $(\Sigma,  \sigma^2)$ if:
\begin{itemize}
\item $\bm{X} \in \mathbb{R}^{p}$ is generated from a zero-mean distribution with covariance $\Sigma$;
\item for any unit vector $\bm{u} \in \mathbb{R}^{p}$, the random variable $\bm{u}^\prime \bm{X}$ is sub-Gaussian with parameter at most $\sigma$.
\end{itemize}

\section{Details about Algorithms}\label{sec:algo_detail}
\setcounter{equation}{0}
We introduce the following detailed algorithms (see Algorithms \ref{algo:1} and \ref{algo:2}) that correspond to each step in the three-step strategy outlined in the main paper. Details about specific data-driven procedures are also provided.
\begin{algorithm}[!ht]
    \DontPrintSemicolon
    \KwInput{The time series data $\{(\bm{x}_t, \bm{y}_t)\}, t = 1, 2, \dots,n$; 
    validation index set $\mathcal{T}$; sets of parameters $\Lambda_1 = \{\lambda_{1,1}, \dots, \lambda_{1,L_1} \}$ and $\Lambda_2 = \{\lambda_{2,1}, \dots, \lambda_{2,L_2} \}$ ; 
    }
    
    Split the time series data $\{(\bm{x}_t, \bm{y}_t)\},\ t = 0,1,\dots,n$ into training data and validation data; 
    
    
    \For{$l_2 = 1, \dots, L_2$}{
    \For{$l_1 = 1, \dots, L_1$}
    {   Set the  the index of estimates $\widehat{\bm{\theta}}$ as   $l = (l_2-1) L_1 + l_1 $.
    
        \KwInit{ If $l_1 = 1$, initially set $\widehat{\bm{\theta}}^{(l)} = \bm{0}$; otherwise, initially set $\widehat{\bm{\theta}}^{(l)} = \widehat{\bm{\theta}}^{(l-1)}$. }
            
        Estimate the sparse components by using the training data $\{(\bm{x}_t, \bm{y}_t)\},\ t \notin \mathcal{T}$:
        \begin{align*}
            &\widehat{\bm{\theta}}^{(l)} := \argmin_{\bm{\theta} \in \R^{\pi_n}} \left\{\frac{1}{n- \vert \mathcal{T} \vert} \left\| \mathbf{y} -\mathbf{Z}\bm{\theta} \right\|_2^2 + \lambda_{1,l_1} \left\|\bm{\theta} \right\|_1 + \lambda_{2,l_2}\sum_{k=1}^{k_n} \left \Vert\sum_{j=1}^k\bm{\Theta}_j \right\Vert_1  \right\} ,
        \end{align*}
        Predict the time series in validation dataset  $\{(\bm{x}_t, \bm{y}_t)\},\ t \in \mathcal{T}$ using the estimated parameter $\widehat{\bm{\theta}}^{(l)}$ and compute the mean squared prediction error (MSPE):
        \[ \text{MSPE}\left(\lambda_{1,l_1}, \lambda_{2,l_2}, \widehat{\bm{\theta}}^{(l)} \right)  =\frac{1}{\vert \mathcal{T} \vert} \left\| \mathbf{y} -\mathbf{Z}\widehat{\bm{\theta}}^{(l)} \right \|_2^2 \]
    }
    }
    Choose the values of  $\lambda_{1,l_1}$, $\lambda_{2,l_2}$ and estimated coefficient $\widehat{\bm{\theta}}^{(l)}$  which minimizes the mean squared prediction error (MSPE) over $\mathcal{T}$, denoted by $\lambda_1^\star$, $\lambda_2^\star$ and  $\widehat{\bm{\theta}}$:
    \[\left(\lambda_1^\star, \lambda_2^\star, \widehat{\bm{\theta}} \right):= \argmin_{ (\lambda_{1,l_2}, \lambda_{2,l_2}) \in \Lambda_1 \times \Lambda_2, \widehat{\bm{\theta}}^{(l)}\in \left\{\widehat{\bm{\theta}}^{(1)}, \dots, \widehat{\bm{\theta}}^{ (L_1L_2)} \right\} } \text{MSPE}\left(\lambda_{1,l_1}, \lambda_{2,l_2}, \widehat{\bm{\theta}}^{(l)} \right).      \]

    \textbf{Hard-Thresholding}:  Denote the indices of candidate blocks  as $\widetilde{I}_n = \left \{i :\left\|\widehat{\bm{\Theta}}_i \right\|_F^2 > c,\  i = 2,\dots, k_n \right\}$ and the set of  candidate change points as 
    $\widetilde{A}_n = \left\{\widetilde{t}_1 ,\dots \widetilde{t}_{\widetilde{m} } \right\} = \left \{r_{i-1} : i \in  \widetilde{I}_n\right\}$
    .

    \KwOutput{$\widetilde{A}_n$,  $\widetilde{I}_n$ and the estimated parameters  $\widehat{\bm{\Theta}}(k_n) \in \R^{k_np_x \times  p_y}$. }
    \caption{\textbf{Block Fused Lasso with Hard-thresholding Procedure}}
    \label{algo:1}
\end{algorithm}

\begin{algorithm}[!ht]
    \DontPrintSemicolon
    \KwInput{The time series data $\{(\bm{x}_t, \bm{y}_t)\}, t = 1, 2, \dots,n$; the candidate change points among blocks $\widetilde{\mathcal{A}}_n = \left\{\widetilde{t}_1 ,\dots \widetilde{t}_{\widetilde{m} }\right\}  $; the estimated parameter matrix $\widehat{\bm{\Theta}}_j$ in each segment, $j = 1, 2, \ldots, \widehat{m} $ after the first step algorithm.
    }
    
    \textbf{Block clustering}:  Partition the $\widetilde{m}$ candidate change points $\widetilde{\mathcal{A}}_n = \left\{\widetilde{t}_1 ,\dots \widetilde{t}_{\widetilde{m} } \right\}$ into $\widetilde{m}^f (\leq \widetilde{m})$ subsets $C = \{C_1, C_2, \dots , C_{\widetilde{m}^f}\}$ so as to minimize the within-cluster distance,
    such that each cluster $R_j$ has a diameter at most $b_n$. 
    
    \For{$i = 1, \dots, \widetilde{m}^f$ }
    {   
        \uIf{$\vert R_j \vert = 1$}{
        \ Set $l_i = (R_j-b_n)$ and $u_i = (R_j+b_n)$. 
        Apply the exhaustive search method for each time point $s$ in the search domain $(l_i, u_i)$ and estimate the change point $\widetilde{t}_i$ by using the observations within the interval $[R_j - b_n, R_j + b_n )$:
        \begin{align*}
           \widetilde{t}_i^f = \argmin_{s \in (l_i, u_i )} \left \{\sum_{t= R_j - b_n}^{s-1} \left\| \bm{y}_{t}  - \widehat{\bm{B}}_{i}\bm{x}_{t} \right\|_2^2 +  \sum_{t= s}^{R_j + b_n -1 } \left\| \bm{y}_{t}  - \widehat{\bm{B}}_{i+1}\bm{x}_{t} \right\|_2^2 \right\}, 
        \end{align*}
        where $ {\widehat{\bm{B}}}_{i} = \sum_{k=1}^{{\lfloor\frac{1}{2}\left(\max(J_{i-1})+\min(J_{i})\right) \rfloor}} \widehat{\bm{\Theta}}_k, \text{ for } i = 1, \dots, \widetilde{m}^f + 1.$
        }
        \uIf{$\vert R_j \vert > 1$}{
        \  Set $l_i = \min (R_j)$ and $u_i = \max (R_j)$.
        Apply the exhaustive search method for each time point $s$ in the search domain $(l_i, u_i)$ and estimate the change point $\widetilde{t}_i$ by using the observations within the interval $[\min(R_j) - b_n, \max(R_j) + b_n$):
        \begin{align*}
           \widetilde{t}_i^f = \argmin_{s \in (l_i, u_i )} \left \{\sum_{t= \min(R_j) - b_n}^{s-1} \left\| \bm{y}_{t}  - \widehat{\bm{B}}_{i}\bm{x}_{t} \right\|_2^2 +  \sum_{t= s}^{\max(R_j) + b_n -1 } \left\| \bm{y}_{t}  - \widehat{\bm{B}}_{i+1}\bm{x}_{t} \right\|_2^2 \right \},
        \end{align*}
       where $ {\widehat{\bm{B}}}_{i} = \sum_{k=1}^{{\lfloor\frac{1}{2}\left(\max(J_{i-1})+\min(J_{i})\right) \rfloor}} \widehat{\bm{\Theta}}_k,  \text{ for } i = 1, \dots, \widetilde{m}^f + 1.$ 
        }
    }
    
    \KwOutput{The final estimated change points $ \widetilde{\mathcal{A}}_n^f = \left\lbrace \widetilde{t}_1^f, \ldots, \widetilde{t}_{\widetilde{m}^f}^f  \right\rbrace $ and estimated coefficient parameters $\left \{\widehat{\bm{B}}_{1}, \dotsm \widehat{\bm{B}}_{\widetilde{m}^f+1} \right\}$. }
    \caption{\textbf{Exhaustive Search Procedure}}
     \label{algo:2}
\end{algorithm}

\subsection{Details about hard-thresholding value selection}\label{subsec:hard-thresholding-tuning}
The main idea of the procedure is to combine the $K$-means clustering method \citep{hartigan1979algorithm} with the BIC criterion \citep{schwarz1978estimating} to cluster the changes in the parameter matrix into two subgroups. 
The detailed steps are:
\begin{itemize}
    \item Step 1 (initial state): Denote the jumps for each block by setting $v_i = {\left\|\widehat{\bm{\Theta}}_i \right\|_F} $, $i = 2, \cdots, k_n$ and $v_1 = 0 $. 
    Set $V = (v_1,  \cdots, v_{k_n})$.
    Denote the set of selected blocks with large jumps as $J$ (initially, this is an empty set) and set $\text{BIC}^{old} = \infty $.   
     \item Step 2 (recursion state): 
     Apply $K$-means clustering to the jump vector $V$ with two centers. Denote the sub-vector with a smaller center as the small subgroup, $V_S$ , and the other sub-vector as the large subgroup, $V_L$.
     Add the corresponding blocks in the large subgroup into $J$. 
    {
    Use the estimated parameters  $\widehat{\bm{\Theta}}_{i} $ for each block $i \in J$ to compute the BIC and denote it by $\text{BIC}^{new}$.
    By adding more blocks into $J$, we increase the number of parameters estimated by the model.
    }
     Compute the difference $\text{BIC}^{\text{diff}}= \text{BIC}^{new} - \text{BIC}^{old}$. Update $\text{BIC}^{old}= \text{BIC}^{new}$ and $V = V_S$. Repeat this step until $\text{BIC}^{\text{diff}} \geq 0$.

     \item Step 3 (output state): Set $\widetilde{I}_n = J$. It contains indices of blocks with large jumps.
\end{itemize}

\subsection{Details about block clustering}\label{subsec:blockcluster}

The block clustering step is based on a data-driven procedure to partition the $\widetilde{m}$ candidate change points into $\widetilde{m}^f$ clusters.  
In particular, we select the optimal number of cluster which maximize the Gap statistics. 
We only accept the optimal solution if the diameter of clusters ($\max(R_j) - \min (R_j)$) is at most $\kappa_1 b_n$.
If not, we continue choose the next optimal solution based on the Gap statistics till we find the solution under the constraint that the diameter of all selected clusters are less than or equal to $\kappa_1 b_n$. 
After that, we decrease the number of clusters if the distance between any two contiguous clusters is too close (i.e.,  $\min(R_j) - \max (C_{i-1})$ less than or equal to $\kappa_2 b_n$).
In practice, we choose a larger $\kappa_1$ and $\kappa_2$ for small $b_n$, and a smaller $\kappa_1$ and $\kappa_2$ for large $b_n$. Specifically, 
If $b_n \leq {\sqrt{n}}/{4}$, then $\kappa_1 = 9$, $\kappa_2 = 7$; if ${\sqrt{n}}/{4} < b_n \leq {\sqrt{n}}/{2}$, then $\kappa_1 = 7$, $\kappa_2 = 5$;  otherwise, we set $\kappa_1 = 5$, $\kappa_2 = 3$.

\subsection{Discussion about the choice of $\widehat{B}_j$}\label{subsec:para_est}

It is worth-noting that we did not threshold any $\widehat{\Theta}_j$ to be 0 in the  hard-thresholding procedure (Step II). 
In fact, the only thing we do in the Step II is selecting some block-end time point whose jumps are large enough (above a threshold value $\omega_n$).
Specifically, the set of estimated change points after hard-thresholding is given by
 \[ \widetilde{A}_n  = \left\{\widetilde{t}_1 ,\dots,  \widetilde{t}_{\widetilde{m} } \right\} = \left \{r_{i-1} : \left\|\widehat{\bm{\Theta}}_i\right \|_F  > \omega_n,\  i = 2,\dots, k_n \right\},
 \]
where $\omega_n$ is the hard-threshold value.
Those $\widehat{\Theta}_i$ with small values (i.e., $\left\|\widehat{\bm{\Theta}}_i\right \|_F  \leq \omega_n$) are  still used in (3.4) to estimate the coefficient $\widehat{B}_j$.

Now, suppose we set some of $\widehat{\Theta}_i$'s  (which have smaller norm values) to 0 in Step II. In this case, we cannot guarantee the same theoretical results of consistent estimation of segment-specific model parameters anymore.

Consider a simple case where there is only one change points $t_{1} = \lfloor n/2 \rfloor$ in the middle, i.e., $m_0 = 1$. 
Suppose there is only one block time point $s_{1}$ within the interval $(t_{1}- b_n, t_{1}+b_n)$. Denote this $s_{1}$ by $\widehat{t}_1$. 
By the first part of the proof in Lemma~\ref{lemma_lm_theta_2}, we have $\widehat{t}_{1} \in \widehat{\mathcal{A}}_n$. 
In that case, we have
{\footnotesize
\begin{align} 
&\left\lVert \widehat{\bm{\Theta}}_{i^
\star} \right\rVert_F = \left\lVert \sum_{k=1}^{i^\star}\widehat{\bm{\Theta}}_k - \sum_{k=1}^{i^\star-1}\widehat{\bm{\Theta}}_k \right\rVert_F   = \left\lVert \left(\sum_{k=1}^{i^\star}\widehat{\bm{\Theta}}_k -\bm{B}_{2}^\star \right) -  \left(\sum_{k=1}^{i^\star-1}\widehat{\bm{\Theta}}_k  - \bm{B}_{1}^\star \right) + \left(\bm{B}_{2}^\star - \bm{B}_{1}^\star \right) \right\rVert_F  \nonumber\\
\geq &\left\lVert \bm{B}_{2}^\star-\bm{B}_{1}^\star\right\rVert_F - \left(\left\lVert \sum_{k=1}^{i^\star}\widehat{\bm{\Theta}}_k - \bm{B}_{2}^\star \right\rVert_F +\left\lVert \sum_{k=1}^{i^\star-1}\widehat{\bm{\Theta}}_k -\bm{B}_{1}^\star\right\rVert_F \right) \nonumber\\
\geq & \nu_n - O_p \Bigg(\sqrt{\frac{d^\star_n\log (p_x p_y \vee n)}{b_n}}\Bigg), \nonumber
\end{align}}where $\widehat{t}_1 = r_{i^\star-1}$ is the selected block-end time point,
$\nu_n =\left\lVert \bm{B}_{2}^\star-\bm{B}_{1}^\star \right\rVert_F $ is the jump size.
Setting $\omega_n = \frac{1}{2}v_n$, we have
the set of indices of candidate blocks $\widetilde{I}_n = \{ i^\star\} = \{ \lfloor k_n/2 \rfloor\}  $.
For any other $i \notin \widetilde{I}_n$, we have 
${\left\Vert \widehat{\bm{\Theta}}_{i}\right\Vert_F = O_p\left(\sqrt{\frac{d^\star_n\log (p_x p_y \vee n)}{b_n}}\right). } $




Now, we define a new local coefficient parameter estimates 
for the second segment as
\begin{equation}
 \Hat{\Hat{{\bm{B}}}}_{2} = \sum_{i\in \left\{1, \dots, U_2 \right\}  \cap \widetilde{I}_n}
 \widehat{\bm{\Theta}}_i, \nonumber
\end{equation}
where $U_2 = \left\lfloor\frac{1}{2}\left(\max(J_{1})+\min(J_{2})\right) \right\rfloor = \left\lfloor\frac{1}{2}\left(i^\star +k_n)\right) \right\rfloor= \left\lfloor\frac{3}{4} k_n\right \rfloor$.

Suppose  for any $i\in \left\{1, \dots, U_2 \right\}  \cap \widetilde{I}_n^c$,  we have 
$\widehat{\bm{\Theta}}_i =  \widehat{\bm{\Theta}}_1 = c\sqrt{\frac{d^\star_n\log (p_x p_y \vee n)}{b_n}}$, where $c$ is some positive constant.  Then we have
{\small\begin{align*}
\left \Vert{\Hat{\Hat{{\bm{B}}}}}_{2} -\bm{B}_2^\star \right\Vert_F &\leq \left \Vert{\Hat{\Hat{{\bm{B}}}}}_{2} - {\widehat{\bm{B}}}_{2} \right\Vert_F  + \left \Vert{{\Hat{{\bm{B}}}}}_{2}  -\bm{B}_2^\star \right\Vert_F \nonumber \\
 &= \left \Vert \sum_{i\in \left\{1, \dots, U_2 \right\}  \cap \widetilde{I}_n}
 \widehat{\bm{\Theta}}_i - \sum_{i\in \left\{1, \dots, U_2 \right\} }
 \widehat{\bm{\Theta}}_i \right\Vert_F  + O_p\left ( 
 \sqrt{\frac{d^\star_n\log (p_x p_y \vee n)}{b_n}} \right )  \nonumber \\
& = \left \Vert \sum_{i\in \left\{1, \dots, U_2 \right\}  \cap \widetilde{I}_n^c}
 \widehat{\bm{\Theta}}_i \right\Vert_F + O_p\left ( 
 \sqrt{\frac{d^\star_n\log (p_x p_y \vee n)}{b_n}} \right ) \nonumber \\
 & = \sum_{i\in \left\{1, \dots, U_2 \right\}  \cap \widetilde{I}_n^c}\left \Vert
 \widehat{\bm{\Theta}}_i \right\Vert_F  + O_p\left ( 
 \sqrt{\frac{d^\star_n\log (p_x p_y \vee n)}{b_n}} \right ) \nonumber \\
 & = \sum_{i\in \left\{1, \dots, U_2 \right\}  \cap \widetilde{I}_n^c}\left \Vert
 \widehat{\bm{\Theta}}_1 \right\Vert_F + O_p\left ( 
 \sqrt{\frac{d^\star_n\log (p_x p_y \vee n)}{b_n}} \right ) \nonumber  \\
 & = O_p\left ( k_n
 \sqrt{\frac{d^\star_n\log (p_x p_y \vee n)}{b_n}} \right ).
\end{align*}}
Note that $k_n = \lfloor \frac{n}{b_n} \rfloor$, which goes to infinity as $n \rightarrow \infty$ (as described in  assumption A4).
Thus,  the consistency result of estimation of segment-specific model parameters does not hold anymore.



Define the thresholded estimate ${\widetilde{\widetilde{\bm{B}}}}_{2}$ as  
$
{\widetilde{\widetilde{\bm{B}}}}_{2} = \Hat{\Hat{{\bm{B}}}}_{2}\mathbbm{1}_{ \left\{ \vert \Hat{\Hat{{\bm{B}}}}_{2} \vert > \eta_{n, 2}  \right\}}.
$
Let $S$ denote the support of $\bm{B}_2^\star$.
To derive the upper bound on the number of false positives selected by thresholded lasso, note that 
{\footnotesize
\begin{align*}
    \left \vert  \text{supp}(\widetilde{\widetilde{\bm{B}}}_2) \backslash  \text{supp}(\bm{B}_2^\star) \right \vert  = \sum_{s \notin S}\mathbbm{1}_{\{ \vert \Hat{\Hat{{\bm{B}}}}_{2,s} \vert > {\eta_n} \}} \leq \sum_{s \notin S} \vert \Hat{\Hat{{\bm{B}}}}_{2, s} \vert/{\eta_{n,2}} \leq \frac{1}{{\eta_{n,2}}} \sum_{s \notin S} \vert v_s \vert  \nonumber \\
    \leq
    \frac{3}{{\eta_{n,2} }} \sum_{s \in S} \vert v_s \vert 
    \leq
    \frac{3\Vert v \Vert_1}{{\eta_{n,2}}} \leq
    {\frac{12 \sqrt{d_n^\star} \Vert v \Vert_F}{{\eta_{n,2}}} }  \nonumber\\
    = O_p\left ( k_n d^\star_n \right ),
\end{align*}} 
where $v = \Hat{\Hat{{\bm{B}}}}_j - \bm{B}_j^\star$ and 
{$\eta_{n,2} = {C_2}\sqrt{\frac{{\log (p_x p_y \vee n)} }{b_n}}$ for some positive constant $C_2$}.
Therefore, the variable selection result by thresholding dose not hold anymore as $k_n$ goes to infinity.


We now present a simulation to investigate the numerical performance of two different parameter estimations.
Here, we consider a multiple linear model with $ n = 2,000 $, $ p = 150 $, $ m_0 = 3 $.  The number of non-zero elements of coefficient vectors in $j$th segments $d_j= 15$, for all $j= 1, \dots, m_0 + 1$.
 The coefficient vector are chosen to have the random sparse structure in each segment, with different entries $ -3 $, $ 5 $, $-3$ and $3$, respectively. 
 The error variance  is $\Sigma = I$.
The  change points are equally spaced: 
    $ t_1 = \lfloor \frac{n}{4} \rfloor  = 500 $, $ t_2 =  \lfloor \frac{2n}{4} \rfloor  = 1000$,
    $ t_3 =  \lfloor \frac{3n}{4} \rfloor  =1500$,
    and $b_n= 40$. 
    For parameter estimation, we evaluated the performance of ${\widehat{\bm{B}}}_{j}$ (not setting zero) and $\Hat{\Hat{{\bm{B}}}}_{j}$ (setting zero) by  reporting 
the mean and standard deviation of relative estimation error (REE), the true positive rate (TPR) and the false positive rate (FPR) (defined in \eqref{REE}). 

Specifically, for TPR and FPR, we use the median number of nonzero and zero elements among 100 replicates.
As shown in Table~\ref{tab:beta_est_comp}, the  ${\widehat{\bm{B}}}_{j}$ has better performance in terms of the parameter estimation and variable selection.
    \begin{table}[!ht]
    \caption{\label{tab:beta_est_comp} Results of mean and standard deviation of relative estimation error (REE), true positive rate (TPR), and false positive rate (FPR) for estimated coefficients.}
        \centering
        \small
        {\begin{tabular}{cc ccccc}
        \hline
        method & mean ( sd) & TPR & FPR\\
         \hline
        ${\widehat{\bm{B}}}_{j}$     & 0.0981 (0.0829) & 1 & 0\\
         $\Hat{\Hat{{\bm{B}}}}_{j}$    &0.2879 (0.1114) & 0.95 & 0.0083\\
          \hline
        \end{tabular}}
    \end{table}


\subsection{Details about thresholding value selection for parameter estimation}\label{subsec:thresholding-tuning}
Similar to the hard-thresholding part, this part is also based on a data-driven procedure for selecting the threshold value. The idea is to use the BIC criterion to grid search the threshold value. We choose the value that has the lowest BIC. 
Here, we simplify the notation $\eta_{n,j}$ as $\eta_{j}$. 
The main steps are:
\begin{itemize}
    \item Step 1 (initial construction): 
    Define the $\left\{ \widehat{\bm{B}}_{j}\mathbbm{1}_{\{ \vert \widehat{\bm{B}}_{j} \vert > \eta_j \}}  \right \}$ as element-wise thresholding such that $ \widehat{\bm{B}}_{j,hl}  = 0$ if $\left\vert \widehat{\bm{B}}_{j,hl}  \right\vert  \leq {\eta_j}$ and unchanged otherwise, for all $j = 1, \dots, m_0+1, h = 1,\dots, p_y, l = 1, \dots, p_x$. 
    Following \cite{friedman2010regularization},  for each $j = 1, \dots, m_0+1$, we construct a sequence of $K$ values for ${\eta_j}$, decreasing from $\eta_{j,\text{max}}$ to $\eta_{j, \text{min}}$ on the log scale, where the maximum value  ${\eta_{j,\text{max}} }$ is the smallest value for which the entire estimated parameter $\widehat{\bm{B}}_{j}$ equals to zero; the minimum value $\eta_{j,\text{min}}$ is the largest value for which the entire estimated parameters $\widehat{\bm{B}}_j$ remain unchanged. 
In practice, we choose  $K=25$ when $\eta_{j,\text{max}}/\eta_{j,\text{min}} < 10^{4}$  and $K = 50$ otherwise. 
     \item Step 2 (grid search): For each segment $j = 1, \dots, m_0+1$,  
     compute the BIC by using the estimated parameters $\widetilde{\bm{B}}_j= \left\{ \widehat{\bm{B}}_{j}\mathbbm{1}_{\{ \vert \widehat{\bm{B}}_{j} \vert > \eta_{j,k} \}}  \right \}$ and denote it by $\text{BIC}_{k}$, $k = 1,\dots, K$.
     Choose the value $\eta_{j, k^\star}$ that minimize the BIC value, i.e., $ k^\star = \argmin_{k \in \{1,\dots, K\}} \mbox{BIC}_k$.
     \item Step 3 (thresholded variant): We define a thresholded variant $\widetilde{\bm{B}}_j$  as the final estimated parameters, where  
$\widetilde{\bm{B}}_{j} =\left\{ \widehat{\bm{B}}_{j}\mathbbm{1}_{\{ \vert \widehat{\bm{B}}_{j} \vert > \eta_{j, k^\star} \}}  \right \}$, for all $j = 1, \dots, m_0+1$. 
\end{itemize}

\subsection{Details about search domain selection}\label{subsec:domain}
Based on Assumption~A3, the search domain $S$ can be selected as follows: if $\sqrt{n} > p_xp_y$, we set the range of block size from $\log n (\log p_x + \log p_y)$ to $\min (\sqrt{n}, n/20)$; if  $\sqrt{n} \leq p_xp_y$, we set the range of block size from $\log n (\log p_x + \log p_y)$ to $\min (\sqrt{n}(\log p_x + \log p_y),  n/20)$. In the next section, we use the optimal block size when comparing with other competing methods.

\section{High-dimensional Multiple Linear Regression Model}\label{sec:MLR}
\setcounter{equation}{0}
We consider the classic univariate multiple linear regression model that the values of coefficient vector change over time. 
In this case,  setting the parameters $\bm{B}_{j}^\star =  {{{\bm{\beta}}_{j}^\star}}^\prime$, $p_x = p$, $p_y = 1$ in  the model representation  in \eqref{eq:model_1}, the  structural break multiple regression model is given by 
\begin{equation}\label{eq:model_mlr}
    y_t = \sum_{j = 1}^{m_0+1}\left (\bm{x}_t^\prime{\bm{\beta}}_{j}^\star + {\varepsilon}_{j,t} \right ) \mathbbm{1}_{\{t_{j-1}\leq t <t_j \}} , \   t =  1,\dots, n,
\end{equation}
where $y_t \in \R$ is the response at time $t$; $\bm{x}_t \in \R^p$ is the predictor vector at time $t$; $\bm{\beta}_{j}^\star\in \R^{p}$ is the sparse coefficient vector during the $j$th segment; and  $\varepsilon_{j,t} \in \R$ is
a white noise during the $j$th segment at time $t$, uncorrelated with $\bm{x}_t$, with mean $0$ and finite variance $\sigma_{j}$.

Define $\bm{\theta}_1 = {{\bm{\beta}}_{1}^\star}$ and 
\[\bm{\theta}_i = \begin{cases}
      {{\bm{\beta}}_{j+1}^\star} -{{\bm{\beta}}_{j}^\star} , & \text{when}\ t_j \in [r_{i-1}, r_{i})\  \text{for some}\ j  \\
      \bm{0}, & \text{otherwise,}
    \end{cases}\]
for $i = 2,3,\dots, k_n$.
In this case, the linear regression model in terms of $\Theta$ can be rewrite as
{\footnotesize
\begin{equation}\label{block_model_mlr}
\underbrace{\left(\begin{array}{c}
\bm{y}_{(1: r_1-1)}\\
\bm{y}_{(r_1: r_2-1)}\\
\vdots \\
\bm{y}_{(r_{k_n-1}:r_{k_n}-1)}\\
\end{array}\right)}_{\mathcal{Y}}= 
\underbrace{\left(\begin{array}{c cccc}
\bm{x}_{(1: r_1-1)}  & \mathbf{0} &\dots & \mathbf{0}\\
\bm{x}_{(r_1: r_2-1)} & \bm{x}_{(r_1: r_2-1)} &\dots &\mathbf{0}\\
\vdots & \vdots & \ddots & \vdots \\
\bm{x}_{(r_{k_n-1}:r_{k_n}-1)} & \bm{x}_{(r_{k_n-1}:r_{k_n}-1)}& \dots  &\bm{x}_{(r_{k_n-1}:r_{k_n}-1)}\\
\end{array}\right)}_{\mathcal{X}}
\underbrace{
\begin{pmatrix}
{\bm{\theta}_{1}} \\
{\bm{\theta}_{2}}\\
\vdots \\
{\bm{\theta}_{k_n}}\\
\end{pmatrix}}_{{\Theta}}
+\underbrace{
\left(\begin{array}{c}
{\bm{\zeta}_{(1: r_1-1)}}\\
{\bm{\zeta}_{(r_1: r_2-1)}}\\
\vdots \\
{\bm{\zeta}_{(r_{k_n-1}:r_{k_n}-1)} }\\
\end{array}\right)}_{E},   
\end{equation}}
where $\bm{y}_{(a:b)}:= ({y}_a, \dots, {y}_{b})^\prime$, $\bm{x}_{(a:b)}:= (\bm{x}_a, \dots, \bm{x}_{b})^\prime$,
{$\bm{\zeta}_{(a:b)}:= ({\zeta}_a, \dots, {\zeta}_{b})^\prime$}; 
$\mathcal{Y} \in \R^{n } $, $\mathcal{X} \in \R^{n \times k_n p} $, ${\Theta} \in \R^{k_np }$ and $E \in \R^{n} $.

The estimated coefficient parameters from block fused lasso is given by 
\begin{equation}\label{est_coef_mlr}
{\widehat{\bm{\beta}}}_{j} = \sum_{i=1}^{\frac{1}{2}\left(\max(J_{j-1})+\min(J_{j})\right)} \widehat{\bm{\theta}}_j, \text{ for } j = 1, \dots, \widetilde{m}^f+1. 
\end{equation}
Define the thresholded variant estimate $\widetilde{\bm{\beta}}_{j}$ as  
\begin{equation}\label{est_coef_threshold_mlr}
{\widetilde{\bm{\beta}}}_{j} = \widehat{\bm{\beta}}_{j}\mathbbm{1}_{ \left\{ \vert \widehat{\bm{\beta}}_{j} \vert > \eta_{n,j}  \right\}}, \text{ for } j = 1, \dots, \widetilde{m}^f+1. 
\end{equation}

To establish consistency properties of the estimation procedure, the following assumptions are needed:

\begin{itemize}
 \item[(E1.)]For the $j$-th segment, where $j = 1, 2, \dots, m_0 +1$, the process $y_{j,t} = {{\bm{x}_{j,t}}^\prime\bm{\beta}_{j}^\star} + {\varepsilon}_{j,t}$ is a linear regression model,  where the $p$-dimensional variables $\{\bm{x}_{j,t}\}$ are sub-Gaussian random vectors with parameters $(\Sigma_{x,j}, \sigma_{x,j}^2)$ and the errors $\{{\varepsilon}_{j,t}\}$ are i.i.d. sub-Gaussian variables with parameter $ \sigma_{\varepsilon,j}^2$  (see the details of sub-Gaussian definition in Appendix~\ref{sec:subGaussian}). Further,
    \begin{align*}
     &1/C_1 \leq  \min_{1\leq j \leq m_0+1}\Lambda_{\min}(\Sigma_{\bm{x},j}) \leq \max_{1\leq j \leq m_0+1}  \Lambda_{\text{max}}(\Sigma_{\bm{x},j}) \leq C_1,  \\
     &1/C_2  < \min_{1\leq j \leq m_0+1} \sigma_{\bm{x},j}^2 \leq \max_{1\leq j \leq m_0+1} \sigma_{\bm{x},j}^2 <C_2,\\
     &\text{ and } 
     1/C_3  < \min_{1\leq j \leq m_0+1} \sigma_{{\varepsilon},j}^2 \leq \max_{1\leq j \leq m_0+1} \sigma_{{\varepsilon},j}^2 <C_3,   
    \end{align*}
where $C_1$, $C_2$  and $C_3$ are positive constants.
    
    \item[(E2.)]The coefficient vectors  $\bm{\beta}_{j}^\star$  are sparse. More specifically, for all $j = 1,2,\dots, m_0+1$, $d_{j} \ll p$, i.e., $d_{j}/ p = o(1)$. Moreover, there exists a positive constant $M_{\bm{\beta}}  > 0$ such that
    \[\text{max}_{1 \leq j \leq m_0 +1}\lVert \bm{\beta}_{j}^\star \rVert_{\infty} \leq M_{\bm{\beta}} .\]
    
    \item[(E3.)]
    Let $\nu_n = \min_{1\leq j \leq m_0} \lVert \bm{\beta}_{j+1}^\star-\bm{\beta}_{j}^\star \rVert_2 $.  
     There exists a  positive sequence $b_n$ such that, as $n \rightarrow \infty$,
\[{\frac{\min_{1\leq j \leq m_0+1} \vert t_{j} - t_{j-1}\vert}{b_n }\rightarrow +\infty,   
\ d^\star_n \frac{\log ( p\vee n)}{b_n} \rightarrow 0 \text{ and }
\nu_n = \Omega\left(\sqrt{\frac{d^\star_n\log (p\vee n)}{b_n}}\right).}
\]
\item[(E4.)] {The regularization
parameters $\lambda_{1,n}$ and $\lambda_{2,n}$  satisfy
$\lambda_{1,n} = C_1 \sqrt{{ \log (p\vee n)}/{n}} \sqrt{{b_n}/{n}}$
, and 
$\lambda_{2,n} = C_2\sqrt{{\log  (p\vee n)}/{n}}\sqrt{{b_n}/{n} }$ 
for some large constant $C_1, C_2 > 0$.}
\end{itemize}

Note that in the Assumption E1, the sub-Gaussian distribution is only a sufficient but not necessary condition. 
Assumption E1 is a typical assumption for high-dimensional linear regression \citep{bickel2009simultaneous}.
As long as the distribution satisfies both  restricted eigenvalue conditions (A1) and a deviation condition (A2), we could establish consistent results.
Assumptions~E2-E4 are special cases of Assumptions A3-A5.

The next proposition is about the estimation consistency in high-dimensional multiple linear regression.

\begin{proposition}[Results of high-dimensional multiple linear regression]\label{thm:mlr} 
  Suppose the Assumptions E1-E4 hold. Then there exists a large enough constant $K>0$ such that , as $n \to +\infty$,
  \[\mathbb{P} \left( \widetilde{m}^f = m_0, \max_{1\leq j \leq m_0} \left|\widetilde{t}_j - t_j \right| \leq \frac{K {d_n^\star} \log( p\vee n)}{\nu_n^2} \right) \to 1.\]
Also, the solution $\widehat{\bm{B}}_j$ from \eqref{est_coef_mlr} satisfies 
    \begin{equation*}
    \max_{1\leq j \leq m_0+1}\left\Vert  \widehat{\bm{\beta}}_j - \bm{\beta}_j^\star\right\Vert_F =  O_p \left(\sqrt{\frac{d^\star_n\log (p\vee n)}{b_n}}\right)
    \end{equation*}
where $\bm{\beta}_j^\star$ is the true value coefficient parameter matrix at $j$th stationary segment. Further,
if {$\eta_{n,j} = {C_j}\sqrt{\frac{{\log (p  \vee n)} }{b_n}}$ for some positive constant $C_j$},
the thresholded variant $\widetilde{\bm{\beta}}_j$ from \eqref{est_coef_threshold_mlr}  satisfies
\begin{equation*}
  \max_{ 1 \leq j \leq m_0 + 1 }  \left \vert  \text{supp}(\widetilde{\bm{\beta}}_j) \backslash  \text{supp}( {\bm{\beta}}_j^\star) \right\vert =  O_p \left(d^\star_n\right).
    \end{equation*} 
\end{proposition}

Proposition~\ref{thm:mlr} shows that the TBFL method achieves a better consistency rate in terms of the  localization error than the Binary Segmentation through Estimated CUSUM statistics (BSE) developed in \cite{wang2019statistically} -as shown in Theorem 1  and Remark 3 of \cite{wang2019statistically}, which is $O_p \left( m_0{d_n^\star}   \log p \right)$. 
Both two methods guarantees on the number of estimated change points.

\section{Additional Details about Gaussian Graphical Model}\label{sec:ggm_more}
\setcounter{equation}{0}

Let  $\Sigma({-l,-k})$ denote the sub-matrix of $\Sigma$ with its $l$-th row and $k$-th column removed, $\Sigma({l,k})$ denote  the entry of matrix $\Sigma$ that lies in the $l$-th row and $k$-th column,  $\Sigma({l,\cdot})$ denote  the $l$-th row of matrix $\Sigma$ and  $\Sigma({\cdot, k})$ denote  the $k$-th column of matrix $\Sigma$. 
Consider the $p$-dimensional multivariate Gaussian distributed random variable
\[\bm{X} = (X^1,\dots,X^p)^\prime \sim \mathcal{N}(\mu,\Sigma ),\]
where $\mu$ is the unknown mean parameter and $\Sigma$ is the non-singular covariance matrix. Let $\Omega:= (\Sigma)^{-1}$ denote the precision matrix, with elements $(\Omega({l,k}))$, $1 \leq l,k \leq p$.
The conditional independence structure of the distribution can be represented by a graphical model $G= (V, E)$, where $V =[p]$  is the set of nodes corresponding to the $p$ coordinates and $E \subset V\times V$ is the set of edges in  capturing conditional independencies among these nodes.
Every pair of variables is not contained in the edge set $E$ if and only if the two variables are conditionally independent, given all remaining variables, and corresponds to a zero entry in the precision matrix $\Omega $, i.e.,
\begin{equation*}
  \Omega({l,k}) = 0  \Leftrightarrow (l,k) \notin E.  
\end{equation*}
Note that for Gaussian graphical models that the elements of
\begin{equation*}
\bm{a}^{l} = \argmin_{\bm{a} \in \R^{p}: a_l = 0} \mathbb{E}  ( X^l - \bm{a}^\prime \bm{X})^2 , 
\end{equation*}
are given by $\bm{a}^{l}_k = -\frac{\Omega({l,k})}{\Omega({l,l})}$, $k = 1, \dots, p$ and $k \neq l$.
 The set of nonzero coefficients of $\bm{a}^{l}$  is identical to the set $\{k \in [p] \backslash \{l\} : \Omega({l,k})\neq 0\}$ of nonzero entries in the precision matrix.
By building on the neighbourhood selection procedure, we could estimate the pattern of the precision matrix.

\cite{meinshausen2006high} proposed a simple approach for covariance selection that can be used for very large Gaussian graphs. 
They estimate a sparse graphical model by
estimating (individually) the neighborhood of each variable. Specifically, they fit a lasso model to each variable, using the others as predictors.
The component precision matrix $\Omega$ is estimated to be non-zero if either the estimated coefficient of variable $X^l$ on $X^k$, or the
estimated coefficient of variable $X^k$ on $X^l$, is non-zero. Given $n$ independent and identically distributed observations of $\bm{X}$, denoted by $\mathcal{D} = \{\bm{x}_1,\dots , \bm{x}_n \}$, then for each variable $X^l$, the Lasso estimate $\widehat{\bm{a}}^{l, \lambda}$ of $\bm{a}^l$ is given by
\begin{equation*}
\widehat{\bm{a}}^{l, \lambda} = \argmin_{\bm{a} \in \R^{p}: a_l = 0} \frac{1}{n} \sum_{i= 1}^n ( x_{i}^l - \bm{a}^\prime\bm{x}_{i} )^2  + \lambda \Vert \bm{a}\Vert_1, 
\end{equation*}
where $l = 1, \dots, p$, $\bm{x}_i = \left(x_{i}^1, \dots, x_{i}^p \right)^\prime \in \R^{p}$. They proved that for i.i.d. sample, the non-zero coefficients of $\widehat{\bm{a}}^{i, \lambda} $  consistently estimate the neighborhood of the node $i$, under a suitably chosen penalty parameter $\lambda$.

It is well known that if $\bm{X}$ follows a multivariate normal distribution $\mathcal{N} (0, \Sigma)$, then the conditional distribution of $X^l$ given $X^{-l}$ remains normally distributed, that is,
\[X^l \vert X^{-l} = \bm{x}^{-l}  \sim \mathcal{N}\left( \Sigma({l,-l})(\Sigma({-l,-l}))^{-1}\bm{x}^{-l},\Sigma({l,l}) -  \Sigma({l,-l}) (\Sigma({-l,-l}))^{-1} \Sigma({-l,l}) \right).\]

On the other hand, by block matrix inversion and take $l=1$ as example, we have
{\footnotesize
\begin{align}\label{eq:omega}
    \Omega = \left(\begin{array}{cc}
\Sigma({1,1}) & \Sigma({1,-1})\\
\Sigma({-1,1}) & \Sigma({-1,-1})
\end{array}\right)^{-1}
= \left(\begin{array}{cc}
c_1 & -c_1 \Sigma({1,-1}) (\Sigma({-1,-1}))^{-1} \\
-c_1(\Sigma({-1,-1}))^{-1} \Sigma({-1,1}) & \left(\Sigma({-1,-1}) - \Sigma({-1,1}) (\Sigma({1,1}))^{-1} \Sigma({1,-1})\right)^{-1}
\end{array}\right),
\end{align}}
where $c_1 = \left(\Sigma({1,1}) - \Sigma({1,-1}) (\Sigma({-1,-1}))^{-1} \Sigma({-1,1})\right)^{-1}$.
Therefore, we have 
\begin{equation}
  (\Sigma({-1,-1}))^{-1} \Sigma({-1,1}) = -(\Omega({1,1}))^{-1}\Omega({-1,1} )  \nonumber
\end{equation}

Applying \eqref{eq:omega}, we have
{\footnotesize
\begin{align}
&\Sigma({1,\cdot}) \left(\Sigma \right)^{-1} \Sigma({\cdot, 2}) \nonumber\\
=& \left(\begin{array}{cc}\Sigma({1, 1}) &  \Sigma({1,-1})\end{array}\right)
 \left(\begin{array}{cc}
c_1 & -c_1 \Sigma({1,-1}) \left(\Sigma({-1,-1})\right)^{-1} \\
-c_1(\Sigma({-1,-1}))^{-1} \Sigma({-1,1}) & \left(\Sigma({-1,-1}) - \Sigma({-1,1}) (\Sigma({1,1}))^{-1} \Sigma({1,-1})\right)^{-1}
\end{array}\right) \nonumber\\
& \quad \Sigma({\cdot, 2} )
\left(\begin{array}{c}\Sigma({1,2}) \\  \Sigma({-1,2})\end{array}\right)\nonumber\\
 =& 
\left(c_1 \Sigma({1,1})   -c_1 \Sigma({1,-1})(\Sigma({-1,-1}))^{-1} \Sigma({-1,1} )\right)\Sigma({1,2}) \nonumber\\
& \quad \quad \quad \quad - \left(c_1 \Sigma({1,1}) \Sigma({1,-1}) (\Sigma({-1,-1}))^{-1} - \Sigma({1,-1})\left(\Sigma({-1,-1}) - (\Sigma({1,1}))^{-1}\Sigma({-1,1} ) \Sigma({1,-1})\right)^{-1} \right)\Sigma({-1,2})
\nonumber\\
 =& 
\Sigma({1,2}) - (\left(1 - (\Sigma({1,1}))^{-1}\Sigma({1,-1}) (\Sigma({-1,-1}))^{-1} \Sigma({-1,1})\right)^{-1}  \Sigma({1,-1}) (\Sigma({-1,-1}))^{-1} \nonumber\\
& \quad \quad \quad - \Sigma({1,-1}) \left(\Sigma({-1,-1}) - (\Sigma({1,1}))^{-1}\Sigma({-1,1} ) \Sigma({1,-1})\right)^{-1} )\Sigma({-1,2})
\nonumber\\
=& 
\Sigma({1,2}) - ( \Sigma({1,-1})\left(\Sigma({-1,-1}) - (\Sigma({1,1}))^{-1}\Sigma({-1,1}) \Sigma({1,-1}) \right)^{-1}   \nonumber\\
& \quad \quad \quad  - \Sigma({1,-1})\left(\Sigma({-1,-1}) - (\Sigma({1,1}))^{-1}\Sigma({-1,1} ) \Sigma({1,-1})\right)^{-1} )\Sigma({-1,2})
\nonumber\\
=& 
\Sigma({1,2})  \nonumber
\end{align}}
And thus, we get
{\footnotesize
\begin{align}
    \text{Cov} \left( \varepsilon_{t}^l, x_{t}^k \right) &=   \text{Cov} \left( x_{t}^l - {\bm{x}_{t}^{-l}}^\prime \bm{\beta}_{-l}^l  , x_{t}^k \right) \nonumber\\
  & =   \text{Cov} \left( x_{t}^l +(\Omega({l,l}))^{-1} {\bm{x}_{t}^{-l}}^\prime \Omega({-l,l})  , x_{t}^k \right) \nonumber\\
  & =   \text{Cov}\left ( x_{t}^l   , x_{t}^k \right) +  \left(\Omega({l,l})\right)^{-1}\text{Cov} \left( {\bm{x}_{t}^{-l}}^\prime \Omega({-l,l}), x_{t}^k \right) \nonumber\\
  & =   \Sigma({l,k}) +  (\Omega({l,l}))^{-1}\E \left( x_{t}^k  {\bm{x}_{t}^{-l}}^\prime \Omega({-l,l})\right) \nonumber\\
  & =   \Sigma({l,k} )+ \sum_{i = 1 , i\neq l}^p \frac{\Omega({i,i})}{\Omega({l,l})} \Sigma({l,k}) \nonumber\\
  & =   \Sigma({l,k}) +  \left(\Omega({l,l})\right)^{-1} \Omega({l,-l}) \Sigma({-l,k}) \nonumber\\
  & =   \Sigma({l,k}) -\Sigma({l,-l}) \left(\Sigma({-l,-l}) \right)^{-1} \Sigma({-l,k}) \nonumber\\
  & =  0, \nonumber
\end{align}}
where the last equation holds by the fact that $\Sigma({l,-l}) \left(\Sigma({-l,-l}) \right)^{-1} \Sigma({-l,k}) =\Sigma({l,k})$. 
Therefore, the $l$-th component in $\bm{\varepsilon}_{t}$  is independent of $k$-th component in $\bm{x}_{t}$ for any $k \in [p] \backslash \{l\} $.



\section{Technical Lemmas}\label{sec:lemma}
\setcounter{equation}{0}

\begin{lemma}\label{lemma_lm_theta_1}
Suppose A1-A5 hold. 
For any $\widehat{t}_j$ in $\widehat{\mathcal{A}}_n$ such that $\min_{j'= 1,\dots,m_0}\left\vert \widehat{t}_j - t_{j'} \right\vert \geq b_n$, $t_{j_0 -1}< \widehat{t}_j  <  t_{j_0}$ and $\widehat{t}_j = r_i$ for some $i \in \{1,\dots, k_n - 1\}$, the following holds:
\[{\left\Vert \bm{B}_{j_0}^\star - \widehat{\bm{L}}_{j}\right\Vert_F =  O_p \left(\sqrt{\frac{d^\star_n\log (p_x p_y \vee n)}{b_n}}\right), \,
\left\Vert \bm{B}_{j_0}^\star - \widehat{\bm{L}}_{j+1}\right\Vert_F =  O_p \Bigg(\sqrt{\frac{d^\star_n\log (p_x p_y \vee n)}{b_n}}\Bigg),}\]
and 
\[{\left\Vert \widehat{\bm{\Theta}}_{i+1}\right\Vert_F = O_p\left(\sqrt{\frac{d^\star_n\log (p_x p_y \vee n)}{b_n}}\right), } \]
where $\widehat{\bm{L}}_{j} = \sum_{k=1}^{\widehat{i}_{j}-1}\widehat{\bm{\Theta}}_k$ and $\widehat{\bm{L}}_{j+1}   = \sum_{k=1}^{\widehat{i}_{j}}\widehat{\bm{\Theta}}_k$ are the partial sums of $\widehat{\bm{\Theta}}_k$;  $\widehat{i}_j$ is the  corresponding indice of candidate point  $\widehat{t}_j$.
\end{lemma}

\begin{proof}

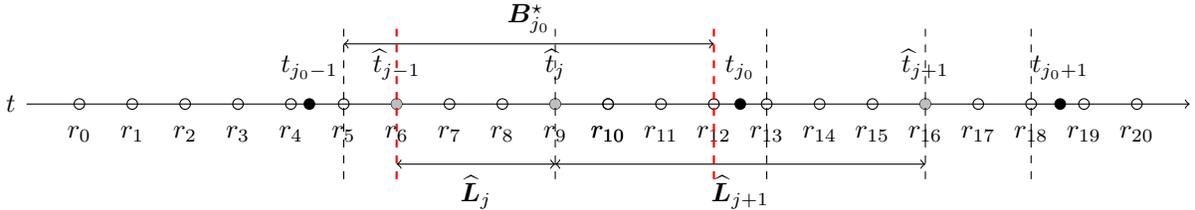
\begin{figure}[ht!]
    \centering
\begin{tikzpicture}
\centering
    \FPeval{\x}{clip(20)}
  \begin{scope}[font=\footnotesize,x=\x pt]
  \foreach \mypt in {0,...,\bound}{
    \FPeval{\result}{clip(\bound+\mypt)}
    \draw(\mypt,0)circle[radius=2pt];
    \draw(-\mypt,0)circle[radius=2pt];
    \node[at={(\mypt,0)},label=below:{$r_{\result}$}]{};
    \ifnum \mypt <1
        \node[at={(\mypt-\bound,0)},label=below:{$r_0$}]{};
    \else
        \node[at={(\mypt-\bound,0)},label=below:{$r_{\mypt}$}]{};
    \fi
  }
  \draw [->] (-\bound-1,0)--(\bound+1,0) node[pos=0,left]{$t$};
  
  \FPeval{\pp}{clip(-1)}
  \fill[gray!50]( \pp * \x pt ,0) circle[radius=2pt];
  \node[at={( \pp * \x pt,0)},above=5pt]{$\widehat{t}_{j}$};
  \FPeval{\ppp}{clip(-4)}
  \fill[gray!50]( \ppp * \x pt ,0) circle[radius=2pt];
  \node[at={( \ppp * \x pt,0)},above=5pt]{$\widehat{t}_{j-1}$};
  \draw [<->] (\ppp * \x pt, -0.8)--(\pp * \x pt,-0.8) node[pos= 1/2,below]{$\widehat{\bm{L}}_{j}$};
  \draw [red, dashed, thick] ( \ppp * \x pt,-1) -- ( \ppp * \x pt,1);
  
  \FPeval{\ppp}{clip(6)}
  \fill[gray!50]( \ppp * \x pt ,0) circle[radius=2pt];
  \node[at={( \ppp * \x pt,0)},above=5pt]{$\widehat{t}_{j+1}$};
  \draw [<->] (\ppp * \x pt, -0.8)--(\pp * \x pt,-0.8) node[pos= 1/2,below]{$\widehat{\bm{L}}_{j+1}$};
  \draw [dashed] ( \pp * \x pt,-1) -- ( \pp * \x pt,1);
  \draw [ dashed] ( \ppp * \x pt,-1) -- ( \ppp * \x pt,1);

  \FPeval{\pp}{clip(-6)}
  \FPeval{\xx}{clip( \pp * \x + 7 )}
  \FPeval{\dxx}{clip( (\pp + 1) * \x  )}
  \filldraw( \xx pt ,0)circle[radius=2pt];
  \node[at={(\xx pt,0)},above=5pt]{$t_{j_0-1}$};
  \draw [dashed] ( \dxx pt,-1) -- ( \dxx pt,1);
  
   \FPeval{\pp}{clip(2)}
   \draw [<->] (\dxx pt, 0.8)--( \pp * \x pt,0.8) node[pos= 1/2,above]{$\bm{B}_{j_0}^\star$};
  \FPeval{\xx}{clip( \pp * \x + 10 )}
  \FPeval{\dxx}{clip( \pp * \x  )}
   \filldraw( \xx pt ,0)circle[radius=2pt];
  \node[at={(\xx pt,0)},above=5pt]{$t_{j_0}$};
  \draw [red, dashed, thick] ( \dxx pt,-1) -- ( \dxx pt,1);
  \draw [dashed] ( \dxx + \x pt,-1) -- ( \dxx +\x pt,1);
  \FPeval{\pp}{clip(8)}
  \FPeval{\xx}{clip( \pp * \x + 11 )}
  \FPeval{\dxx}{clip( \pp * \x  )}
  \filldraw( \xx pt ,0)circle[radius=2pt];
  \node[at={(\xx pt,0)},above=5pt]{$t_{j_0+1}$};
  \draw [dashed] ( \dxx pt,-1) -- ( \dxx pt,1);
  
\end{scope}
\end{tikzpicture}
    \caption{Consider an estimated change point $\widehat{t}_j$ lies within the interval $[t_{j_0-1},t_{j_0} ]$ which is isolated from all the true change points, i.e.,
$t_{j_0-1} < \widehat{t}_j < t_{j_0}$, with $\vert \widehat{t}_j - t_{j_0} \vert \geq b_n$ and $\vert \widehat{t}_j - t_{j_0-1} \vert \geq b_n$.
The idea is to show the estimated coefficient parameters  $\widehat{\bm{L}}_j$ and $\widehat{\bm{L}}_{j+1}$  in intervals $[s_{j_0-1}\vee \widehat{t}_{j-1} ,\widehat{t}_j]$ and $[\widehat{t}_j , s_{j_0} \wedge \widehat{t}_{j+1}]$ both converge in $\ell_2$ to $\bm{B}_{j_0}^\star$.}
    \label{fig:interval_1}
\end{figure}

We consider an estimated change point $\widehat{t}_j$ lies within the interval $[t_{j_0-1},t_{j_0} ]$ which is isolated from all the true change points, i.e.,
$t_{j_0-1} < \widehat{t}_j < t_{j_0}$, with $\vert \widehat{t}_j - t_{j_0} \vert \geq b_n$ and $\vert \widehat{t}_j - t_{j_0-1} \vert \geq b_n$.
The idea is to show the estimated coefficient parameters  $\widehat{\bm{L}}_j$ and $\widehat{\bm{L}}_{j+1}$  in intervals $[t_{j_0-1}\vee \widehat{t}_{j-1} ,\widehat{t}_j]$ and $[\widehat{t}_j , t_{j_0} \wedge \widehat{t}_{j+1}]$ both converge in $\ell_2$ to $\bm{B}_{j_0}^\star$.

Denote the closest $r_i$ to the right side of $t_{j_0-1}$ by  $s_{j_0 -1}$, and denote the closest $r_i$ to the left side of $t_{j_0}$ by  $s_{j_0}$. Note that $s_{j_0 -1}-t_{j_0-1} < b_n $ and $t_{j_0 }-s_{j_0} < b_n $. If $\left\vert \widehat{t}_j - t_{j_0-1} \right\vert = b_n$, then $t_{j_0-1} = s_{j_0-1} $ and $\left\vert \widehat{t}_j - s_{j_0-1} \right\vert = b_n$; similarly, if $\left\vert \widehat{t}_j - t_{j_0} \right\vert = b_n$, then $t_{j_0} = s_{j_0} $ and $\left\vert \widehat{t}_j - s_{j_0} \right\vert = b_n$. If $\left\vert \widehat{t}_j - t_{j_0-1} \right\vert > b_n$ and $\left\vert \widehat{t}_j - t_{j_0} \right\vert > b_n$, then $\left\vert \widehat{t}_j - t_{j_0-1} \right\vert \geq 2b_n$ and $\vert \widehat{t}_j - t_{j_0} \vert \geq 2b_n$, which means that $\left\vert \widehat{t}_j - s_{j_0-1} \right\vert \geq b_n$ and $\left\vert \widehat{t}_j - s_{j_0} \right\vert \geq b_n$. Therefore, we have $\widehat{t}_j - s_{j_0-1} \vee \widehat{t}_{j-1} \geq  \min ( \widehat{t}_j - s_{j_0-1}, \widehat{t}_j - \widehat{t}_{j-1}) \geq b_n$ and $s_{j_0} \wedge \widehat{t}_{j+1}- \widehat{t}_{j}
 \geq  \min \left(s_{j_0} - \widehat{t}_{j}, \widehat{t}_{j+1} - \widehat{t}_{j} \right) \geq b_n$ . In other words, the lengths of the intervals $[s_{j_0-1}\vee \widehat{t}_{j-1} ,\widehat{t}_j]$ and $[\widehat{t}_j , s_{j_0} \wedge \widehat{t}_{j+1}]$  are large enough to verify restricted eigenvalue and deviation bound inequalities.

We first focus on the interval $[s_{j_0-1}\vee \widehat{t}_{j-1} ,\widehat{t}_j]$. Define a new parameter sequence $\bm{\Psi}_k$'s, $k= 1,\dots, k_n$ with $\bm{\Psi}_k = \widehat{\bm{\Theta}}_k$ except for two time points $k=s_{j_0-1}\vee \widehat{t}_{j-1}$ and $k=\widehat{t}_j$. If  $\widehat{t}_{j-1} > s_{j_0-1} $, set $\bm{\Psi}_{\widehat{i}_{j-1}} = \bm{B}_{j_0}^\star - \widehat{\bm{L}}_{j-1} $ and $\bm{\Psi}_{\widehat{t}_j} = \widehat{\bm{L}}_{j+1} -\bm{B}_{j_0}^\star $; if  $\widehat{t}_{j-1} \leq s_{j_0-1} $, set $\bm{\Psi}_{s_{j_0-1}} = \bm{B}_{j_0}^\star - \widehat{\bm{L}}_{j} $ and $\bm{\Psi}_{\widehat{t}_j} = \widehat{\bm{L}}_{j+1} -\bm{B}_{j_0}^\star $, where $\widehat{\bm{L}}_{j-1} = \sum_{k=1}^{\widehat{i}_{j-1} -1}\widehat{\bm{\Theta}}_k$, 
$\widehat{\bm{L}}_{j} = \sum_{k=1}^{\widehat{i}_{j-1}}\widehat{\bm{\Theta}}_k  = \sum_{k=1}^{\widehat{i}_{j}-1}\widehat{\bm{\Theta}}_k$ and $\widehat{\bm{L}}_{j+1}   = \sum_{k=1}^{\widehat{i}_{j}}\widehat{\bm{\Theta}}_k$ .
where $\widehat{i}_{j-1}$ and $\widehat{i}_j$ are the  corresponding indices of candidate points $\widehat{t}_{j-1}$ and $\widehat{t}_j$ .
By the definition of $\widehat{\bm{\theta}}$ in \eqref{eq:estimation_block}, the value of the function in \eqref{eq:estimation_block} is minimized at $\widehat{\bm{\theta}}$. Denoting $\bm{\psi} = \text{vec}(\bm{\Psi}_1, \dots, \bm{\Psi}_{k_n}) \in \mathbb{R}^{\pi_n \times 1}$, where $\pi_n = k_np_xp_y$, we have
{\footnotesize\begin{equation} 
\frac{1}{n}\left\Vert \mathbf{y} - \mathbf{Z}\widehat{\bm{\theta}}\right\Vert_{2}^{2} + \lambda_{1,n}\left\Vert\widehat{\bm{\theta}}\right\Vert_1 + \lambda_{2,n}\sum_{k=1}^{k_n}\left\Vert\sum_{j=1}^k\widehat{\bm{\Theta}}_j\right\Vert_1 \leq 
\frac{1}{n}\left\Vert \mathbf{y} - \bm{Z}\bm{\psi}\right\Vert_{2}^{2} + \lambda_{1,n}\left\Vert\bm{\psi}\right\Vert_1 + \lambda_{2,n}\sum_{k=1}^{k_n}\left\Vert\sum_{j=1}^k\bm{\Psi}_j\right\Vert_1.
\label{eq:obj_lm_1}
\end{equation}}

When $\widehat{t}_{j-1} > s_{j_0-1} $, some rearrangement of equation \eqref{eq:obj_lm_1} leads to

{\footnotesize
\begin{align} 
0 
&\leq \frac{1}{\widehat{t}_j - \widehat{t}_{j-1}} \sum_{l= \widehat{t}_{j-1} }^{\widehat{t}_j - 1} \bm{x}^{\prime}_{l}\left(\bm{B}_{j_0}^\star - \widehat{\bm{L}}_{j}\right)^{\prime}\left(\bm{B}_{j_0}^\star - \widehat{\bm{L}}_{j}\right)\bm{x}_{l}\nonumber\\
&\leq \frac{2}{\widehat{t}_j -  \widehat{t}_{j-1}} \sum_{l= \widehat{t}_{j-1} }^{\widehat{t}_j - 1} \bm{x}^{\prime}_{l}\left(\bm{B}_{j_0}^\star - \widehat{\bm{L}}_{j}\right)^\prime\bm{\varepsilon}_l\nonumber\\
& + \frac{n\lambda_{1,n}}{\widehat{t}_j -  \widehat{t}_{j-1}} 
\left(\left\Vert \bm{B}_{j_0}^\star - \widehat{\bm{L}}_{j+1}\right\Vert_1 + \left\Vert \bm{B}_{j_0}^\star - \widehat{\bm{L}}_{j-1}\right\Vert_1 -\left\Vert \widehat{\bm{L}}_{j+1}-\widehat{\bm{L}}_{j}\right\Vert_1  -\left\Vert \widehat{\bm{L}}_{j}-\widehat{\bm{L}}_{j-1}\right\Vert_1  \right) 
\nonumber\\
& +  \frac{n\lambda_{2,n}}{b_n}\left(\left\Vert \bm{B}_{j_0}^\star \right\Vert_1 -\left\Vert \widehat{\bm{L}}_{j} \right\Vert_1 \right) \nonumber\\
& \leq  \frac{2}{\widehat{t}_j -  \widehat{t}_{j-1}} 
\left\Vert \sum_{l = \widehat{t}_{j-1} }^{t_j - 1} \bm{x}_{l} \bm{\varepsilon}_l^\prime  \right\Vert_\infty \left\Vert \bm{B}_{j_0}^\star - \widehat{\bm{L}}_{j}\right\Vert_1
 + \frac{2n\lambda_{1,n}}{\widehat{t}_j - \widehat{t}_{j-1}} 
\left\Vert \bm{B}_{j_0}^\star - \widehat{\bm{L}}_{j}\right\Vert_1
 + \frac{n \lambda_{2,n}}{b_n}\left(\left\Vert \bm{B}_{j_0}^\star \right\Vert_1 -\left\Vert \widehat{\bm{L}}_{j} \right\Vert_1 \right ) \nonumber\\
& \leq \left(\frac{2n\lambda_{1,n}}{\widehat{t}_j - \widehat{t}_{j-1}} + 
 C \sqrt{\frac{{\log (p_x  p_y \vee n)} }{b_n}}\right)
\left\Vert \bm{B}_{j_0}^\star - \widehat{\bm{L}}_{j}\right\Vert_1
 +  \frac{n\lambda_{2,n}}{b_n}\left(\left\Vert \bm{B}_{j_0}^\star \right\Vert_1 -\left\Vert \widehat{\bm{L}}_{j} \right\Vert_1 \right ) \nonumber\\
 & \leq  \frac{n\lambda_{2,n}}{2b_n}
\left\Vert \bm{B}_{j_0}^\star - \widehat{\bm{L}}_{j}\right\Vert_1
 + \frac{n \lambda_{2,n}}{b_n}\left(\left\Vert \bm{B}_{j_0}^\star \right\Vert_1 -\left\Vert \widehat{\bm{L}}_{j} \right\Vert_1 \right) \nonumber\\
 & \leq  \frac{n\lambda_{2,n}}{2b_n}
\left\Vert \bm{B}_{j_0}^\star - \widehat{\bm{L}}_{j}\right\Vert_{1,\mathcal{I}} +  \frac{n\lambda_{2,n}}{2b_n}
\left\Vert \bm{B}_{j_0}^\star - \widehat{\bm{L}}_{j}\right\Vert_{1,\mathcal{I}^c}
 +\frac{ n \lambda_{2,n}}{b_n}\left( \left\Vert \bm{B}_{j_0}^\star - \widehat{\bm{L}}_{j}\right\Vert_{1,\mathcal{I}} -
 \left\Vert \bm{B}_{j_0}^\star - \widehat{\bm{L}}_{j}\right\Vert_{1,\mathcal{I}^c}\right) \nonumber\\
 & \leq \frac{3n \lambda_{2,n}}{2b_n}
\left\Vert \bm{B}_{j_0}^\star - \widehat{\bm{L}}_{j}\right\Vert_{1,\mathcal{I}}
 - \frac{n\lambda_{2,n}}{2b_n}
\left\Vert \bm{B}_{j_0}^\star - \widehat{\bm{L}}_{j}\right\Vert_{1,\mathcal{I}^c}.
 \label{eq:lemma_1_1}
\end{align}}

The third inequality holds due to the H\"{o}lder's inequality and the triangle inequality. The fourth inequality holds by 
deviation bound condition in \eqref{eq:dev_bound}. The fifth inequality is based on the selection of $\lambda_{1,n}$ and $\lambda_{2,n}$.  

Similarly, when $\widehat{t}_{j-1} \leq s_{j_0-1} $, some rearrangement of equation \eqref{eq:obj_lm_1} leads to
{\footnotesize\begin{align}
0
&\leq \frac{1}{\widehat{t}_j - s_{j_0-1}} \sum_{l= s_{j_0-1} }^{\widehat{t}_j - 1} \bm{x}^{\prime}_{l}\left(\bm{B}_{j_0}^\star - \widehat{\bm{L}}_{j}\right)^{\prime}\left(\bm{B}_{j_0}^\star - \widehat{\bm{L}}_{j}\right)\bm{x}_{l}\nonumber\\
&\leq \frac{2}{\widehat{t}_j -  s_{j_0-1}} \sum_{l= s_{j_0-1} }^{\widehat{t}_j - 1}\bm{x}^{\prime}_{l}\left(\bm{B}_{j_0}^\star - \widehat{\bm{L}}_{j}\right)^{\prime}\bm{\varepsilon}_l\nonumber\\
& + \frac{n\lambda_{1,n}}{\widehat{t}_j -  s_{j_0-1}} 
\left(\left\Vert \bm{B}_{j_0}^\star - \widehat{\bm{L}}_{j+1}\right\Vert_1 + \left\Vert \bm{B}_{j_0}^\star - \widehat{\bm{L}}_{j} \right\Vert_1-\left\Vert \widehat{\bm{L}}_{j+1}-\widehat{\bm{L}}_{j}\right\Vert_1  \right) 
\nonumber\\
& +  \frac{n\lambda_{2,n}}{b_n}\left(\left\Vert \bm{B}_{j_0}^\star \right\Vert_1 -\left\Vert \widehat{\bm{L}}_{j} \right\Vert_1 \right) \nonumber\\
&\leq \left(\frac{2n\lambda_{1,n}}{\widehat{t}_j - s_{j_0-1}} + 
 C \sqrt{\frac{{\log (p_x p_y \vee n)} }{b_n}}\right)
\left\Vert \bm{B}_{j_0}^\star - \widehat{\bm{L}}_{j}\right\Vert_1
 + \frac{n\lambda_{2,n}}{b_n}\left(\left\Vert \bm{B}_{j_0}^\star \right\Vert_1 -\left\Vert \widehat{\bm{L}}_{j} \right\Vert_1 \right) \nonumber\\
 & \leq \frac{3n\lambda_{2,n}}{2b_n}
\left\Vert \bm{B}_{j_0}^\star - \widehat{\bm{L}}_{j}\right\Vert_{1,\mathcal{I}}
 -\frac{n\lambda_{2,n}}{2b_n}
\left\Vert \bm{B}_{j_0}^\star - \widehat{\bm{L}}_{j}\right\Vert_{1,\mathcal{I}^c}
\label{eq:lemma_1_2} 
\end{align}}

Based on \eqref{eq:lemma_1_1} and \eqref{eq:lemma_1_2}, we have 
{\footnotesize\begin{align}  
\left\Vert \bm{B}_{j_0}^\star - \widehat{\bm{L}}_{j}\right\Vert_{1}  \leq 4
\left\Vert \bm{B}_{j_0}^\star - \widehat{\bm{L}}_{j}\right\Vert_{1,\mathcal{I}}\leq 4
\sqrt{d^\star_n}\left\Vert \bm{B}_{j_0}^\star - \widehat{\bm{L}}_{j}\right\Vert_{F},\nonumber
\end{align} }
which leads to
{\footnotesize\begin{align}\label{ineq_l1_l2} 
\left\Vert \bm{B}_{j_0}^\star - \widehat{\bm{L}}_{j}\right\Vert_{1}^2  \leq 16
d^\star_n\left\Vert \bm{B}_{j_0}^\star - \widehat{\bm{L}}_{j}\right\Vert_{F}^2. 
\end{align} }

Combine \eqref{ineq_l1_l2} with  the restricted eigenvalue condition in \eqref{eq:RE_bound} and the fact that ${\ d^\star_n \frac{\log (p_x  p_y\vee n)}{b_n} \rightarrow 0}$, 
{there exist constants $\alpha, \tau >0$ such that   }
{\footnotesize\begin{equation} 0 
\leq {\alpha \left\Vert \bm{B}_{j_0}^\star - \widehat{\bm{L}}_{j}\right\Vert_F^2 - \tau  \left\Vert \bm{B}_{j_0}^\star - \widehat{\bm{L}}_{j}\right\Vert_1^2 }
\leq  \frac{1}{\widehat{t}_j - s_{j_0-1}\vee \widehat{t}_{j-1}} \sum_{l= s_{j_0-1}\vee \widehat{t}_{j-1} }^{\widehat{t}_j - 1} \bm{x}^{\prime}_{l}\left(\bm{B}_{j_0}^\star - \widehat{\bm{L}}_{j}\right)^{\prime}\left(\bm{B}_{j_0}^\star - \widehat{\bm{L}}_{j}\right)\bm{x}_{l}.
\label{eq:lemma_2_RE}
\end{equation} }

Combine \eqref{eq:lemma_1_1}, \eqref{eq:lemma_1_2}, \eqref{eq:lemma_2_RE}  with the selection of $\lambda_{2,n}$, 
{there exist a constant $c >0$ such that   }
{\footnotesize\begin{align}  c \left\Vert \bm{B}_{j_0}^\star - \widehat{\bm{L}}_{j}\right\Vert_F^2  &\leq  {\frac{2n\lambda_{2,n}}{b_n}}
\left\Vert \bm{B}_{j_0}^\star - \widehat{\bm{L}}_{j}\right\Vert_{1,\mathcal{I}}
 \nonumber\\
&\leq  {\frac{2n\lambda_{2,n}}{b_n}}
\left\Vert \bm{B}_{j_0}^\star - \widehat{\bm{L}}_{j}\right\Vert_{1}\nonumber\\
&\leq  {\frac{8n\lambda_{2,n}}{b_n}}
\sqrt{d^\star_n}\left\Vert \bm{B}_{j_0}^\star - \widehat{\bm{L}}_{j}\right\Vert_{F}\nonumber
\end{align} }
which leads to
{\footnotesize\begin{align} 
\left\Vert \bm{B}_{j_0}^\star - \widehat{\bm{L}}_{j}\right\Vert_F^2  
\leq  {\frac{64n^2\lambda_{2,n}^2}{c^2b_n^2}}
d^\star_n= {\frac{64C_2^2}{c^2}} \frac{d^\star_n{\log (p_x  p_y \vee n)) }}{b_n}\nonumber
\end{align} }

This implies that
{\footnotesize\begin{equation}
\left\Vert \bm{B}_{j_0}^\star - \widehat{\bm{L}}_{j}\right\Vert_F = { O_p \left(\sqrt{\frac{d^\star_n \log (p_x  p_y \vee n) }{b_n}}\right),}
\label{eq:small_o_1}
\end{equation}}
which means that $\left\Vert \bm{B}_{j_0}^\star - \widehat{\bm{L}}_{j}\right\Vert_F$ converges to zero in probability based on Assumption A3.

Same procedure can be applied to the interval $[\widehat{t}_j , s_{j_0} \wedge \widehat{t}_{j+1}]$ which lead to 
{\footnotesize\begin{equation}
\left\Vert \bm{B}_{j_0}^\star - \widehat{\bm{L}}_{j+1}\right\Vert_F ={  O_p \left(\sqrt{\frac{d^\star_n \log (p_x  p_y \vee n)}{b_n}}\right),}
\label{eq:small_o_2}
\end{equation}}
which means that $\left\Vert \bm{B}_{j_0}^\star - \widehat{\bm{L}}_{j+1}\right\Vert_F$ converges to zero in probability based on Assumption A3.

By triangular inequality, we have
{\footnotesize\begin{align} 
&\left\lVert \widehat{\bm{\Theta}}_{\widehat{t}_j} \right\rVert_F = \left\lVert \widehat{\bm{L}}_{j+1}-\widehat{\bm{L}}_{j} \right\rVert_F 
=\left\lVert \widehat{\bm{L}}_{j+1}-\bm{B}_{j_0}^\star+\bm{B}_{j_0}^\star-\widehat{\bm{L}}_{j} \right\rVert_F  \nonumber\\
\leq &\left\lVert \widehat{\bm{L}}_{j+1} - \bm{B}_{j_0}^\star \right\rVert_F +\left\lVert \widehat{\bm{L}}_{j} -\bm{B}_{j_0}^\star\right\rVert_F \nonumber\\
=  & {\ O_p \left(\sqrt{\frac{d^\star_n\log \left(p_xp_y \vee n \right)}{b_n}}\right).}
\end{align}}
This completes the proof of the lemma.

\end{proof}

\begin{lemma}\label{lemma_lm_theta_2} Suppose A1-A4 hold. Choose {
$\lambda_{1,n} = C_1 \sqrt{{ \log (p_x p_y \vee n)}/{n}} \sqrt{{b_n}/{n}}$
, and 
$\lambda_{2,n} = C_2\sqrt{{\log( p_x p _y \vee n)}/{n}}\sqrt{{b_n}/{n} }$ 
for some large constant $C_1, C_2 > 0$.} Then, for any $t_{j_0}$ in $\mathcal{A}_n$, there exist at least one $\widehat{t}_j$ in $\widehat{\mathcal{A}}_n$ such that $\left\vert \widehat{t}_j - t_{j_0} \right\vert < b_n$ and $\widehat{t}_j = r_i$ for some $i \in \{1,\dots, k_n - 1\}$.
Moreover,
\[{\left\lVert \widehat{\bm{\Theta}}_{i+1}\right\rVert_F \geq \frac{1}{2}\nu_n - O_p \left(\sqrt{\frac{d^\star_n\log (p_x  p_y \vee n) }{b_n}}\right). }\]
\end{lemma}
\begin{proof}

\begin{figure}[ht!]
    \centering
\begin{tikzpicture}
\centering
    \FPeval{\x}{clip(20)}
  \begin{scope}[font=\footnotesize,x=\x pt]
  \foreach \mypt in {0,...,\bound}{
    \FPeval{\result}{clip(\bound+\mypt)}
    \draw(\mypt,0)circle[radius=2pt];
    \draw(-\mypt,0)circle[radius=2pt];
    \node[at={(\mypt,0)},label=below:{$r_{\result}$}]{};
    \ifnum \mypt <1
        \node[at={(\mypt-\bound,0)},label=below:{$r_0$}]{};
    \else
        \node[at={(\mypt-\bound,0)},label=below:{$r_{\mypt}$}]{};
    \fi
  }
  \draw [->] (-\bound-1,0)--(\bound+1,0) node[pos=0,left]{$t$};
  
  \FPeval{\pp}{clip(-1)}
  \fill[gray!50]( \pp * \x pt ,0) circle[radius=2pt];
  \node[at={( \pp * \x pt,0)},above=5pt]{$\widehat{t}_{j}$};
  \FPeval{\ppp}{clip(-4)}
  \fill[gray!50]( \ppp * \x pt ,0) circle[radius=2pt];
  \draw [dashed] ( \ppp * \x pt,-1) -- ( \ppp * \x pt,1);
  
  \FPeval{\ppp}{clip(6)}
  \fill[gray!50]( \ppp * \x pt ,0) circle[radius=2pt];
  \node[at={( \ppp * \x pt,0)},above=5pt]{$\widehat{t}_{j+1}$};
  \draw [<->] (\ppp * \x pt, -0.8)--(\pp * \x pt,-0.8) node[pos= 1/2,below]{$\widehat{\bm{L}}_{j+1}$};
  \draw [dashed, red, thick] ( \pp * \x pt,-1) -- ( \pp * \x pt,1);
  \draw [dashed, red, thick] ( \ppp * \x pt,-1) -- ( \ppp * \x pt,1);

  \FPeval{\pp}{clip(-6)}
  \FPeval{\xx}{clip( \pp * \x + 7 )}
  \FPeval{\dxx}{clip( (\pp + 1) * \x  )}
  \filldraw( \xx pt ,0)circle[radius=2pt];
  \node[at={(\xx pt,0)},above=5pt]{$t_{j_0-1}$};
  \draw [dashed] ( \dxx pt,-1) -- ( \dxx pt,1);

   \FPeval{\pp}{clip(2)}
   \draw [<->] (\dxx pt, 0.8)--(\pp * \x pt,0.8) node[pos= 1/2,above]{$\bm{B}_{j_0}^\star$};
  \FPeval{\xx}{clip( \pp * \x + 10 )}
  \FPeval{\dxx}{clip( \pp * \x  )}
   \filldraw( \xx pt ,0)circle[radius=2pt];
  \node[at={(\xx pt,0)},above=5pt]{$t_{j_0}$};
  \draw [dashed, red, thick] ( \dxx pt,-1) -- ( \dxx pt,1);
  \draw [dashed, red, thick] ( \dxx + \x pt,-1) -- ( \dxx +\x pt,1);
  \FPeval{\pp}{clip(8)}
  \draw [<->] (\dxx + \x pt, 0.8)--( \pp * \x pt ,0.8) node[pos=1/2,above]{$\bm{B}_{j_0+1}^\star$};
  \FPeval{\xx}{clip( \pp * \x + 11 )}
  \FPeval{\dxx}{clip( \pp * \x  )}
  \filldraw( \xx pt ,0)circle[radius=2pt];
  \node[at={(\xx pt,0)},above=5pt]{$t_{j_0+1}$};
  \draw [dashed] ( \dxx pt,-1) -- ( \dxx pt,1);
  
\end{scope}
\end{tikzpicture}
    \caption{Suppose there exists a true change point $t_{j_0}$ which is isolated from all the estimated points, i.e., $\min_{ 1\leq j \leq \widehat{m}}\vert \widehat{t}_j -t_{j_0}\vert \geq b_n $. 
Denote the closest estimated change point to the left side of $t_{j_0}$ by  $\widehat{t}_j$, denote the closest estimated change point  to the right side of $t_{j_0}$ by  $\widehat{t}_{j+1}$.
The idea is to show the estimated coefficient parameter  $\widehat{\bm{L}}_j$  in interval $[t_{j_0-1}\vee \widehat{t}_{j}  , t_{j_0+1} \wedge \widehat{t}_{j+1}]$ converges in $\ell_2$ to both $\bm{B}_{j_0}^\star$ and $\bm{B}_{j_0+1}^\star$, which contradicts Assumption A4. }
    \label{fig:interval_2}
\end{figure}
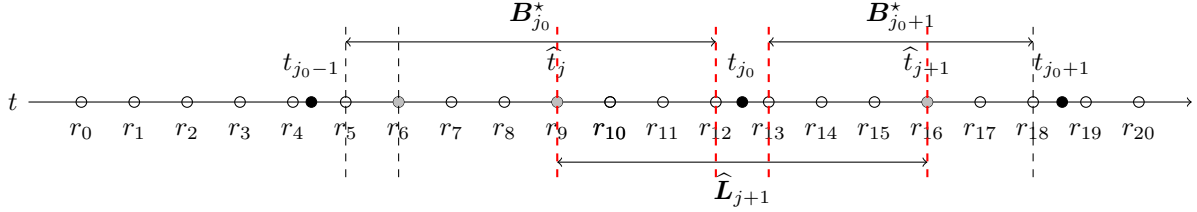

We first focus on the first part. Suppose there exists a true change point $t_{j_0}$ which is isolated from all the estimated points, i.e., $\min_{ 1\leq j \leq \widehat{m}}\vert \widehat{t}_j -t_{j_0}\vert \geq b_n $. 
Denote the closest estimated change point to the left side of $t_{j_0}$ by  $\widehat{t}_j$, denote the closest estimated change point  to the right side of $t_{j_0}$ by  $\widehat{t}_{j+1}$.
By assumption, these two consecutive estimated change points $\widehat{t}_j$ and $\widehat{t}_{j+1}$ are located far away from the true change point $t_{j_0}$,  i.e, $t_{j_0} -t_{j_0-1} \vee \widehat{t}_j \geq b_n$ and $t_{j_0+1}\wedge \widehat{t}_{j+1} -t_{j_0}  \geq b_n$. 
The idea is to show the estimated coefficient parameter  $\widehat{\bm{L}}_{j+1}$  in interval $[t_{j_0-1}\vee \widehat{t}_{j}  , t_{j_0+1} \wedge \widehat{t}_{j+1}]$ converges in $\ell_2$ to both $\bm{B}_{j_0}^\star$ and $\bm{B}_{j_0+1}^\star$, which contradicts Assumption A3.     



Denote the closest $r_i$ to the right side of $t_{j_0-1}$ by  $s_{j_0 -1}$, denote the closest $r_i$ to the left side of $t_{j_0}$ by  $s_{j_0}$, and denote the closest $r_i$ to the right side of $t_{j_0}$ by  $s'_{j_0}$. Note that $s_{j_0 -1}-t_{j_0-1} < b_n $ , $t_{j_0 }-s_{j_0} < b_n $ and $s'_{j_0} -t_{j_0 }  < b_n $. If $\vert t_{j_0} - \widehat{t}_j  \vert = b_n$, then $t_{j_0} = s_{j_0} $ and $\vert s_{j_0} - \widehat{t}_j \vert = b_n$; similarly, if $\vert \widehat{t}_{j+1} - t_{j_0} \vert = b_n$, then $t_{j_0} = s'_{j_0} $ and $\vert \widehat{t}_{j+1} - s'_{j_0} \vert = b_n$. If $t_{j_0} - \widehat{t}_j > b_n$ and $\widehat{t}_{j+1} -t_{j_0} > b_n$, then $\vert t_{j_0} - \widehat{t}_j  \vert \geq 2b_n$ and $\vert\widehat{t}_{j+1} -t_{j_0} \vert \geq 2b_n$, which means that $\vert s_{j_0} - \widehat{t}_j \vert \geq b_n$ and $\vert \widehat{t}_{j+1} -s'_{j_0}\vert \geq b_n$. 
Therefore, we have $s_{j_0} - s_{j_0-1} \vee \widehat{t}_{j} \geq  \min ( s_{j_0} - s_{j_0-1}, s_{j_0} - \widehat{t}_{j}) \geq b_n$ and $s_{j_0+1} \wedge \widehat{t}_{j+1}- s'_{j_0}
 \geq  \min (s_{j_0+1} - s'_{j_0}, \widehat{t}_{j+1} - s'_{j_0} ) \geq b_n$. In other words, the lengths of the intervals $[s_{j_0-1}\vee \widehat{t}_{j}  , s_{j_0}]$ and $[s'_{j_0} , s_{j_0+1} \wedge \widehat{t}_{j+1}]$  are large enough to verify restricted eigenvalue and deviation bound inequalities.  

First, we focus on the interval $[s_{j_0-1}\vee \widehat{t}_{j}  , s_{j_0}]$. Define a new parameter sequence $\bm{\Psi}_k$'s, $k= 1,\dots, k_n$ with $\bm{\Psi}_k = \widehat{\bm{\Theta}}_k$ except for two time points $k=s_{j_0-1}\vee \widehat{t}_{j}$ and $k=s_{j_0}$. If $\widehat{t}_{j} > s_{j_0-1} $, set $\bm{\Psi}_{\widehat{t}_j} = \bm{B}_{j_0}^\star - \widehat{\bm{L}}_{j} $ and $\bm{\Psi}_{s_{j_0}} = \widehat{\bm{L}}_{j+1} -\bm{B}_{j_0}^\star $ ;
if $\widehat{t}_{j} \leq s_{j_0-1} $, set $\bm{\Psi}_{s_{j_0-1}} = \bm{B}_{j_0}^\star - \widehat{\bm{L}}_{j+1} $ and $\bm{\Psi}_{\widehat{t}_j} = \widehat{\bm{L}}_{j+1} -\bm{B}_{j_0}^\star $, 
where $\widehat{\bm{L}}_{j+1} = \sum_{k=1}^{\widehat{i}_{j}}\widehat{\bm{\Theta}}_k$,  $\widehat{i}_j$ is the  corresponding index of candidate point $\widehat{t}_j$ .

By the definition of $\widehat{\bm{\theta}}$ in \eqref{eq:estimation_block}, the value of the function in \eqref{eq:estimation_block} is minimized at $\widehat{\bm{\theta}}$.
Denoting $\bm{\psi} = \text{vec}(\bm{\Psi}_1, \dots, \bm{\Psi}_{k_n}) \in \mathbb{R}^{\pi_n \times 1}$, where $\pi_n = k_np_xp_y$, we have
{\footnotesize\begin{equation} 
\frac{1}{n}\left\Vert \mathbf{y} - Z\widehat{\bm{\theta}}\right\Vert_{2}^{2} + \lambda_{1,n}\left\Vert\widehat{\bm{\theta}}\right\Vert_1 + \lambda_{2,n}\sum_{k=1}^{k_n}\left\Vert\sum_{j=1}^k\widehat{\bm{\Theta}}_j\right\Vert_1 \leq 
\frac{1}{n}\left\Vert \mathbf{y} - Z\bm{\psi}\right\Vert_{2}^{2} + \lambda_{1,n}\left\Vert\bm{\psi}\right\Vert_1 + \lambda_{2,n}\sum_{k=1}^{k_n}\left\Vert\sum_{j=1}^k\bm{\Psi}_j\right\Vert_1
\label{eq:obj_var_2}
\end{equation}}

When $\widehat{t}_{j} > s_{j_0-1} $, some rearrangement of equation \eqref{eq:obj_var_2} leads to

{\footnotesize\begin{align}
0 
&\leq \frac{1}{s_{j_0}- \widehat{t}_{j}} \sum_{l= \widehat{t}_{j} }^{s_{j_0}- 1} \bm{x}^{\prime}_{l}\left(\bm{B}_{j_0}^\star - \widehat{\bm{L}}_{j+1} \right)^\prime\left(\bm{B}_{j_0}^\star - \widehat{\bm{L}}_{j+1}\right)\bm{x}_{l}\nonumber\\
&\leq \frac{2}{s_{j_0}-  \widehat{t}_{j}} \sum_{l= \widehat{t}_{j} }^{s_{j_0}- 1} \bm{x}^{\prime}_{l}\left(\bm{B}_{j_0}^\star - \widehat{\bm{L}}_{j+1}\right)^\prime\bm{\varepsilon}_l\nonumber\\
& + \frac{n\lambda_{1,n}}{s_{j_0}-  \widehat{t}_{j}} 
\left(\left\Vert \bm{B}_{j_0}^\star - \widehat{\bm{L}}_{j+1}\right\Vert_1 + \left\Vert \bm{B}_{j_0}^\star - \widehat{\bm{L}}_{j}\right\Vert_1 -\left\Vert \widehat{\bm{L}}_{j+1}-\widehat{\bm{L}}_{j}\right\Vert_1  \right) 
\nonumber\\
& +  \frac{n\lambda_{2,n}}{b_n}\left(\left\Vert \bm{B}_{j_0}^\star \right\Vert_1 -\left\Vert \widehat{\bm{L}}_{j+1} \right\Vert_1 \right) \nonumber\\
& \leq  \frac{2}{s_{j_0}- \widehat{t}_{j+1}} 
\left\Vert \sum_{l = \widehat{t}_{j} }^{s_{j_0}- 1} \bm{x}_{l} \bm{\varepsilon}_l^\prime  \right\Vert_\infty \left\Vert \bm{B}_{j_0}^\star - \widehat{\bm{L}}_{j+1}\right\Vert_1
 + \frac{2n\lambda_{1,n}}{s_{j_0}- \widehat{t}_{j}} 
\left\Vert \bm{B}_{j_0}^\star - \widehat{\bm{L}}_{j+1}\right\Vert_1
 + \frac{n \lambda_{2,n}}{b_n}\left(\left\Vert \bm{B}_{j_0}^\star \right\Vert_1 -\left\Vert \widehat{\bm{L}}_{j+1} \right\Vert_1 \right ) \nonumber\\
& \leq \left(\frac{2n\lambda_{1,n}}{s_{j_0}- \widehat{t}_{j}} + 
 C \sqrt{\frac{{\log (p_x p_y \vee n) }}{b_n}}\right)
\left\Vert \bm{B}_{j_0}^\star - \widehat{\bm{L}}_{j}\right\Vert_1
 +  \frac{n\lambda_{2,n}}{b_n}\left(\left\Vert \bm{B}_{j_0}^\star \right\Vert_1 -\left\Vert \widehat{\bm{L}}_{j+1} \right\Vert_1 \right ) \nonumber\\
 & \leq  \frac{n\lambda_{2,n}}{2b_n}
\left\Vert \bm{B}_{j_0}^\star - \widehat{\bm{L}}_{j+1}\right\Vert_1
 + \frac{n \lambda_{2,n}}{b_n}\left(\left\Vert \bm{B}_{j_0}^\star \right\Vert_1 -\left\Vert \widehat{\bm{L}}_{j+1} \right\Vert_1 \right) \nonumber\\
 & \leq  \frac{n\lambda_{2,n}}{2b_n}
\left\Vert \bm{B}_{j_0}^\star - \widehat{\bm{L}}_{j+1}\right\Vert_{1,\mathcal{I}} +  \frac{n\lambda_{2,n}}{2b_n}
\left\Vert \bm{B}_{j_0}^\star - \widehat{\bm{L}}_{j+1}\right\Vert_{1,\mathcal{I}^c}
 +\frac{ n \lambda_{2,n}}{b_n}\left( \left\Vert \bm{B}_{j_0}^\star - \widehat{\bm{L}}_{j+1}\right\Vert_{1,\mathcal{I}} -
 \left\Vert \bm{B}_{j_0}^\star - \widehat{\bm{L}}_{j+1}\right\Vert_{1,\mathcal{I}^c}\right) \nonumber\\
 & \leq \frac{3n \lambda_{2,n}}{2b_n}
\left\Vert \bm{B}_{j_0}^\star - \widehat{\bm{L}}_{j+1}\right\Vert_{1,\mathcal{I}}
 - \frac{n\lambda_{2,n}}{2b_n}
\left\Vert \bm{B}_{j_0}^\star - \widehat{\bm{L}}_{j+1}\right\Vert_{1,\mathcal{I}^c}.
  \label{eq:lemma_2_1}
\end{align}}

When $\widehat{t}_{j} \leq s_{j_0-1} $, some rearrangement of equation \eqref{eq:obj_var_2} leads to
{\footnotesize
\begin{align}
0
&\leq \frac{1}{s_{j_0}- s_{j_0-1}} \sum_{l= s_{j_0-1} }^{s_{j_0}- 1} \bm{x}^{\prime}_{l}\left(\bm{B}_{j_0}^\star - \widehat{\bm{L}}_{j+1} \right)^\prime\left(\bm{B}_{j_0}^\star - \widehat{\bm{L}}_{j+1}\right)\bm{x}_{l}\nonumber\\
&\leq \frac{2}{s_{j_0}- s_{j_0-1}} \sum_{l= s_{j_0-1} }^{s_{j_0}- 1}\bm{x}^{\prime}_{l}\left(\bm{B}_{j_0}^\star - \widehat{\bm{L}}_{j+1}\right)^\prime\bm{\varepsilon}_l +
\frac{2n\lambda_{1,n}}{s_{j_0}- s_{j_0-1}} \left\Vert \bm{B}_{j_0}^\star - \widehat{\bm{L}}_{j+1}\right\Vert_1    
 +  \frac{n\lambda_{2,n}}{b_n}\left(\left\Vert \bm{B}_{j_0}^\star \right\Vert_1 -\left\Vert \widehat{\bm{L}}_{j+1} \right\Vert_1 \right) \nonumber\\
&\leq \left(\frac{2n\lambda_{1,n}}{s_{j_0}- s_{j_0-1}} + 
 C \sqrt{\frac{{\log (p_x p_y \vee n) }}{b_n}}\right)
\left\Vert \bm{B}_{j_0}^\star - \widehat{\bm{L}}_{j+1}\right\Vert_1
 + \frac{n\lambda_{2,n}}{b_n}\left(\left\Vert \bm{B}_{j_0}^\star \right\Vert_1 -\left\Vert \widehat{\bm{L}}_{j+1} \right\Vert_1 \right) \nonumber\\
 & \leq \frac{3n\lambda_{2,n}}{2b_n}
\left\Vert \bm{B}_{j_0}^\star - \widehat{\bm{L}}_{j+1}\right\Vert_{1,\mathcal{I}}
 -\frac{n\lambda_{2,n}}{2b_n}
\left\Vert \bm{B}_{j_0}^\star - \widehat{\bm{L}}_{j+1}\right\Vert_{1,\mathcal{I}^c}
\label{eq:lemma_2_2}
\end{align}}

Based on \eqref{eq:lemma_2_1} and \eqref{eq:lemma_2_2}, we have 
{\footnotesize\begin{align}  
\left\Vert \bm{B}_{j_0}^\star - \widehat{\bm{L}}_{j+1}\right\Vert_{1}  \leq 4
\left\Vert \bm{B}_{j_0}^\star - \widehat{\bm{L}}_{j+1}\right\Vert_{1,\mathcal{I}}\leq 4
\sqrt{d^\star_n}\left\Vert \bm{B}_{j_0}^\star - \widehat{\bm{L}}_{j+1}\right\Vert_{F},\nonumber
\end{align} }
which leads to
{\footnotesize\begin{align}\label{ineq_l1_l2_2} 
\left\Vert \bm{B}_{j_0}^\star - \widehat{\bm{L}}_{j+1}\right\Vert_{1}^2  \leq 16
d^\star_n\left\Vert \bm{B}_{j_0}^\star - \widehat{\bm{L}}_{j+1}\right\Vert_{F}^2. 
\end{align} }

Combine \eqref{ineq_l1_l2_2} with  the restricted eigenvalue condition in \eqref{eq:RE_bound} and the fact that {$\ d^\star_n {\log (p_x  p_y \vee n)}/{b_n} \rightarrow 0$}, 
{there exist constants $\alpha, \tau >0$ such that   }
{\footnotesize
\begin{equation} 0 \leq
{\alpha \left\Vert
\bm{B}_{j_0}^\star- \widehat{\bm{L}}_{j+1}\right\Vert_F^2 - \tau \left\Vert
\bm{B}_{j_0}^\star- \widehat{\bm{L}}_{j+1}\right\Vert_1^2 }
\leq  \frac{1}{s_{j_0} -s_{j_0-1}\vee \widehat{t}_{j}  } \sum_{l= s_{j_0-1}\vee \widehat{t}_{j}   }^{s_{j_0} - 1} \bm{x}^{\prime}_{l}\left(\bm{B}_{j_0}^\star - \widehat{\bm{L}}_{j+1}\right)^{\prime}\left(\bm{B}_{j_0}^\star - \widehat{\bm{L}}_{j+1}\right)\bm{x}_{l}.
\label{eq:lemma_2}
\end{equation} }

Combine \eqref{eq:lemma_2_1}, \eqref{eq:lemma_2_2}, \eqref{eq:lemma_2}  with the selection of $\lambda_{2,n}$, 
{there exist a constant $c >0$ such that } 
{\footnotesize\begin{align}  c \left\Vert \bm{B}_{j_0}^\star - \widehat{\bm{L}}_{j+1}\right\Vert_F^2  
&\leq  {\frac{2n\lambda_{2,n}}{b_n}}
\left\Vert \bm{B}_{j_0}^\star - \widehat{\bm{L}}_{j+1}\right\Vert_{1,\mathcal{I}}\nonumber\\
&\leq  {\frac{2n\lambda_{2,n}}{b_n}}
\left\Vert \bm{B}_{j_0}^\star - \widehat{\bm{L}}_{j+1}\right\Vert_{1}\nonumber\\
&\leq  {\frac{8n\lambda_{2,n}}{b_n}}
\sqrt{d^\star_n}\left\Vert \bm{B}_{j_0}^\star - \widehat{\bm{L}}_{j+1}\right\Vert_{F}\nonumber
\end{align} }
which leads to
{\footnotesize\begin{align} 
\left\Vert \bm{B}_{j_0}^\star - \widehat{\bm{L}}_{j+1}\right\Vert_F^2  
\leq  {\frac{64n^2\lambda_{2,n}^2}{c^2b_n^2}}
d^\star_n={ \frac{64C_2^2}{c^2}} \frac{d^\star_n{\log\left( p_x  p_y \vee n\right)} }{b_n}\nonumber
\end{align} }

This implies that
{\footnotesize\begin{equation}
\Vert \bm{B}_{j_0}^\star - \widehat{\bm{L}}_{j+1}\Vert_F = { O_p \Bigg(\sqrt{\frac{d^\star_n\log( p_x p_y \vee n)}{b_n}}\Bigg) },
\label{eq:small_o_3_var}
\end{equation}}
which means that $\Vert \bm{B}_{j_0}^\star - \widehat{\bm{L}}_{j+1}\Vert_F$ converges to zero in probability based on Assumption A4.

Similarly, same procedure can be applied to the interval $[s'_{j_0} , s_{j_0+1} \wedge \widehat{t}_{j+1}]$ which lead to 
\begin{equation}
\Vert \bm{B}_{j_0+1}^\star - \widehat{\bm{L}}_{j+1}\Vert_F =  {O_p \Bigg(\sqrt{\frac{d^\star_n\log (p_x  p_y\vee n)}{b_n}}\Bigg)},
\label{eq:small_o_4_var}
\end{equation}
which means that $\left \Vert \bm{B}_{j_0+1}^\star - \widehat{\bm{L}}_{j+1} \right \Vert_F$ converges to zero in probability based on Assumption A4.

The results in \eqref{eq:small_o_3_var} and \eqref{eq:small_o_4_var}  yield a contradiction to the Assumption A4, and therefore, completes the first part of the proof.

The second part can be proved as follows.
On the one hand, based on the first part of the proof,  for any true change point $t_{j_0} \in \mathcal{A}_n$, there exists at least one estimated change point $\widehat{t_j} \in \widehat{\mathcal{A}}_n$ such that $\left\vert \widehat{t}_j - t_{j_0} \right\vert < b_n$. 
On the other hand, 
for any true change point $t_{j_0} \in \mathcal{A}_n$, 
there exists at most two estimated change point $\widehat{t_j} \in \widehat{\mathcal{A}}_n$ such that $\left\vert \widehat{t}_j - t_{j_0} \right\vert < b_n$. 

Denote the closest $r_i$ to the left side of $t_{j_0-1}$ by  $s_{j_0-1}$, Denote the closest $r_i$ to the left side of $t_{j_0}$ by  $s_{j_0}$, denote the closest $r_i$ to the right side of $t_{j_0}$ by  $s'_{j_0}$. Here, we consider two different cases: (a) there is only one block time point $s_{j_0}$ within the interval $(t_{j_0}- b_n,t_{j_0}+b_n)$; (b) there is two block time point $s_{j_0}$ and $s'_{j_0}$  within the interval $(t_{j_0}- b_n,t_{j_0}+b_n)$.

\textit{Case (a)}. If $t_{j_0} = s_{j_0} =s'_{j_0}$, then there is only one block time point $s_{j_0}$ within the interval $(t_{j_0}- b_n,t_{j_0}+b_n)$. Denote this $s_{j_0}$ by $\widehat{t}_j$. By the first part of the proof, we have $\widehat{t}_{j} \in \widehat{\mathcal{A}}_n$. 

We consider the previous consecutive estimated change point $\widehat{t}_{j-1} \in \widehat{\mathcal{A}}_n$. If $t_{j_0-1} + b_n \leq \widehat{t}_{j-1} \leq  t_{j_0} - b_n $ , by Lemma \ref{lemma_lm_theta_1}, 
$\left\Vert \bm{B}_{j_0}^\star - \widehat{\bm{L}}_{j}\right\Vert_F = { O_p \left(\sqrt{\frac{d^\star_n\log (p_x  p_y \vee n)}{b_n}}\right)}$;  
if  $t_{j_0-1} - b_n < \widehat{t}_{j-1} < t_{j_0-1} + b_n $, 
then $\widehat{t}_j \wedge s_{j_0}- s_{j_0-1} \vee \widehat{t}_{j-1} \geq  b_n$. 
Similar procedure as in Lemma \ref{lemma_lm_theta_1} can be applied to $\left[s_{j_0-1} \vee \widehat{t}_{j-1},\widehat{t}_j \wedge s_{j_0}\right]$ which lead to $\left\Vert \bm{B}_{j_0}^\star - \widehat{\bm{L}}_{j}\right\Vert_F =  {O_p \left(\sqrt{\frac{d^\star_n\log( p_x  p_y \vee n)}{b_n}}\right).}$    

Similarly, consider the following estimated change point $\widehat{t}_{j+1} \in \widehat{\mathcal{A}}_n$. If $t_{j_0} + b_n \leq \widehat{t}_{j+1} \leq  t_{j_0+1} - b_n $ , by Lemma \ref{lemma_lm_theta_1}, 
$\left\Vert \bm{B}_{j_0+1}^\star - \widehat{\bm{L}}_{j+1}\right\Vert_F =  {O_p \left( \sqrt{\frac{d^\star_n\log (p_x  p_y \vee n )}{b_n}}\right)}$. If not, then $t_{j_0+1} - b_n < \widehat{t}_{j+1} < t_{j_0+1} + b_n $. Similar procedure as in Lemma \ref{lemma_lm_theta_1}  can be applied to $\left[\widehat{t}_j \vee t_{j_0}, t_{j_0+1}\wedge \widehat{t}_{j+1}\right]$ which lead to $\left\Vert \bm{B}_{j_0+1}^\star - \widehat{\bm{L}}_{j+1}\right\Vert_F =  {O_p \left(\sqrt{\frac{d^\star_n\log (p_x p_y \vee n)}{b_n}}\right).}$     

By triangular inequality, we have
{\footnotesize
\begin{align} 
&\left\lVert \widehat{\bm{\Theta}}_{i+1} \right\rVert_F =\left\lVert \widehat{\bm{L}}_{j+1}-\widehat{\bm{L}}_{j} \right\rVert_F 
=\left\lVert \widehat{\bm{L}}_{j+1}-\bm{B}_{j_0+1}^\star+\bm{B}_{j_0+1}^\star-\bm{B}_{j_0}^\star+\bm{B}_{j_0}^\star-\widehat{\bm{L}}_{j} \right\rVert_F  \nonumber\\
\geq &\left\lVert \bm{B}_{j_0+1}^\star-\bm{B}_{j_0}^\star\right\rVert_F - \left\lVert\widehat{\bm{L}}_{j+1}-\bm{B}_{j_0+1}^\star+\bm{B}_{j_0}^\star-\widehat{\bm{L}}_{j} \right\rVert_F  \nonumber\\
\geq &\left\lVert \bm{B}_{j_0+1}^\star-\bm{B}_{j_0}^\star\right\rVert_F - \left(\left\lVert \widehat{\bm{L}}_{j+1} - \bm{B}_{j_0+1}^\star \right\rVert_F +\left\lVert \widehat{\bm{L}}_{j} -\bm{B}_{j_0}^\star\right\rVert_F \right) \nonumber\\
\geq & \nu_n -{ O_p \left(\sqrt{\frac{d^\star_n\log (p_x  p_y \vee n)}{b_n}}\right)},
\end{align} }
where $\widehat{t}_j = r_i$.

\textit{Case (b)}. If $s_{j_0} < t_{j_0} < s'_{j_0}$, then there is two block time point $s_{j_0}$ and $s'_{j_0}$  within the interval $(t_{j_0}- b_n,t_{j_0}+b_n)$. By the first part of the lemma, at least one block time point is in  $\widehat{\mathcal{A}}_n$. If there is only one point $\widehat{t}_{j} \in \widehat{\mathcal{A}}_n$ within the interval, we have the same result as with the case (a). If there are two estimated change points within the interval, which are denoted by $\widehat{t}_{j}$ and $\widehat{t}_{j-1}$ , then by triangular inequality and the first part of the proof, we have, 
{\footnotesize
\begin{align} 
&\text{max}\left(\left \lVert \widehat{\bm{\Theta}}_{i+1} \right\rVert_F, \left\lVert \widehat{\bm{\Theta}}_{i} \right\rVert_F \right) \geq  \left\lVert \frac{1}{2} \left(\widehat{\bm{L}}_{j+1}-\widehat{\bm{L}}_{j-1} \right) \right\rVert_F =\frac{1}{2}\left\lVert \widehat{\bm{L}}_{j+1}-\bm{B}_{j_0+1}^\star+\bm{B}_{j_0+1}^\star-\bm{B}_{j_0}^\star+\bm{B}_{j_0}^\star-\widehat{\bm{L}}_{j-1} \right\rVert_F  \nonumber\\
\geq &\frac{1}{2}\left (\left\lVert \bm{B}_{j_0+1}^\star-\bm{B}_{j_0}^\star\right\rVert_F - \left\lVert\widehat{\bm{L}}_{j+1}-\bm{B}_{j_0+1}^\star+\bm{B}_{j_0}^\star-\widehat{\bm{L}}_{j-1} \right\rVert_F \right)  \nonumber\\
\geq &\frac{1}{2}\left (\left\lVert \bm{B}_{j_0+1}^\star-\bm{B}_{j_0}^\star\right\rVert_F - \left(\left\lVert \widehat{\bm{L}}_{j+1} - \bm{B}_{j_0+1}^\star \right\rVert_F +\left\lVert \widehat{\bm{L}}_{j-1} -\bm{B}_{j_0}^\star\right\rVert_F \right)\right) \nonumber\\
\geq & \frac{\nu_n}{2} - { O_p \left(\sqrt{\frac{d^\star_n\log (p_x  p_y \vee n)}{b_n}}\right)},
\end{align}}
where $\widehat{t}_j = r_i$ and $\widehat{t}_{j-1} = r_{i-1}$.

\end{proof}

\section{Proof of Main Results}\label{sec:main_proof}
\setcounter{equation}{0}
\begin{proof}[Proof of Theorem~\ref{thm_selection_block}]
From Lemma \ref{lemma_lm_theta_1}, we know that any points $\widehat{t}_j \in \widehat{\mathcal{A}}_n$ isolated from all true change points will have small jump. In other words, suppose $\widehat{t}_j = r_i$ for some $i \in \{1, \dots, k_n - 1\}$  and  $\min_{j'= 1,\dots,m_0}\left\vert \widehat{t}_j - t_{j'} \right\vert \geq b_n$, we have  $\left\lVert \widehat{\bm{\Theta}}_{i+1}\right\rVert_F = { O_p \left(\sqrt{\frac{d^\star_n\log (p_x  p_y \vee n)}{b_n}}\right)}$. By the definition of $\widetilde{\mathcal{A}}_n$, for any $\widetilde{t}_{j}$ in $\widetilde{\mathcal{A}}_n$, there exist a true change point  $t_{j_0}$ in $\mathcal{A}_n$ such that  $\left\vert \widehat{t}_j - t_{j_0} \right\vert < b_n$.

On the other hand, from Lemma \ref{lemma_lm_theta_2}, we know that
for any $t_{j_0}$ in $\mathcal{A}_n$, there exist an $\widehat{t}_j$ in $\widehat{\mathcal{A}}_n$ such that $\vert \widehat{t}_j - t_{j_0} \vert < b_n$ and $\widehat{t}_j = r_i$ for some $i \in \{1,\dots, k_n - 1\}$, which satisfies 
$\left\lVert \widehat{\bm{\Theta}}_{i+1}\right\rVert_F \geq \frac{1}{2}\nu_n - { O_p \left(\sqrt{\frac{d^\star_n\log (p_x  p_y \vee n)}{b_n}}\right)}. $
Again, by the definition of $\widetilde{\mathcal{A}}_n$, 
for any $t_{j_0}$ in $\mathcal{A}_n$, there exist an $\widetilde{t}_j$ in $\widetilde{\mathcal{A}}_n$ such that $\left \vert \widetilde{t}_j - t_{j_0} \right\vert < b_n$.
Therefore, we have
\[\mathbb{P} \left(d_H \left((\widetilde{\mathcal{A}}_n, \mathcal{A}_n \right) < b_n \right) \to 1. \]
This completes the proof of the first part.

Based on the Lemma \ref{lemma_lm_theta_1} and Lemma \ref{lemma_lm_theta_2}, all the points in $\widetilde{\mathcal{A}}_n$ are very close to a true change points, i.e., in the $b_n$-neighborhood of a true change point. By Lemma \ref{lemma_lm_theta_2}, for any true change point, say $t_{j_0}$, there are at least one estimated change point in $\widetilde{\mathcal{A}}_n$  in the $b_n$-neighborhood of $t_{j_0}$. On the other hand, by the setting of the block fused lasso model, there are at most two estimated points in the $b_n$-neighborhood of $t_{j_0}$. This proves that $\mathbb{P}\left (m_0 \leq \left \vert \widetilde{\mathcal{A}}_n  \right\vert \leq 2m_0 \right) \rightarrow 1.$ 

Denote all the selected change points in $\widetilde{\mathcal{A}}_n$  in the $b_n$-neighborhood of $t_{j}$ by $R_{j}$, $j = 1, \dots, \widetilde{m}^f $. Note that the diameter of each $R_j$ is at most $b_n$. Based on the Assumption A3, all $R_j$'s are disjoint. Therefore, the collection of all $R_j$'s  form $\mbox{cluster}(\widetilde{\mathcal{A}}_n)$ which has cardinality equal to $m_0$. This proves that $\mathbb{P} ( \widetilde{m}^f  = m_0) \rightarrow 1$ and completes the proof of the last part.  
\end{proof}

\begin{proof}[Proof of Theorem~\ref{thm:final_consistency_rate}]

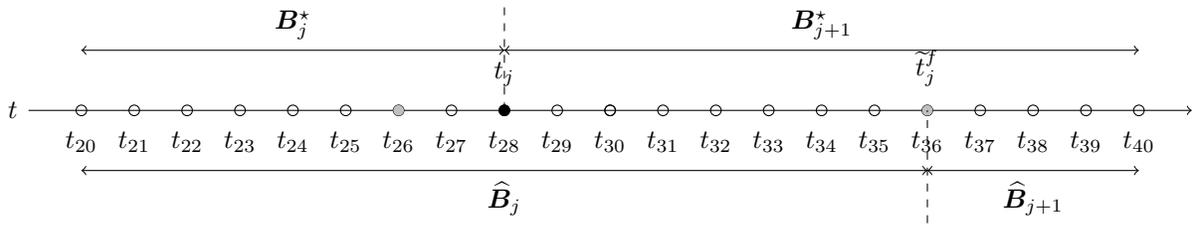
\begin{figure}[ht!]
    \centering
\begin{tikzpicture}
\centering
    \FPeval{\x}{clip(20)}
  \begin{scope}[font=\footnotesize,x=\x pt]
  \foreach \mypt in {0,...,\bound}{
    \FPeval{\result}{clip(\bound+\mypt + 20)}
    \draw(\mypt,0)circle[radius=2pt];
    \draw(-\mypt,0)circle[radius=2pt];
    \ifnum \mypt >0
        \node[at={(\mypt,0)},label=below:{$t_{\result}$}]{};
    \fi
    \FPeval{\result}{clip(\mypt + 20)}
    \node[at={(\mypt-\bound,0)},label=below:{$t_{\result}$}]{};
  }
  \draw [->] (-\bound-1,0)--(\bound+1,0) node[pos=0,left]{$t$};
  
  \FPeval{\pp}{clip(6)}
  \fill[gray!50]( \pp * \x pt ,0) circle[radius=2pt];
  \node[at={( \pp * \x pt,0)},above=5pt]{$\widetilde{t}^{f}_j$};
  \FPeval{\ppp}{clip(-4)}
  \fill[gray!50]( \ppp * \x pt ,0) circle[radius=2pt];
  
  \FPeval{\ppp}{clip(10)}
  \draw [<->] (\ppp * \x pt, -0.8)--(\pp * \x pt,-0.8) node[pos= 1/2,below]{$\widehat{\bm{B}}_{j+1}$};
  \draw [dashed] ( \pp * \x pt,-1.5) -- ( \pp * \x pt,0);
  \FPeval{\ppp}{clip(-10)}
\draw [<->] (\ppp * \x  pt, -0.8)--(\pp * \x pt,-0.8) node[pos= 1/2,below]{$\widehat{\bm{B}}_{j}$};

  \FPeval{\pp}{clip(-10)}
  \FPeval{\xx}{clip( \pp * \x  )}
  \FPeval{\dxx}{clip( (\pp) * \x  )}
 
   \FPeval{\pp}{clip(-2)}
   \draw [<->] (\dxx pt, 0.8)--(\pp * \x pt,0.8) node[pos= 1/2,above]{$\bm{B}_{j}^\star$};
  \FPeval{\xx}{clip( \pp * \x  )}
  \FPeval{\dxx}{clip( \pp * \x  )}
   \filldraw( \xx pt ,0)circle[radius=2pt];
  \node[at={(\xx pt,0)},above=5pt]{$t_{j}$};
  \draw [dashed] ( \dxx pt,0) -- ( \dxx pt,1.5);

  \FPeval{\pp}{clip(10)}
  \draw [<->] (\dxx  pt, 0.8)--( \pp * \x pt ,0.8) node[pos=1/2,above]{$\bm{B}_{j+1}^\star$};
  \FPeval{\xx}{clip( \pp * \x + 11 )}
  \FPeval{\dxx}{clip( \pp * \x  )}

\end{scope}
\end{tikzpicture}
    \caption{Here, we assume that $\widetilde{t}^{f}_j > t_j$, $b_n = 10$, $R_{j} = \{31\}$, the search domain is $(\min (R_{j})- b_n, \max (R_{j}) + b_n)= (21, 41)$. }
    \label{fig:interval_3}
\end{figure}

Suppose for any constant $K>0$, there exist some change point $t_j \in \mathcal{A}_n$ such that $ \left|\widetilde{t}^{f}_j - t_j \right| >{\frac{K {d_n^\star} \log (p_x  p_y \vee n)}{\nu_n^2}.}$
From Theorem \ref{thm_selection_block}, we know that for any estimated  change point $\widetilde{t}_{j'} \in \widetilde{\mathcal{A}}_n $, there exists a true change point $t_j$ lies in $(\widetilde{t}_{j'} - b_n, \widetilde{t}_{j'} + b_n )$. Therefore, there exists a true change point $t_j$ lies within interval $[\min (R_j) - b_n, \max (R_j) + b_n )$.  Here, without loss of generality,  we assume that $\widetilde{t}^{f}_j > t_j$. 

Based on exhaustive search algorithm, we  define the loss function $L_n(\widetilde{t}^{f}_j)$ as follow:
{\footnotesize\[
 \sum_{t= \min(R_j)-b_n}^{\widetilde{t}^{f}_j-1}\left\| \bm{y}_{t}  - \widehat{\bm{B}}_{j}\bm{x}_{t} \right\|_2^2 +  \sum_{t= \widetilde{t}^{f}_j}^{\max (R_j)+b_n-1} \left\| \bm{y}_{t}  - \widehat{\bm{B}}_{j+1}\bm{x}_{t}\right \|_2^2   \overset{\text{def}}{=} I_1 + I_2,
\]}
and the loss function $L_n(t_j)$ of true change point $t_j$  as follow:
{\footnotesize\[\sum_{t= \min (R_j) -b_n}^{t_j-1}\| \bm{y}_{t}  - \widehat{\bm{B}}_{j}\bm{x}_{t} \|_2^2 +  \sum_{t= t_j}^{\max (R_j)+b_n-1} \| \bm{y}_{t}  - \widehat{\bm{B}}_{j+1}\bm{x}_{t} \|_2^2 \overset{\text{def}}{=} I_3 + I_4.\]}

 where  $\widetilde{t}^{f}_j \in (l_j, u_j)$, $l_j = (R_j-b_n) \mathbbm{1}_{\{\vert R_j\vert = 1\}} + \min (R_j) \mathbbm{1}_{\{\vert R_j\vert > 1\}}$ and $u_j = (R_j+b_n) \mathbbm{1}_{\{\vert R_j\vert = 1\}} + \max (R_j) \mathbbm{1}_{\{\vert R_j\vert > 1\}}$;
$ {\widehat{\bm{B}}}_{j} = \sum_{k=1}^{{\lfloor\frac{1}{2}\left(\max(J_{j-1})+\min(J_{j})\right)\rfloor}} \widehat{\bm{\Theta}}_k,  \text{ for } j = 1, \dots, m_0 + 1.$ $\widehat{\bm{B}}_{j}$ are the local coefficient parameter estimates for the $j$-th segment, derived from the threshold block fused lasso step.


{\footnotesize
\begin{align}
        I_1
        &= \sum_{t=\min (R_j)-b_n}^{t_j-1}\left\|\bm{y}_{t}  - \widehat{\bm{B}}_{j}\bm{x}_{t}\right\|_2^2 + \sum_{t=t_j}^{\widetilde{t}^{f}_j-1}\left\|\bm{y}_{t}  - \widehat{\bm{B}}_{j}\bm{x}_{t}\right\|_2^2 \nonumber \\  &\geq  I_3 
        + \sum_{t=t_j}^{\widetilde{t}^{f}_j-1}\|\bm{\varepsilon}_t\|_2^2
        + \sum_{t=t_j}^{\widetilde{t}^{f}_j -1}\left\|  \left(\widehat{\bm{B}}_{j}  - \bm{B}_{j+1}^\star \right) \bm{x}_{t}\right\|_2^2 
        - 2\left| \sum_{t=t_j}^{\widetilde{t}^{f}_j-1}\bm{x}^{\prime}_{t}\left(\widehat{\bm{B}}_{j} - \bm{B}_{j+1}^\star\right)\bm{\varepsilon}_t\right|\nonumber \\
        &=  I_3 
        + \sum_{t=t_j}^{\widetilde{t}^{f}_j-1}\|\bm{\varepsilon}_t\|_2^2
        + \sum_{t=t_j}^{\widetilde{t}^{f}_j -1}\left\|  \left(\widehat{\bm{B}}_{j}  - \bm{B}_{j}^\star +   \bm{B}_{j}^\star - \bm{B}_{j+1}^\star \right) \bm{x}_{t}\right\|_2^2 
        - 2\left| \sum_{t=t_j}^{\widetilde{t}^{f}_j-1}\bm{x}^{\prime}_{t}\left(\widehat{\bm{B}}_{j}  -\bm{B}_{j}^\star + \bm{B}_{j}^\star-  \bm{B}_{j+1}^\star\right)\bm{\varepsilon}_t\right|\nonumber \\
        &\geq   I_3 
        + \sum_{t=t_j}^{\widetilde{t}^{f}_j-1}\|\bm{\varepsilon}_t\|_2^2
        + \sum_{t=t_j}^{\widetilde{t}^{f}_j -1}\left\|  \left(\widehat{\bm{B}}_{j}  - \bm{B}_{j}^\star 
        \right) \bm{x}_{t}\right\|_2^2 
        + \sum_{t=t_j}^{\widetilde{t}^{f}_j -1}\left\|  \left(
         \bm{B}_{j}^\star - \bm{B}_{j+1}^\star \right) \bm{x}_{t}\right\|_2^2
         \nonumber\\
         & \quad \quad  \quad - 2\left|\sum_{t=t_j}^{\widetilde{t}^{f}_j -1}\bm{x}^\prime_{t} \left(\widehat{\bm{B}}_{j}  - \bm{B}_{j}^\star 
         \right)^\prime  \left( \bm{B}_{j}^\star - \bm{B}_{j+1}^\star \right) \bm{x}_{t} \right | 
       - 2\left| \sum_{t=t_j}^{\widetilde{t}^{f}_j-1}\bm{x}^{\prime}_{t}\left(\widehat{\bm{B}}_{j}  -\bm{B}_{j}^\star \right)\bm{\varepsilon}_t\right|
       - 2\left| \sum_{t=t_j}^{\widetilde{t}^{f}_j-1}\bm{x}^{\prime}_{t}\left( \bm{B}_{j}^\star-  \bm{B}_{j+1}^\star\right)\bm{\varepsilon}_t\right|\nonumber \\
       &\overset{\text{(i)}}{\geq}  c^\prime  \left\vert \widetilde{t}^{f}_j - t_j \right\vert \left(\left\|\widehat{\bm{B}}_{j} - \bm{B}_{j}^\star \right\|_F^2 + \left\| \bm{B}_{j}^\star - \bm{B}_{j+1}^\star \right\|_F^2 \right)
         - c^{\prime\prime} \left\vert \widetilde{t}^{f}_j - t_j \right\vert\left \Vert \bm{B}_{j}^\star - \bm{B}_{j+1}^\star\right \Vert_{F} \left \Vert\widehat{\bm{B}}_{j} - \bm{B}_{j}^\star\right \Vert_F\nonumber\\
       & \quad \quad \quad   - c^{\prime\prime\prime}\sqrt{\left\vert \widetilde{t}^{f}_j - t_j \right\vert {\log (p_x  p_y \vee n)} }\left(\left\|   \widehat{\bm{B}}_{j}- \bm{B}_{j}^\star\right\|_1 + \left\|   \bm{B}_{j}^\star- \bm{B}_{j+1}^\star\right\|_1 \right) + I_3 
        + \sum_{t=t_j}^{\widetilde{t}^{f}_j-1}\|\bm{\varepsilon}_t\|_2^2
       \nonumber \\
       &\overset{\text{(ii)}}{\geq}   I_3 
        + \sum_{t=t_j}^{\widetilde{t}^{f}_j-1}\|\bm{\varepsilon}_t\|_2^2 + c^\prime  \left\vert \widetilde{t}^{f}_j - t_j \right\vert \left\|   \bm{B}_{j}^\star- \bm{B}_{j+1}^\star\right\|_F^2 
        - c^{\prime\prime} \left\vert \widetilde{t}^{f}_j - t_j \right\vert\left \Vert \bm{B}_{j}^\star - \bm{B}_{j+1}^\star\right \Vert_{F} \left \Vert\widehat{\bm{B}}_{j} - \bm{B}_{j}^\star\right \Vert_F \nonumber \\
       & \quad \quad \quad \quad    -c^{\prime\prime\prime}\sqrt{\left\vert \widetilde{t}^{f}_j - t_j \right\vert {\log( p_x  p_y \vee n)} } \left\|   \bm{B}_{j}^\star- \bm{B}_{j+1}^\star\right\|_1 
       \nonumber \\
       &\overset{}{\geq}   I_3 
        + \sum_{t=t_j}^{\widetilde{t}^{f}_j-1}\|\bm{\varepsilon}_t\|_2^2  
        + C^\prime \left\vert \widetilde{t}^{f}_j - t_j \right\vert\left \Vert \bm{B}_{j}^\star - \bm{B}_{j+1}^\star\right \Vert_{F} \left( \left \Vert \bm{B}_{j}^\star - \bm{B}_{j+1}^\star\right \Vert_{F}- \left \Vert\widehat{\bm{B}}_{j} - \bm{B}_{j}^\star\right \Vert_F  \right)\nonumber \\
       & \quad \quad \quad \quad    -c^{\prime\prime\prime}\sqrt{\left\vert \widetilde{t}^{f}_j - t_j \right\vert {\log (p_x  p_y \vee n)} } \left\|   \bm{B}_{j}^\star- \bm{B}_{j+1}^\star\right\|_1 
       \nonumber \\
       &\overset{\text{(iii)}}{\geq}  I_3 
        + \sum_{t=t_j}^{\widetilde{t}^{f}_j-1}\|\bm{\varepsilon}_t\|_2^2
        + c_1 \left\vert \widetilde{t}^{f}_j - t_j \right\vert \left\|    \bm{B}_{j}^\star- \bm{B}_{j+1}^\star\right\|_F \left( \left\|   \bm{B}_{j}^\star- \bm{B}_{j+1}^\star\right\|_F - \sqrt{\frac{d_n^\star {\log (p_x p_y \vee n) }}{\left\vert \widetilde{t}^{f}_j - t_j \right\vert}} \right) \nonumber \\
       &\geq  I_3 
        + \sum_{t=t_j}^{\widetilde{t}^{f}_j-1}\|\bm{\varepsilon}_t\|_2^2
        + C \left\vert \widetilde{t}^{f}_j - t_j \right\vert \left\|    \bm{B}_{j}^\star- \bm{B}_{j+1}^\star\right\|_F \left(\left\|    \bm{B}_{j}^\star- \bm{B}_{j+1}^\star\right\|_F
        -\nu_n\sqrt{\frac{1}{K}}   \right)
       \nonumber\\
         &\overset{\text{(iv)}}{\geq} I_3+  \sum_{t=t_j}^{\widetilde{t}^{f}_j-1}\|\bm{\varepsilon}_t\|_2^2
        + K_1\left\vert \widetilde{t}^{f}_j - t_j \right\vert\nu_n^2 ,
    \end{align}}
   where the first term in inequality (i) holds  by {lower}-RE condition in {\eqref{eq:RE_bound}} and the fact that
{\footnotesize\begin{align}
\left\Vert \bm{B}_{j}^\star - \widehat{\bm{B}}_{j} \right\Vert_{1}^2  \leq 16
d^\star_n\left\Vert \bm{B}_{j}^\star - \widehat{\bm{B}}_{j}\right\Vert_{F}^2;  \nonumber
\end{align} } 
the  second term inequality (i) holds  by  H\"{o}lder's inequality;
the  third term inequality (i) holds  by deviation bound condition in \eqref{eq:dev_bound};  the third term in inequality (ii) holds  by the fact that $ \widehat{\bm{B}}_{j}$ converges to $\bm{B}_{j}^\star$, $\left\| \bm{B}_{j+1}^\star -\bm{B}_{j}^\star \right\|_F^2 \geq \nu_n^2 $ and Assumption A4 that $\nu_n = {\Omega\left(\sqrt{\frac{d^\star_n\log (p_x p_y \vee n)}{b_n}}\right)}$;  
the inequality (iii) holds  by  $\left\lVert \bm{B}_{j+1}^\star - \bm{B}_{j}^\star \right\rVert_{1} \leq  \sqrt{2d_n^\star}\left\lVert \bm{B}_{j+1}^\star - \bm{B}_{j}^\star \right\rVert_{F}$;   
and the inequality (iv) holds by ${ \left|\widetilde{t}^{f}_j - t_j \right| >\frac{K {d_n^\star} \log( p_x  p_y \vee n)}{\nu_n^2}}$ and choosing large enough constant $K>0$ such that $\sqrt{\frac{1}{K}} < \frac{1}{4}$.

Similarly, we have
{\footnotesize
  \begin{align}
        I_4
        &= \sum_{t= t_j}^{\max (R_j)+b_n-1} \left\| \bm{y}_{t}  - \widehat{\bm{B}}_{j+1}\bm{x}_{t} \right\|_2^2\nonumber \\
        &= \sum_{t= t_j}^{\widetilde{t}^{f}_j-1} \left\| \bm{y}_{t}  - \widehat{\bm{B}}_{j+1}\bm{x}_{t} \right\|_2^2
        + \sum_{t= \widetilde{t}^{f}_j}^{\max (R_j)+b_n-1} \left\| \bm{y}_{t}  - \widehat{\bm{B}}_{j+1}\bm{x}_{t} \right\|_2^2\nonumber \\
        &\leq I_2 
        + \sum_{t=t_j}^{\widetilde{t}^{f}_j-1}\|\bm{\varepsilon}_t\|_2^2
        + \sum_{t=t_j}^{\widetilde{t}^{f}_j-1}
        \left\|  \left(\widehat{\bm{B}}_{j+1} - \bm{B}_{j+1}^\star \right) \bm{x}_{t}\right\|_2^2 
        + 2\left| \sum_{t=t_j}^{\widetilde{t}^{f}_j-1}\bm{x}^{\prime}_{t}\left(\widehat{\bm{B}}_{j+1} - \bm{B}_{j+1}^\star\right)\bm{\varepsilon}_t\right| \nonumber \\
        &\overset{\text{(i)}}\leq  I_2
        + \sum_{t=t_j}^{\widetilde{t}^{f}_j-1}\|\bm{\varepsilon}_t\|_2^2
        + c^\prime b_n \left\|\widehat{\bm{B}}_{j+1} - \bm{B}_{j+1}^\star \right\|_F^2 
        + c^{\prime\prime}\sqrt{b_n {\log( p_x  p_y\vee n) }}\left\|\widehat{\bm{B}}_{j+1} - \bm{B}_{j+1}^\star\right\|_1  \nonumber \\
        &\leq I_2 
        + \sum_{t=t_j}^{\widetilde{t}^{f}_j-1}\|\bm{\varepsilon}_t\|_2^2
        + c^\prime b_n \left\|\widehat{\bm{B}}_{j+1} - \bm{B}_{j+1}^\star \right\|_F\left( \left\|\widehat{\bm{B}}_{j+1} - \bm{B}_{j+1}^\star \right\|_F + \frac{c^{\prime\prime}}{c} \sqrt{\frac{ d_{j+1} {\log( p_x p_y \vee n) }}{b_n}}\right) \nonumber \\
        &\overset{\text{(ii)}}\leq I_2 
        + \sum_{t=t_j}^{\widetilde{t}^{f}_j-1}\|\bm{\varepsilon}_t\|_2^2
        +{ K_2d_{j+1}\log( p_x p_y \vee n) },\nonumber
    \end{align}}
where the second term in inequality (i) holds  by upper-RE condition in \eqref{eq:upper-RE_bound} and the fact that
{\footnotesize\begin{align}
\left\Vert \bm{B}_{j+1}^\star - \widehat{\bm{B}}_{j+1} \right\Vert_{1}^2  \leq 16
d^\star_n\left\Vert \bm{B}_{j+1}^\star - \widehat{\bm{B}}_{j+1}\right\Vert_{F}^2;  \nonumber
\end{align} }
the third term in inequality (i) holds  by deviation bound condition in \eqref{eq:dev_bound}; the inequality (ii) holds  by the result \eqref{eq:small_o_3_var} and  \eqref{eq:small_o_4_var}  in  Lemma 2. 

Now based on the definition of \eqref{eq:Exhaust:final}, we have $ L_n\left(\widetilde{t}_j^f \right)  \leq L_n(t_j) $ and then
{\footnotesize\begin{equation}
        I_3+  \sum_{t=t_j}^{\widetilde{t}^{f}_j-1}\|\bm{\varepsilon}_t\|_2^2
        + K_1\nu_n^2 \left|\widetilde{t}^{f}_j-t_j \right| + I_2 \leq L_n(\widetilde{t}_j^f)  \leq L_n(t_j) \leq  I_3 + I_2 
        + \sum_{t=t_j}^{\widetilde{t}^{f}_j-1}\|\bm{\varepsilon}_t\|_2^2
        + {K_2d_{j+1}\log( p_xp_y \vee n)}, \nonumber
    \end{equation}}
which leads to
{\footnotesize\begin{equation}
        { \left |\widetilde{t}_j^f -t_j \right|  \leq  \frac{K_j d_n^\star  \log( p_x p_y \vee n)}{\nu_n^2} .}\nonumber
\end{equation}}
 This contradicts the setting. Under the Assumption A1 and A2, we set $K^\star = \max_{1\leq j \leq m_0} K_j$ and complete the proof.    
\end{proof}

\begin{proof}[proof of Theorem~\ref{thm:estimation}]

Let $\widehat{\bm{B}}_{j} = \sum_{i=1}^{{\big\lfloor\frac{1}{2}\left(\max(J_{j-1})+\min(J_{j})\right)\big\rfloor}} \widehat{\bm{\Theta}}_i$, 
where $\max(J_{j-1})$ and $\min(J_{j})$ are the  corresponding block indices  for  change point clusters $R_{j-1}$ and $R_j$.
By Theorem~\ref{thm_selection_block}, we have
{\footnotesize\begin{equation*}
{
 \left\lfloor\frac{1}{2}\left(\max(J_{j-1 })+\min(J_{j})\right) \right\rfloor- \max(J_{j-1 })  \geq   \frac{\min_{1\leq j \leq m_0+1} \vert t_{j} - t_{j-1}\vert - 2b_n }{2}  \geq b_n. }
\end{equation*}}

Apply the result from Lemma~\ref{lemma_lm_theta_1},  we have 
{\footnotesize\begin{equation*}
   \left\Vert  \widehat{\bm{B}}_j - \bm{B}_j^\star \right\Vert_F = { O_p \left(\sqrt{\frac{d^\star_n\log (p_x  p_y \vee n)}{b_n}}\right) }
    \end{equation*}}
which complete the proof of the first part.


Let $S$ denote the support of $\bm{B}_j^\star$.
To derive the upper bound on the number of false positives selected by thresholded lasso, note that 
{\footnotesize
\begin{equation*}
    \left \vert  \text{supp}(\widetilde{\bm{B}}_j) \backslash  \text{supp}(\bm{B}_j^\star) \right \vert  = \sum_{s \notin S}\mathbbm{1}_{\{ \vert \widehat{\bm{B}}_{j,s} \vert > {\eta_n} \}} \leq \sum_{s \notin S} \vert \widehat{\bm{B}}_{j, s} \vert/{\eta_{n,j}} \leq \frac{1}{{\eta_{n,j}}} \sum_{s \notin S} \vert v_s \vert \leq
    \frac{3}{{\eta_{n,j} }} \sum_{s \in S} \vert v_s \vert 
    \leq
    \frac{3\Vert v \Vert_1}{{\eta_{n,j}}} \leq
    {\frac{12 \sqrt{d_n^\star} \Vert v \Vert_F}{{\eta_{n,j}}} } .
    \end{equation*}}where $v = \widehat{\bm{B}}_j - \bm{B}_j^\star$.
\end{proof}

\begin{proof}[Proof of Theorem~\ref{thm:mean}]
We only need to verify the general Assumptions A1-A2 hold, under Assumption B1.
Specifically, we want to prove that there exist constants $c_1, c_2 > 0 $ such that for any 
${a_n = \Omega( \log (p \vee n ))} $,
{\footnotesize\begin{equation}\label{eq:DB_bound_mean}
\sup_{1 \leq j \leq {m_0+1},  t_j > u > l \geq t_{j-1}, | u - l  |  > a_n }  \left|\left| {(l-u)}^{-1}  \sum_{t=l}^{u - 1} \bm{\varepsilon}_t^\prime    \right|\right|_\infty \leq { c_0\sigma_{\bm{\varepsilon}}^2 \sqrt{\frac{\log( p\vee n)}{a_n}},}
\end{equation}}
with probability at least ${1- c_1 \exp( -c_2 \log( p \vee n))}$.

Note that for mean model,  the $\bm{x}_t = 1$ in each segment. Therefore, it is easy to verify that Assumption A1 holds. Specifically, with any $0< \alpha_1 \leq1$ and $\alpha_2 > 1$,
{\footnotesize\begin{equation}
\inf_{1 \leq j \leq {m_0+1}, t_j > u > l \geq t_{j-1}, | u - l  |  > a_n  }  v^\prime \left({(l-u)}^{-1}  \sum_{t=l}^{u - 1} \bm{x}_{t}\bm{x}_{t}^\prime \right) v = \Vert v \Vert_2^2  \geq
\alpha_1 \Vert v \Vert_2^2, \nonumber
\end{equation}}
and
{\footnotesize\begin{equation}
\sup_{1 \leq j \leq {m_0+1}, t_j > u > l \geq t_{j-1}, | u - l  |  > a_n  }  v^\prime \left({(l-u)}^{-1}  \sum_{t=l}^{u - 1} \bm{x}_{t}\bm{x}_{t}^\prime \right) v = \Vert v \Vert_2^2  \leq
\alpha_2 \Vert v \Vert_2^2. \nonumber
\end{equation}}

Using the fact that
{\footnotesize\begin{equation}
  \left|\left| {(l-u)}^{-1}  \sum_{t=l}^{u - 1} \bm{x}_{t} \bm{\varepsilon}_t^\prime    \right|\right|_\infty  =  \left|\left| {(l-u)}^{-1}  \sum_{t=l}^{u - 1} \bm{\varepsilon}_t   \right|\right|_\infty = \max_{1\leq i \leq p}  \left \{ {(l-u)}^{-1}  \sum_{t=l}^{u - 1} e_i^\prime\bm{\varepsilon}_t  \right\}, \nonumber
\end{equation}}
where $e_i \in \R^{p}$ is a unit vector with the $i$-th element being one and the rest zero, we may  first consider the deviation bound for ${(l-u)}^{-1}  \sum_{t=l}^{u - 1} e_i^\prime\bm{\varepsilon}_t$. 

Let  $Z = ( e_i^\prime\varepsilon_l, e_i^\prime\varepsilon_{l+1}, \dots,  e_i^\prime\varepsilon_{u-1})^\prime$ and
apply Proposition 5.10 (Hoeffding-type inequality) in \cite{vershynin2010introduction}. Given the fact that
$e_i^\prime\bm{\varepsilon}_t$ is sub-Gaussian with parameter at most $\sigma_{\bm{\varepsilon}} = \max_{1 \leq j \leq m_0 + 1} \sigma_j$, we have
{\footnotesize\begin{equation}
\mathbb{P} \left( \left \vert  \left \{ {(l-u)}^{-1}  \sum_{t=l}^{u - 1} e_i^\prime\bm{\varepsilon}_t  \right\} \right\vert  > t \right)
= \mathbb{P} \left( \left|  \frac{1}{u-l}\sum_{i=1}^{u-l}Z_{i} \right| > t_1 \right) 
\leq 2 \exp (-c a_n  \frac{t}{2\sigma_{\bm{\varepsilon}}^2} ).\nonumber
\end{equation}}

Taking union over all $1 \leq i \leq p$, we have
{\footnotesize\begin{equation}
  \mathbb{P} \left( \left \Vert {(l-u)}^{-1}  \sum_{t=l}^{u - 1}\bm{\varepsilon}_t  \right\Vert_\infty > t \right) \leq  2p  \exp \left(-c a_n \frac{t}{2\sigma_{\bm{\varepsilon}}^2} \right). \nonumber
\end{equation}}

Setting ${t = c_0\sigma_{\bm{\varepsilon}}^2 \sqrt{\frac{\log (p\vee n)}{a_n}}}$ and the assumption that 
${a_n = \Omega( \log (p\vee n))} $,
we have
{\footnotesize\begin{align}
 &\mathbb{P} \left( \left \Vert {(l-u)}^{-1}  \sum_{t=l}^{u - 1} \bm{x}_{t}\varepsilon_t  \right\Vert_\infty > c_0\sigma_{\bm{\varepsilon}}^2 \sqrt{\frac{\log (p\vee n)}{a_n}} \right)  \nonumber\\
 \leq  & c_1  \exp \left (-c a_n  \frac{c_0\sqrt{\log (p\vee n) }}{  \sqrt{a_n}} \right)    \nonumber \\
 \leq& c_1  \exp (-c_2 \log (p\vee n) ).  \nonumber
\end{align}}
This completes the proof.
\end{proof}

\begin{proof}[Proof of Theorem~\ref{thm:ggm}]
We only need to prove that under Assumption D1, the general Assumptions A1-A2 hold.
Let $\Lambda_{\min}(\Sigma_{\bm{x}}) = \min_{1 \leq j \leq m_0+1}\Lambda_{\min}(\Sigma_{j})$ and  $\sigma_{\varepsilon} = \max_{1 \leq j \leq m_0+1}\sigma_{\varepsilon, j}$.
Specifically, we want to prove that there exist constants $c_1, c_2 > 0 $ such that for any 
${a_n =\Omega \left( \log (p\vee n)  \frac{\Lambda^2_{\max}(\Sigma_{\bm{x}}) }{\Lambda^2_{\min}(\Sigma_{\bm{x}})}\right)}$
and  all $ v \in \mathbb{R}^{p^2}$,
{\footnotesize\begin{equation}
\inf_{1 \leq j \leq {m_0+1}, t_j > u > l \geq t_{j-1}, | u - l  |  > a_n  }  v^\prime I_p \otimes \left({(l-u)}^{-1}  \sum_{t=l}^{u - 1} \bm{x}_{t}\bm{x}_{t}^\prime\right) v  \geq
\alpha_1 \Vert v \Vert_2^2 - \tau_1 \Vert v \Vert_1^2 ,\nonumber
\end{equation}}
{\footnotesize\begin{equation}
\sup_{1 \leq j \leq {m_0+1}, t_j > u > l \geq t_{j-1}, | u - l  |  > a_n  }  v^\prime I_p \otimes \left({(l-u)}^{-1}  \sum_{t=l}^{u - 1} \bm{x}_{t}\bm{x}_{t}^\prime \right) v  \leq
\alpha_2 \Vert v \Vert_2^2 + \tau_2 \Vert v \Vert_1^2 ,\nonumber
\end{equation}}
 with parameters $\alpha_1 = \frac{\lambda_{\min}(\Sigma_{\bm{x}})}{2}$, $\alpha_2 = \frac{3\lambda_{\max}(\Sigma_{\bm{x}})}{2}$, ${\tau_1 = \tau_2 = c_0\frac{\log (p\vee n)}{a_n}  \frac{\Lambda^2_{\max}(\Sigma_{\bm{x}}) }{\Lambda_{\min}(\Sigma_{\bm{x}})} }$,
and  probability at least 
$ 1 - c_1 \exp\left(-c_2a_n  \frac{\Lambda^2_{\min}(\Sigma_{\bm{x}})}{\Lambda^2_{\max}(\Sigma_{\bm{x}})}  \right) $.
Moreover,

{\footnotesize
\begin{equation}\label{eq:DB_bound_ggm}
\sup_{1 \leq j \leq {m_0+1},  t_j > u > l \geq t_{j-1}, | u - l  |  > a_n }  \left|\left| {(l-u)}^{-1}  \sum_{t=l}^{u - 1} \bm{x}_{t} \bm{\varepsilon}_t^\prime    \right|\right|_\infty \leq {c_0\Lambda_{\max}(\Sigma_{\bm{x}}) \sqrt{\frac{\log (p\vee n)}{a_n}}}, \nonumber
\end{equation}}
with probability at least ${1- c_1 \exp( -c_2 \log (p\vee n))}$.


The first part of the proof is similar to the proof of Lemma \ref{thm:mlr}.
Based on Lemma B.1 in \cite{Basu_2015},  it is enough to show that for all $ v \in \mathbb{R}^{p}$, 
{\footnotesize\begin{equation}\label{eq:RE_bound_ggm_2}
\inf_{1 \leq j \leq {m_0+1}, t_j > u > l \geq t_{j-1}, | u - l  |  > a_n  }  v^\prime \left({(l-u)}^{-1}  \sum_{t=l}^{u - 1} \bm{x}_{t}\bm{x}_{t}^\prime - \Sigma_{\bm{x}}\right) v  \geq
\alpha_1 \Vert v \Vert_2^2 - \tau_1 \Vert v \Vert_1^2 ,
\end{equation}}
and 
{\footnotesize\begin{equation}\label{eq:RE_bound_ggm_3}
\sup_{1 \leq j \leq {m_0+1}, t_j > u > l \geq t_{j-1}, | u - l  |  > a_n  }  v^\prime \left({(l-u)}^{-1}  \sum_{t=l}^{u - 1} \bm{x}_{t}\bm{x}_{t}^\prime - \Sigma_{\bm{x}}\right) v  \leq
\alpha_2 \Vert v \Vert_2^2 + \tau_2 \Vert v \Vert_1^2 .
\end{equation}}

Note that the vector $\bm{x}$ is Gaussian random variable with zero mean and covariance $\Sigma_{\bm{x}}$.  For any unit vector $\bm{u} \in \mathbb{R}^{p}$, the random variable $\bm{u}^\prime \bm{x}$ is a Gaussian random variable with zero mean and variance at most $\Lambda_{\max}(\Sigma_{\bm{x}})$. 
The result in Lemma 1 in \cite{loh2012}, together with the substitutions $\sigma_x^2 = \Lambda_{\max}(\Sigma_{\bm{x}}) $, completes the first part of the proof.

 
  For the second part, note that $\bm{x}_t$ and $\bm{\varepsilon}_t$ are independent, we have
 {\footnotesize
\begin{align}
e_i^\prime \left \{ {(l-u)}^{-1}  \sum_{t=l}^{u - 1} \bm{x}_{t}\bm{\varepsilon}_t^\prime   \right\}e_j
= & \frac{1}{2} \Bigg[\left (  {(l-u)}^{-1}\sum_{t=l}^{u - 1} (e_i^\prime \bm{x}_{t} + e_j^\prime \bm{\varepsilon}_t )^\prime(e_i^\prime\bm{x}_{t} + e_j^\prime \bm{\varepsilon}_t )-\Sigma_{e_i^\prime\bm{x} + e_j^\prime \bm{\varepsilon}} \right) \nonumber\\
& \quad \quad  -  \left({(l-u)}^{-1} \sum_{t=l}^{u - 1} \bm{x}_{t}^\prime e_i e_i^\prime\bm{x}_{t} -e_i^\prime\Sigma_{\bm{x} }e_i\right)  -   \left({(l-u)}^{-1}\sum_{t=l}^{u - 1}  \bm{\varepsilon}_t^\prime e_j e_j^\prime\bm{\varepsilon}_t-e_j^\prime\Sigma_{\varepsilon}e_j\right) \Bigg] , \nonumber\\
= & \frac{1}{2}(I_1 - I_2 -I_3), \nonumber
\end{align}}
where $e_i^\prime\bm{x}_{t} + e_j^\prime\bm{\varepsilon}_t$ is Gaussian variable with mean zero and variance 2.

 We combine the three upper bounds, which gives a upper  bound of the deviation:
 {\footnotesize
\begin{align}
\mathbb{P} \left(  \left \vert e_i^\prime \left \{ {(l-u)}^{-1}  \sum_{t=l}^{u - 1} \bm{x}_{t}\bm{\varepsilon}_t^\prime    \right\}e_j \right \vert  > t \right) & = \mathbb{P} \left( \left| I_1 - I_2 - I_3 \right| > 2t \right) \nonumber\\    
& \leq \mathbb{P} \left( \vert I_1  \vert  + \vert I_2  \vert + \vert I_3 \vert  > 2t \right) \nonumber\\    
& \leq \mathbb{P} \left( \vert I_1  \vert  > \frac{2t}{3} \right) + \mathbb{P} \left( \vert I_2  \vert  > \frac{2t}{3} \right) + \mathbb{P} \left( \vert I_3  \vert  > \frac{2t}{3} \right) \nonumber \\   
& \leq 6  \exp \left(-c^\prime a_n \min \left(t^2, t \right) \right). \nonumber
\end{align}}
Taking union over all $1 \leq i \leq p$ and $1 \leq j \leq p$ and rescalling the sub-exponential parameter, we have 
{\footnotesize
\begin{equation}
  \mathbb{P} \left( \left \Vert {(l-u)}^{-1}  \sum_{t=l}^{u - 1} \bm{x}_{t}\bm{\varepsilon}^\prime_t  \right\Vert_\infty > t \right) \leq  6p  \exp \left(-c a_n \min \left(\frac{t^2}{  \Lambda_{\max}^2(\Sigma_{\bm{x}})}, \frac{t}{  \Lambda_{\max}(\Sigma_{\bm{x}})}\right) \right) \nonumber
\end{equation}}

Setting ${t = c_0\Lambda_{\max}(\Sigma_{\bm{x}}) \sqrt{\frac{\log( p\vee n)}{a_n}}}$ and the assumption that 
${a_n = \Omega( \log (p\vee n))}$,
we have
{\footnotesize
\begin{align}
 \mathbb{P} \left( \left \Vert {(l-u)}^{-1}  \sum_{t=l}^{u - 1} \bm{x}_{t}\bm{\varepsilon}_t^\prime  \right\Vert_\infty > c_0\Lambda_{\max}(\Sigma_{\bm{x}}) \sqrt{\frac{\log (p\vee n)}{a_n}} \right)  
 \leq c_1  \exp (-c_2 \log (p\vee n) ). \nonumber 
\end{align}}
  \end{proof}

\begin{proof}[Proof of Proposition~\ref{thm:mlr}]
We only need to prove that under Assumption E1, the general Assumptions A1-A2 hold.
Let $\Lambda_{\min}(\Sigma_{\bm{x}}) = \min_{1 \leq j \leq m_0+1}\Lambda_{\min}(\Sigma_{\bm{x},j})$, $\sigma_{\bm{x}} = \max_{1 \leq j \leq m_0+1}\sigma_{\bm{x},j}$ and  $\sigma_{\varepsilon} = \max_{1 \leq j \leq m_0+1}\sigma_{\varepsilon,j}$.
Specifically, we want to prove that there exist constants $c_1, c_2 > 0 $ such that for any 
${a_n = \Omega \left( \log (p\vee n) \max \left\{ \frac{\sigma^4_{\bm{x}}}{\Lambda^2_{\min}(\Sigma_{\bm{x}})}, 1\right\} \right)}$, 
{\footnotesize\begin{equation}\label{eq:RE_bound_mlr_1}
\inf_{1 \leq j \leq {m_0+1}, t_j > u > l \geq t_{j-1}, | u - l  |  > a_n  }  v^\prime \left({(l-u)}^{-1}  \sum_{t=l}^{u - 1} \bm{x}_{t}\bm{x}_{t}^\prime \right) v  \geq
\alpha_1 \Vert v \Vert_2^2 - \tau_1 \Vert v \Vert_1^2 ,
\end{equation}}
{\footnotesize\begin{equation}\label{eq:RE_bound_mlr_2}
\sup_{1 \leq j \leq {m_0+1}, t_j > u > l \geq t_{j-1}, | u - l  |  > a_n  }  v^\prime \left({(l-u)}^{-1}  \sum_{t=l}^{u - 1} \bm{x}_{t}\bm{x}_{t}^\prime \right) v  \leq
\alpha_2 \Vert v \Vert_2^2 + \tau_2 \Vert v \Vert_1^2 ,
\end{equation}}
for all $ v \in \mathbb{R}^p$, with parameters $\alpha_1 = \frac{\lambda_{\min}(\Sigma_{\bm{x}})}{2}$, $\alpha_2 = \frac{3\lambda_{\max}(\Sigma_{\bm{x}})}{2}$, ${\tau_1 = c_0\frac{\log (p\vee n) \lambda_{\min}(\Sigma_{\bm{x}})}{a_n} \max \left\{ \frac{\sigma^4_{\bm{x}}}{\Lambda^2_{\min}(\Sigma_{\bm{x}})}, 1\right\}} $, ${\tau_2 = c_0\frac{\log (p\vee n) \lambda_{\max}(\Sigma_{\bm{x}})}{a_n} \max \left\{ \frac{\sigma^4_{\bm{x}}}{\Lambda^2_{\min}(\Sigma_{\bm{x}})}, 1\right\} }$ and  probability at least $ 1 - c_1 \exp\left(-c_2a_n \min \left\{ \frac{\Lambda^2_{\min}(\Sigma_{\bm{x}})}{\sigma^4_{\bm{x}}}, 1\right\}  \right) $. Moreover,

{\footnotesize\begin{equation}\label{eq:DB_bound_mlr}
\sup_{1 \leq j \leq {m_0+1},  t_j > u > l \geq t_{j-1}, | u - l  |  > a_n }  \left|\left| {(l-u)}^{-1}  \sum_{t=l}^{u - 1} \bm{x}_{t} \varepsilon_t    \right|\right|_\infty \leq {c_0\sigma_{\bm{x}}\sigma_{\varepsilon} \sqrt{\frac{\log( p\vee n)}{a_n}} },\nonumber
\end{equation}}
with probability at least ${1- c_1 \exp( -c_2 \log (p\vee n)) }$.

The  proof of the first part  lemma is similar to those of Lemma 1 in \cite{loh2012}. 
Applying Supplementary Lemma 13 in \cite{loh2012}
, together with the substitutions 
{\footnotesize\begin{equation}
\widehat{\Gamma} - \Sigma = {(u - l)}^{-1}  \sum_{t=l}^{u - 1} \bm{x}_{t}\bm{x}_{t}^\prime - \Sigma_{\bm{x}}, \text{ and } {s= \frac{u-l}{c\log (p\vee  n)} \min \left\{ \frac{\Lambda^2_{\min}(\Sigma_{\bm{x}})}{\sigma^4_{\bm{x}}}, 1\right\} }\nonumber
\end{equation}}
where $c$ is chosen sufficiently small so $s>1$, we see that it suffices to show that
{\footnotesize\begin{equation}
\sup_{v \in \mathbb{K}(2s))} \left | v^\prime \left({(u - l)}^{-1}  \sum_{t=l}^{u - 1} \bm{x}_{t}\bm{x}_{t}^\prime - \Sigma_{\bm{x}}\right) v \right | \leq \frac{\Lambda^2_{\min}(\Sigma_{\bm{x}})}{54}
, \nonumber
\end{equation}}
with high probability, where $\mathbb{K}(2s) = \left \{ v \in \mathbb{R}^p: \Vert v\Vert_2 \leq 1, \Vert v \Vert_0 \leq 2s \right\}$  is the set of sparse vectors .

Note that the vector $\bm{x}$ is sub-Gaussian with parameters $(\Sigma_x, \sigma_x^2)$.
Consequently, by Lemma 15 in \cite{loh2012}, we have
{\footnotesize\begin{equation}
\mathbb{P} \left[\sup_{v \in \mathbb{K}(2s))} \left | v^\prime \left({(l-u)}^{-1}  \sum_{t=l}^{u - 1} \bm{x}_{t}\bm{x}_{t}^\prime - \Sigma_{\bm{x}}\right) v \right | > t \right] \leq 2 \exp \left ( - c^\prime (u - l) \min \left(\frac{t^2}{\sigma_x^4} , \frac{t}{\sigma_x^2} \right) + 2s \log p  \right)
, \nonumber
\end{equation}}
for some universal constant $c^\prime > 0$.
Setting $t=\frac{\Lambda^2_{\min}(\Sigma_{\bm{x}})}{54}$, we see that as long as the constant $c$ is chosen sufficiently small, we are guaranteed that
{\footnotesize\begin{equation}
\mathbb{P} \left[\sup_{v \in \mathbb{K}(2s))} \left | v^\prime \left({(l-u)}^{-1}  \sum_{t=l}^{u - 1} \bm{x}_{t}\bm{x}_{t}^\prime - \Sigma_{\bm{x}}\right) v \right | >\frac{\Lambda^2_{\min}(\Sigma_{\bm{x}})}{54} \right] \leq 2 \exp \left ( - c_1 (u - l) \min \left( \frac{\Lambda^2_{\min}(\Sigma_{\bm{x}})}{\sigma_x^4} , 1 \right) \right)
. \nonumber
\end{equation}}

Thus, we have the lower-RE condition
{\footnotesize\begin{equation}
v^\prime \left({(l-u)}^{-1}  \sum_{t=l}^{u - 1} \bm{x}_{t}\bm{x}_{t}^\prime \right) v  \geq \frac{\lambda_{\min}(\Sigma_{\bm{x}})}{2} \Vert v \Vert_2^2 - \frac{\lambda_{\min}(\Sigma_{\bm{x}})}{2s} \Vert v \Vert_1^2, \nonumber
\end{equation}}
and the upper-RE condition 
{\footnotesize\begin{equation}
v^\prime \left({(l-u)}^{-1}  \sum_{t=l}^{u - 1} \bm{x}_{t}\bm{x}_{t}^\prime \right) v  \leq \frac{3\lambda_{\max}(\Sigma_{\bm{x}})}{2} \Vert v \Vert_2^2 + \frac{\lambda_{\min}(\Sigma_{\bm{x}})}{2s} \Vert v \Vert_1^2\leq \frac{3\lambda_{\max}(\Sigma_{\bm{x}})}{2} \Vert v \Vert_2^2 + \frac{\lambda_{\max}(\Sigma_{\bm{x}})}{2s} \Vert v \Vert_1^2. \nonumber
\end{equation}}

 
 The  rest of proof is similar to that of Lemma 14 in \cite{loh2012}. Based on the fact that  $\text{Cov}( e_i^\prime\bm{x}_t, \varepsilon_t) = 0$ for any $i = 1 , \dots, p$, we have
 {\footnotesize
\begin{align}
& e_i^\prime \left \{ {(l-u)}^{-1}  \sum_{t=l}^{u - 1} \bm{x}_{t}\varepsilon_t   \right\} \nonumber \\
= & \frac{1}{2} \left [\left (  {(l-u)}^{-1}\sum_{t=l}^{u - 1} (e_i^\prime \bm{x}_{t} + \varepsilon_t )^\prime(e_i^\prime\bm{x}_{t} + \varepsilon_t )-\Sigma_{e_i^\prime\bm{x}_{t} + \varepsilon_t} \right)-  \left({(l-u)}^{-1} \sum_{t=l}^{u - 1} \bm{x}_{t}^\prime e_i e_i^\prime\bm{x}_{t} -\Sigma_{\bm{x}_{t} }\right) -   \left({(l-u)}^{-1}\sum_{t=l}^{u - 1}  \varepsilon_t^2-\sigma^2_{\varepsilon_t}\right) \right] , \nonumber\\
= & \frac{1}{2}(I_1 - I_2 -I_3), \nonumber
\end{align}}
where $e_i^\prime\bm{x}_{t} + \varepsilon_t$ is sub-Gaussian with parameter at most $\sqrt{\sigma_{x}^2 +\sigma_{\varepsilon}^2 }$.

Note that if $X$ is a zero-mean sub-Gaussian random variable with parameter $\sigma$, then the random variable $Z = X^2- \E(X^2)$ is sub-exponential with parameter $\Vert Z \Vert_{\psi_1}^2 \leq 2 \Vert X \Vert_{\psi_2}^2 = 2 \sigma^2$ \citep{vershynin2010introduction}.  
Let $Z_1 = (e_i^\prime\bm{x}_{l} + \varepsilon_l, e_i^\prime\bm{x}_{l+1} + \varepsilon_{l+1}, \dots, e_i^\prime\bm{x}_{u-1} + \varepsilon_{u-1})^\prime$,  $Z_2 = (e_i^\prime\bm{x}_{l}, e_i^\prime\bm{x}_{l+1}, \dots, e_i^\prime\bm{x}_{u-1} )^\prime$ and $Z_3 = ( \varepsilon_l, \varepsilon_{l+1}, \dots,  \varepsilon_{u-1})^\prime$, note that if $X$ is a zero-mean sub-Gaussian variable with parameter  $\sigma$, then the rescaled variable $X/ \sigma$ is sub-Gaussian with parameter 1.
We may assume that  $\sigma_x =  \sigma_{\varepsilon} = 1$ without loss of generality.
Applying Proposition 5.16 (Bernstein-type inequality) in \cite{vershynin2010introduction}, we have
{\footnotesize
\begin{equation}
\mathbb{P} \left( \vert I_1 \vert  > t_1 \right)
= \mathbb{P} \left( \left|  \frac{1}{u-l}\sum_{i=1}^{u-l}Z_{1,i}^2 -  \E(Z_1^\prime Z_1)  \right| > t_1 \right) 
\leq 2 \exp (-c a_n \min(\frac{t_1^2}{16}, \frac{t_1}{4}) ).\nonumber
\end{equation}
\begin{equation}
\mathbb{P} \left( \vert I_2 \vert  > t_2 \right)
= \mathbb{P} \left( \left|  \frac{1}{u-l}\sum_{i=1}^{u-l}Z_{2,i}^2 -  \E(Z_2^\prime Z_2)  \right| > t_1 \right) 
\leq 2 \exp (-c a_n \min(\frac{t_2^2}{4 }, \frac{t_2}{ 2 }) ).\nonumber
\end{equation}
\begin{equation}
\mathbb{P} \left( \vert I_3 \vert  > t_3 \right)
= \mathbb{P} \left( \left|  \frac{1}{u-l}\sum_{i=1}^{u-l}Z_{3,i}^2 -  \E(Z_3^\prime Z_3)  \right| > t_1 \right) 
\leq 2 \exp (-c a_n \min(\frac{t_3^2}{  4}, \frac{t_3}{  2}) ). \nonumber
\end{equation}}

Combining the three upper bounds, we establish the result
{\footnotesize
\begin{align}
\mathbb{P} \left(  \left \vert e_i^\prime \left \{ {(l-u)}^{-1}  \sum_{t=l}^{u - 1} \bm{x}_{t}\varepsilon_t   \right\} \right \vert  > t \right) & = \mathbb{P} \left( \left| I_1 - I_2 - I_3 \right| > 2t \right) \nonumber\\    
& \leq \mathbb{P} \left( \vert I_1  \vert  + \vert I_2  \vert + \vert I_3 \vert  > 2t \right) \nonumber\\    
& \leq \mathbb{P} \left( \vert I_1  \vert  > \frac{2t}{3} \right) + \mathbb{P} \left( \vert I_2  \vert  > \frac{2t}{3} \right) + \mathbb{P} \left( \vert I_3  \vert  > \frac{2t}{3} \right) \nonumber \\   
& \leq 6  \exp \left(-c^\prime a_n \min \left(t^2, t \right) \right). \nonumber
\end{align}}

Taking union over all $1 \leq i \leq p$ and rescalling the sub-exponential parameter, we have
{\footnotesize
\begin{equation}
  \mathbb{P} \left( \left \Vert {(l-u)}^{-1}  \sum_{t=l}^{u - 1} \bm{x}_{t}\varepsilon_t  \right\Vert_\infty > t \right) \leq  6p  \exp \left(-c a_n \min \left(\frac{t^2}{  (\sigma_x \sigma_{\varepsilon})^2}, \frac{t}{  \sigma_x \sigma_{\varepsilon}}\right) \right) \nonumber
\end{equation}}

Setting ${t = c_0\sigma_{\bm{x}}\sigma_{\varepsilon} \sqrt{\frac{\log (p\vee n)}{a_n}}}$ and the assumption that 
${a_n = \Omega( \log (p\vee n)) }$,
we have
{\footnotesize
\begin{align}
 &\mathbb{P} \left( \left \Vert {(l-u)}^{-1}  \sum_{t=l}^{u - 1} \bm{x}_{t}\varepsilon_t  \right\Vert_\infty > c_0\sigma_{\bm{x}}\sigma_{\varepsilon}\ \sqrt{\frac{\log (p\vee n)}{a_n}} \right)  \nonumber\\
 \leq  & c_1  \exp \left(-c a_n \min \left(\frac{c_0^2 \log (p\vee n)}{  a_n}, \frac{c_0 \sqrt{\log( p\vee n)}}{  \sqrt{a_n}} \right) \right)   \nonumber \\
 =& c_1  \exp (-c_2 \log (p\vee n) ).  \nonumber
\end{align}}
 \end{proof}

\section{Details about the illustration plots in the introduction}\label{sec:plot}

In this scenario, we set $ n = 1,000 $, $ p = 20 $, $ m_0 = 2 $. 
The  number of non-zero elements of the coefficient vectors  in $j$th segments $d_j= 2$, for all $j= 1, \dots, m_0 + 1$.
 The coefficient vector are chosen to have the random sparse structure in each segment, with different entries $ -2 $, $2 $ and $-2$,  respectively. 
 The true jump size between two consecutive stationary segments is given by $\left\|{\bm{B}}_{j+1} - {\bm{B}}_{j} \right\|_F = 4$ for $j = 1,\dots, m_0$.
 The error variance  is $\Sigma = I$ (which is an identity matrix).
The two change points are equally spaced   ($t_1 = \lfloor \frac{n}{3} \rfloor = 333$ and $t_2=  \lfloor \frac{2n}{3} \rfloor = 666$) with the block size of $b_n=30$.

Denote the \emph{jump} for each block by setting $v_i = \left\|\widehat{\bm{\Theta}}_i \right\|_F $, $i = 2, \cdots, k_n$ and $v_1 = 0 $. 
    Set $V = (v_1,  \cdots, v_{k_n})$, where $k_n$ is the number of blocks in~\eqref{block_model_1}.
    We apply $K$-means clustering to the jump vector $V$ with two centers. Denote the sub-vector with a smaller center as the small subgroup, $V_S$ , and the other sub-vector as the large subgroup, $V_L$.
     Add the corresponding blocks in the large subgroup into $J$.  Set $\widetilde{I}_n = J$,  which contains indices of blocks with large jumps.
     The hard-thresholding value $\omega_n$ would be  any value within the interval
     $ [\max_{v \in V_S}{v}, \min_{v \in V_L}{v})$.
     
     The set of estimated change points after hard-thresholding is given by
 \[ \widetilde{A}_n  = \left\{\widetilde{t}_1 ,\dots,  \widetilde{t}_{\widetilde{m} } \right\} = \left \{r_{i-1} : \left\|\widehat{\bm{\Theta}}_i\right \|_F  > \omega_n,\  i = 2,\dots, k_n \right\},
 \]
where $\omega_n$ is the hard-threshold value.


\section{Additional Results of Simulation Scenario A}\label{sec:add_sim}
\setcounter{equation}{0}
In this section, additional results of Simulation Scenario A  with different $\gamma$ settings are provided.
The detection results are robust with respect to changes in $\gamma$ as shown in Table \ref{table_sim_E_bn}.
From Table \ref{table_sim_E_bn}, we can also see the median of optimal block size are decreasing as the number of change points increases. This results is consistent with the results of selection frequency of block size as shown in Table~\ref{table_sim_E_bn_2}.

\begin{table}[!ht]
\caption{\label{table_sim_E_bn}Results of TBFL with optimal block size ($b_n = 30, 40, 50, 60, 70$) in simulation scenario A. The $\widehat{b}_n$ stands for the median of optimal block size. }
\scriptsize
\centering
\begin{tabular}{l|ccccc|ccccc|ccccc} 
  \hline
  & \multicolumn{5}{c}{TBFL $\gamma = 1$}& \multicolumn{5}{|c}{TBFL $\gamma = 1.5$}& \multicolumn{5}{|c}{TBFL $\gamma = 3$} \\
  $\vert \widetilde{m}^f -m_0\vert$ & $\widehat{b}_n$ & 0 &1 &2 & $>2$& $\widehat{b}_n$ & 0 &1 &2 & $>2$& $\widehat{b}_n$ & 0 &1 &2 & $>2$ \\ 
  \hline 
 $m_0= 2$ &50 & 94 & 4 & 1 & 1 & 50 & 94 & 4 & 2 & 0 & 50 & 94 & 5 & 1 & 0 \\ 
 $m_0= 4$ & 50 & 98 & 2 & 0 & 0 & 50 & 98 & 2 & 0 & 0 & 50 & 98 & 2 & 0 & 0 \\ 
$m_0= 6$ &  40 & 93 & 7 & 0 & 0 & 40 & 93 & 7 & 0 & 0 & 40 & 95 & 5 & 0 & 0 \\ 
$m_0= 8$ &  40 & 93 & 7 & 0 & 0 & 40 & 93 & 7 & 0 & 0 & 40 & 93 & 7 & 0 & 0 \\ 
$m_0= 10$ &  40 & 99 & 1 & 0 & 0 & 40 & 99 & 1 & 0 & 0 & 40 & 99 & 1 & 0 & 0 \\ 
$m_0= 12$ &  40 & 97 & 3 & 0 & 0 & 40 & 97 & 3 & 0 & 0 & 40 & 98 & 2 & 0 & 0 \\ 
 $m_0= 14$ & 40 & 95 & 5 & 0 & 0 & 40 & 95 & 5 & 0 & 0 & 40 & 95 & 5 & 0 & 0 \\ 
$m_0= 16$ &  40 & 74 & 16 & 6 & 4 & 40 & 74 & 16 & 6 & 4 & 40 & 73 & 17 & 6 & 4 \\ 
  \hline
  \end{tabular}
  \end{table}

\begin{table}[!ht]
\caption{\label{table_sim_E_bn_2}Results of selection frequency of  block size $b_n = 30, 40, 50, 60, 70$ in simulation scenario A. }
\scriptsize
\centering
\begin{tabular}{l|ccccc} 
  \hline
  & \multicolumn{5}{c}{TBFL $\gamma = 1.5$} \\
 & $b_n = 30$ & $b_n = 40$& $b_n = 50$& $b_n = 60$& $b_n = 70$  \\ 
  \hline 
 $m_0= 2$ &14 & 22 & 23 & 25 & 16 \\ 
 $m_0= 4$ & 20 & 22 & 20 & 20 & 18 \\ 
 $m_0= 6$ & 34 & 25 & 14 & 21 & 6 \\ 
 $m_0= 8$ & 32 & 31 & 15 & 10 & 12 \\ 
 $m_0= 10$ & 35 & 20 & 8 & 19 & 18 \\ 
 $m_0= 12$ & 37 & 32 & 8 & 10 & 13 \\ 
 $m_0= 14$ & 40 & 33 & 22 & 5 & 0 \\ 
 $m_0= 16$ & 29 & 36 & 24 & 8 & 3 \\ 
  \hline
  \end{tabular}
  \end{table}


\section{Additional Simulation Results}\label{Sim_Scenarios}
\setcounter{equation}{0}
We evaluate the performance of the proposed three-stage estimator with respect to both structural break detection and parameter estimation. In this section, we consider four simulation scenarios. 

For all settings, we report the error of locations of the estimated break points and the selection rate, i.e., the percentage of replicates where each break point is correctly identified. 
The error of the locations of estimated break points is defined as
$\text{error}_j = {\left\vert\tilde{t}_j^f - t_j\right\vert}$,
$j = 1,\dots, m_0$.
The percentage is calculated as the proportion of replicates, where the estimated break points by TBFL are close to each of the true break points. 
Specifically, to compute the selection rate,  a selected break point is counted as a ``success'' for the $j$-th true break point, $t_j$, if it falls in the interval $ [ t_j - \frac{t_{j} - t_{j-1}}{5}, t_j + \frac{t_{j+1} - t_{j}}{5}  ] $, $j = 1,\dots, m_0$. 
The results are reported in Table~\ref{table_sim_A}, Table~\ref{table_sim_B} and  Table~\ref{table_sim_C}.

For parameter estimation, we evaluate the performance of our procedure by reporting 
the mean and standard deviation of relative estimation error (REE), the true positive rate (TPR) and the false positive rate (FPR). 
The  relative estimation error (REE), the true positive rate (TPR) and the false positive rate (FPR)  are calculated by 
\begin{align}\label{REE}
    \text{REE} = \frac{\Vert {\widetilde{\bm{B}}} - \bm{B}^\star \Vert_F}{\Vert \bm{B}^\star \Vert_F},\ 
    \text{TPR} = \frac{\text{TP}}{\text{TP} + \text{FN}},\ 
    \text{FPR} = \frac{\text{FP}}{\text{FP} + \text{TN}}.
\end{align}
Specifically, for TPR and FPR, we use the median number of nonzero and zero elements among 100 replicates.
The results of all simulation settings are reported in Table~\ref{table_sim_tpr}.

\begin{table}[!ht]
\caption{\label{table_sim_A}Means and standard deviations of estimation error and selection rates.}
\scriptsize
\centering
{
\begin{tabular}{lccccc} 
  \hline
  & break point & truth & mean (error) & std (error) & selection rate  \\ 
  \hline
    Simulation B.1  & & &  \\
 & 1 & 333 & 1.0714 & 5.6281 & 0.98 \\ 
   & 2 & 666 & 5.4796 & 18.6692 & 0.94 \\ 
    Simulation B.2 & &  \\
     & 1 & 333 & 0.11 & 0.3145 & 1 \\ 
   & 2 & 666 & 0.08 & 0.4422 & 1 \\ 
    Simulation B.3  & \\
  & 1 & 333 & 3.5833 & 9.2106 & 0.96 \\ 
   & 2 & 666 & 8.3636 & 20.0959 & 0.94 \\ 
  Simulation B.4  & & &  \\
  $b_n = 20$\\
& 1 & 333 & 0.78 & 4.1889 & 1 \\ 
   & 2 & 666 & 3.66 & 14.8869 & 0.98 \\ 
 $b_n = 30$\\
 & 1 & 333 & 3.94 & 13.841 & 0.99 \\ 
   & 2 & 666 & 5.8817 & 17.7034 & 0.91 \\ 
 $b_n = 40$\\
 & 1 & 333 & 1.1837 & 3.6703 & 0.98 \\ 
   & 2 & 666 & 3.398 & 13.8048 & 0.96 \\ 
$b_n = 50$\\
   & 1 & 333 & 3.433 & 13.2672 & 0.95 \\ 
   & 2 & 666 & 5.2143 & 13.3696 & 0.98 \\ 
   \hline
\end{tabular}
}
\end{table}
The details of the simulation setting in each scenario are explained as follows. 

\textit{ Setting B (Constant Model)}. In the scenario B, $ n = 1,000 $, $ p = 100 $, $ m_0 = 2 $, $ t_1 = \left\lfloor \frac{n}{3} \right\rfloor  = 333 $, $ t_2 =  \lfloor \frac{2n}{3} \rfloor  = 666$ and $b_n=\left \lfloor n^{\frac{1}{2}} \right\rfloor = 31$ for Scenario B.1-B.3, while the  mean coefficient vector vary across scenarios. 
The number of non-zero elements in $j$th segments $d_j= 10$, for all $j= 1, \dots, m_0 + 1$.
\begin{itemize}
    \item[B.1] (Simple $\bm{\mu}$): In the scenario B.1, the mean coefficient vector are chosen to have the same structure in each segment, but different magnitude  entries $ -0.5 $, $ 0.5 $, and $ -0.3 $, respectively. 
    \item[B.2] (Random $\bm{\mu}$): In the scenario B.2, the mean coefficient vector in each segmentation is chosen to have a random sparse structure and random entries sampled from $\mbox{Uniform}(-1,-0.5)$, $\mbox{Uniform}(0.5,1)$ and $\mbox{Uniform}(-1,0.5)$, respectively.
    \item[B.3] (Changes in only a subset of components in $\bm{\mu}$) : In the scenario B.3, the mean coefficient vector are chosen similar to scenario B.1, but only a subset of components have abrupt changes in their parameters. There are only 5 out of 100 components which are dealing with abrupt change. 
    \item[B.4] (Simple $\bm{\mu}$ with different blocks size): In the scenario B.4, all the settings are same as those in the scenario B.1 except the blocks size. 
\end{itemize}

\begin{table}[!ht]
\caption{\label{table_sim_B}Results of means and standard deviations of estimation error and selection rates.}
\scriptsize
\centering
\begin{tabular}{lccccc} 
  \hline
  & break point & truth & mean (error) & std (error) & selection rate  \\ 
  \hline
    Simulation C.1  & & &  \\
 & 1 & 500 & 0.58 & 2.1376 & 1 \\ 
   & 2 & 1000 & 0.53 & 3.1155 & 1 \\ 
   & 3 & 1500 & 1.04 & 2.8459 & 1 \\ 
Simulation C.2  & & &  \\
 &   1 & 400 & 1.2424 & 3.6731 & 0.99 \\ 
&  2 & 800 & 1.48 & 7.2704 & 1 \\ 
 & 3 & 1600& 1.43 & 2.4173 & 1 \\ 
  Simulation C.3  & & &  \\
  $b_n = 30$ \\
  & 1 & 500 & 0.25 & 0.6256 & 1 \\ 
   & 2 & 1000 & 0.81 & 2.4769 & 1 \\ 
   & 3 & 1500 & 2.84 & 6.8956 & 1 \\ 
  $b_n = 50$ \\
  & 1 & 500 & 1.4343 & 6.3216 & 0.99 \\ 
   & 2 & 1000 & 5.8788 & 24.6776 & 0.98 \\ 
   & 3 & 1500 & 2.8283 & 12.7031 & 0.98 \\ 
$b_n = 60$ \\
    & 1 & 500 & 0.97 & 4.7725 & 1 \\ 
   & 2 & 1000 & 1.28 & 4.895 & 1 \\ 
   & 3 & 1500 & 3.65 & 12.5781 & 1 \\ 
  Simulation C.4  & & &  \\
 & 1 & 500 & 5.08 & 16.4198 & 1 \\ 
   \hline
\end{tabular}
\end{table}

\textit{ Setting C (Multiple Linear Model)}. In the scenario C, we set $ n = 2,000 $, $ p = 150 $, $ m_0 = 3 $ for Scenario C.1-C.3.  
 The  coefficient vectors in  Scenario C.1-C.3 are the same, with the number of non-zero elements in $j$th segments $d_j= 15$, for all $j= 1, \dots, m_0 + 1$.
 The coefficient vector are chosen to have the random sparse structure in each segment, with different magnitude  entries $ -3 $, $ 5 $, $-3$ and $3$, respectively. 
 The error variance in  Scenario C.1-C.3 is $\Sigma = I$.
 
\begin{itemize}
    \item[C.1] (Random $\bm{\beta}$, change points equally spaced): 
    $ t_1 = \lfloor \frac{n}{4} \rfloor  = 500 $, $ t_2 =  \lfloor \frac{2n}{4} \rfloor  = 1000$,
    $ t_3 =  \lfloor \frac{3n}{4} \rfloor  =1500$,
    and $b_n= 40$. 
    \item[C.2] (Random $\bm{\beta}$, change points not equally spaced ): In the scenario C.2, all the setting are the same as scenario C.1, except that change points are not equally spaced and $b_n = 30$. Specifically, we set
    $ t_1 = 400 $, $ t_2 = 800$, $ t_3 =  1600$. 
    \item[C.3] (Different block size) : In the scenario C.3, all the settings are same as those in the scenario C.1 except  the blocks sizes change. Here, we consider four different block size settings: $b_n= 30$, $b_n= 50$ and $b_n= 60$.
    \item[C.4] (Larger $p$ case): In the scenario C.4, we set $n= 1000$, $p = 200$, $b_n = 30$ with boundary block size $b_n^b= 120$, and only one change point in the middle $t_1 = 500$. 
    The number of non-zero elements in $j$th segments $d_j= 12$, for all $j= 1, \dots, m_0 + 1$, with entries $ -1 $ and $1$, respectively. 
    The error variance is smaller in this setting with $\Sigma = 0.01I$.

\end{itemize}



\textit{ Setting D (Gaussian Graphical Model)}. In the scenario D, $ n = 3,000 $,  $ m_0 = 2 $, $ t_1 = \lfloor \frac{n}{3} \rfloor  = 1000 $, $ t_2 =  \lfloor \frac{2n}{3} \rfloor  = 2000$ and $b_n=50$, while the covariance matrix $\Sigma_x$ vary across scenarios. All precision matrices are depicted in Figure~\ref{fig:sim_ggm_oemga}.

\begin{figure}[!ht]
\begin{center}
\includegraphics[width=0.32\textwidth, clip=TRUE, trim=0mm 25mm 0mm 16mm]{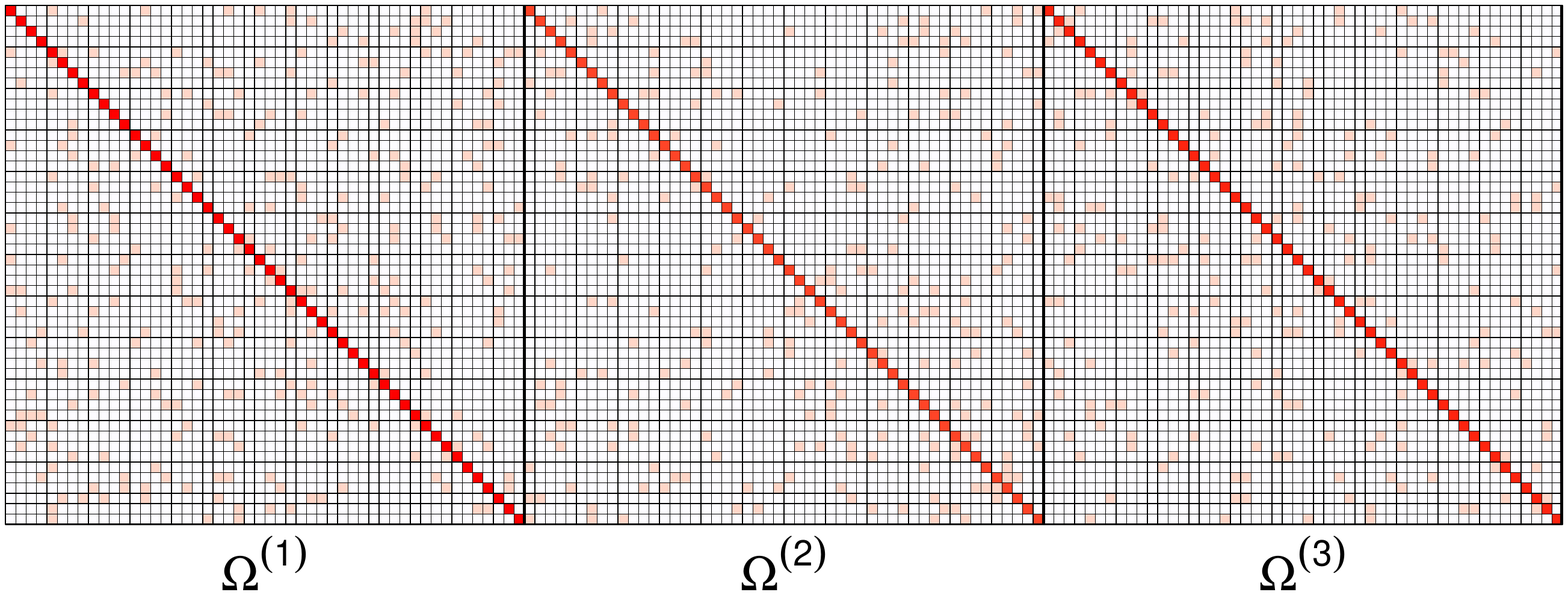}
\includegraphics[width=0.32\textwidth, clip=TRUE, trim=0mm 25mm 0mm 16mm]{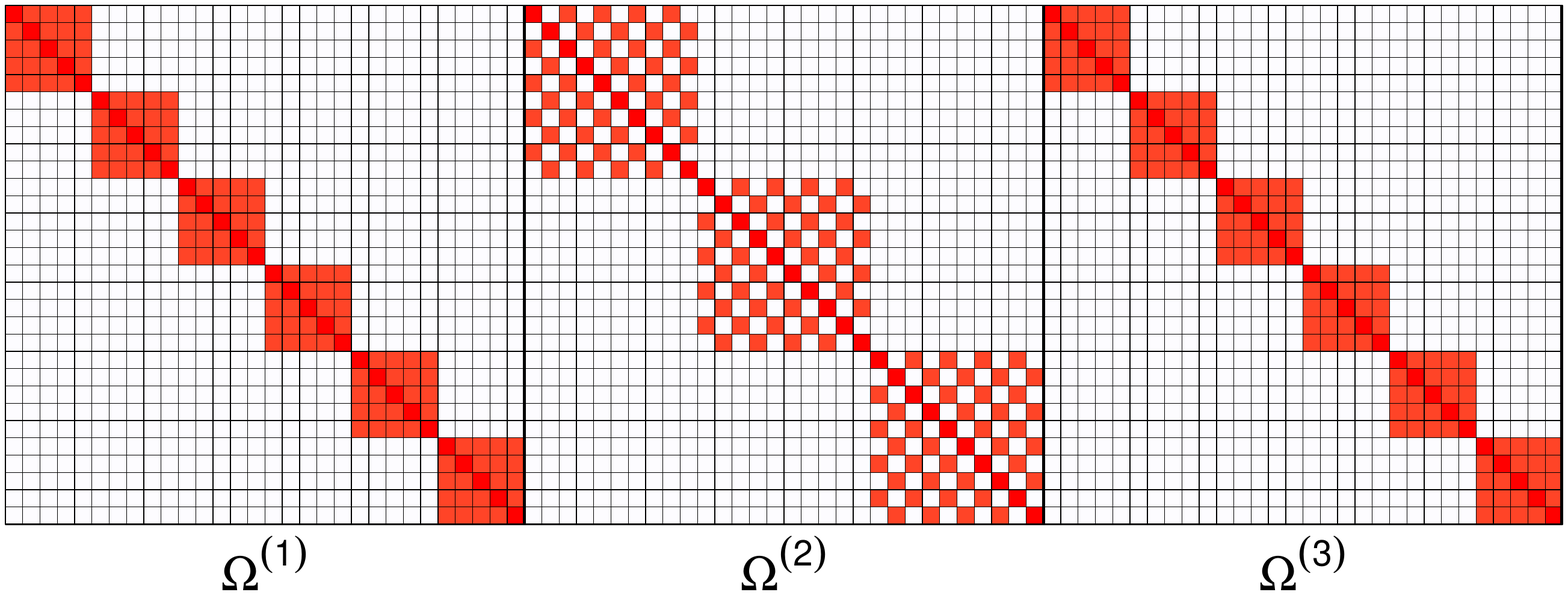}
\includegraphics[width=0.32\textwidth, clip=TRUE, trim=0mm 25mm 0mm 16mm]{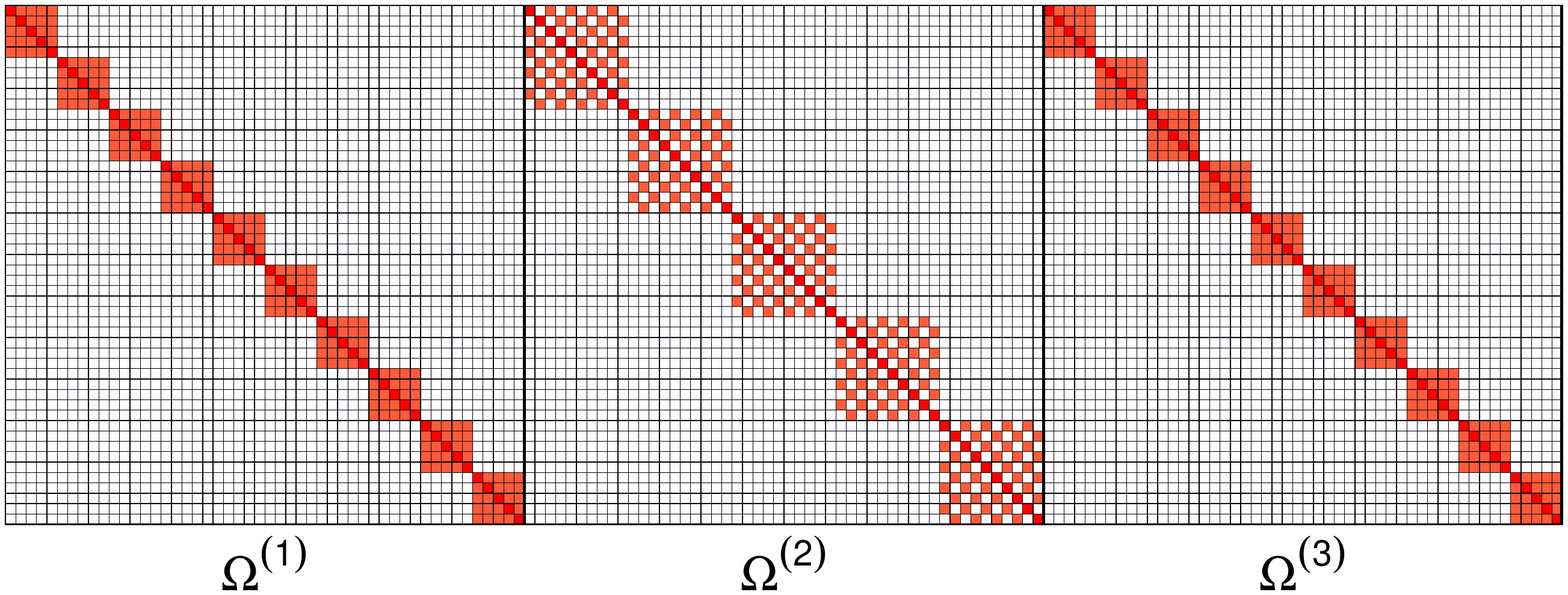}
\caption{True precision matrix $\Omega_x$ in Gaussian Graphical Model in Setting D.}
\label{fig:sim_ggm_oemga}
\end{center}
\end{figure}

\begin{table}[!ht]
\caption{\label{table_sim_C}Results of means and standard deviations of estimation error and selection rates.}
\scriptsize
\centering
\begin{tabular}{lccccc} 
  \hline
  & break point& truth & mean (error) & std (error) & selection rate  \\ 
  \hline
  Simulation D.1  & & &  \\
      & 1 & 1000 & 0.6495 & 4.975 & 0.97 \\ 
   & 2 & 2000 & 0.91 & 5.9086 & 1 \\ 
Simulation D.2  & & &  \\
   & 1 & 1000 & 0.0303 & 0.1723 & 0.99 \\ 
   & 2 & 2000 & 0.23 & 0.6333 & 1 \\ 
   Simulation D.3  & & &  \\
   & 1 & 1000 & 0.043 & 0.2917 & 0.93 \\ 
   & 2 & 2000 & 0.22 & 0.6289 & 1 \\ 
   \hline
\end{tabular}
\end{table}

\begin{itemize}
    \item[D.1] (Random $\bm{\Omega}$. Erd\"os-R\'enyi random graph): In the scenario D.1, the precision matrices are chosen to have a random sparse structure. 
    \item[D.2] (Toeplitz type $\bm{\Omega}$.): In the scenario D.2, the precision matrices are chosen to have a  Toeplitz type structure with $p= 30$.
     \item[D.3] (Toeplitz type $\bm{\Omega}$.): In the scenario D.3, the precision matrices are chosen to have a  Toeplitz type structure with $p = 50$.
\end{itemize}




\begin{table}[!ht]
\caption{\label{table_sim_tpr}Results of mean and standard deviation of relative estimation error (REE), true positive rate (TPR), and false positive rate (FPR) for estimated coefficients.}
\scriptsize
\centering
\resizebox{\linewidth}{!}{
\begin{tabular}{lccc|lccc|lccc} 
  \hline
&REE &  TPR& FPR & &REE & TPR& FPR & &REE & TPR& FPR \\ 
  \hline
   Sim B.1  &&& &Sim C.1  & &  &&Sim D.1  & &  \\
    & 0.3126 (0.1477)  & 1 & 0.0037 && 0.0816 (0.0365)& 1 & 0 && 0.3492 (0.0541) &0.9695 & 0.0103 \\ 
    Sim B.2 & & &&Sim C.2 & && &Sim D.2 & & \\
    & 0.1662 (0.0563) & 1 & 0.0037&  & 
   0.2592 (0.1446)  & 1 & 0&&0.0987 (0.0476)& 1 & 0  \\ 
    Sim B.3 &&&&Sim C.3  &  && &Sim D.3  & \\
    & 0.2559 (0.1036)  & 1 & 0.0037  &$b_n=30$& 0.0873 (0.0506) & 1 & 0 &&  0.1387 (0.0443)&  1 & 3e-04\\ 
     Sim B.4 && &&$b_n=50$& 0.1642 (0.1425)  & 1 & 0&\\
    $b_n = 20$ & 0.2691 (0.0875)& 1 & 0.0037 &$b_n=60$& 0.0993 (0.0494)& 1 & 0  \\
    $b_n = 30$ & 0.301 (0.155)& 1 & 0.0056 &Sim C.4  &&& \\
    $b_n = 40$& 0.3235 (0.1531)& 1 & 0.0037 & &  0.2046 (0.1535) & 1 & 0 &&&\\
    $b_n = 50$& 0.3407 (0.1557) & 1 & 0.0037&&&&& \\ 
   \hline
\end{tabular}}
\end{table}

\textit{ Setting E (Robustness of block size.)}. In the setting E, $n=5000$ and $p=25$, with the number of non-zero elements in $j$th segments $d_j= 3$, for all $j= 1, \dots, m_0 + 1$. The coefficient vector $\bm{\beta}$ are chosen to  have the random sparse structure in each segment, with different random entries sampled from $\mbox{Uniform}(-3,-1)\mathbbm{1}_{\{j \text{ is odd}\}} + \mbox{Uniform}(1,3)\mathbbm{1}_{\{j \text{ is even}\}}$, for each $j = 1, \dots, m_0 + 1$.  
We consider different setting of $m_0$ starting from 1 to 8.


\begin{table}[!ht]
\caption{\label{table_sim_D}Results of TBFL in simulation scenario E.}
\scriptsize
\centering
\begin{tabular}{l|cccc|cccc|cccc|cccc} 
  \hline
  & \multicolumn{4}{c|}{TBFL $b_n = 20$} & \multicolumn{4}{c|}{TBFL $b_n = 40$} & \multicolumn{4}{c|}{TBFL $b_n = 60$}  & \multicolumn{4}{c}{TBFL $b_n = 80$}\\
  $\vert \widehat{m} -m_0\vert$ & 0 &1 &2 & $>2$& 0 &1 &2 & $>2$& 0 &1 &2 & $>2$& 0 &1 &2 & $>2$ \\ 
  \hline 
  $m_0 =1$& 100 &   0 &   0 &   0 &  96 &   2 &   2 &   0 &  97 &   3 &   0 &   0 &  90 &  10 &   0 &   0 \\ 
 $m_0 =2$&  96 &   4 &   0 &   0 &  98 &   0 &   1 &   1 &  98 &   2 &   0 &   0 &  92 &   7 &   1 &   0 \\ 
 $m_0 =3$&  97 &   3 &   0 &   0 &  98 &   2 &   0 &   0 &  97 &   2 &   1 &   0 &  97 &   2 &   1 &   0 \\ 
 $m_0 =4$&  97 &   3 &   0 &   0 &  99 &   0 &   1 &   0 &  99 &   1 &   0 &   0 & 100 &   0 &   0 &   0 \\ 
$m_0 =5$&   99 &   1 &   0 &   0 & 100 &   0 &   0 &   0 &  97 &   3 &   0 &   0 &  97 &   3 &   0 &   0 \\ 
 $m_0 =6$&  98 &   2 &   0 &   0 & 100 &   0 &   0 &   0 &  99 &   1 &   0 &   0 &  95 &   4 &   1 &   0 \\ 
$m_0 =7$&  100 &   0 &   0 &   0 &  99 &   1 &   0 &   0 & 100 &   0 &   0 &   0 & 100 &   0 &   0 &   0 \\ 
 $m_0 =8$&  97 &   3 &   0 &   0 &  99 &   1 &   0 &   0 &  99 &   1 &   0 &   0 &  95 &   0 &   1 &   4 \\ 
  \hline
  \end{tabular}
  \end{table}


The performance of the TBFL algorithm is robust to the changes in the parameters’ zero/non-zero pattern, block size $b_n$, the number of time points $n$, the dimensions of response $p_y$ and the dimensions of predictor variables, and finally including many break points as investigated in simulation settings B through D. 
In scenario C.1, three change points are not equally spaced. Specifically, the first two change points are closer than the last two change points, which leads to a lower selection rate for the first change point.
In scenario C.4, a larger $p$ and smaller $n$ are chosen. Here to solve the numerical problem in the block fused lasso step, we applied a time-varying block sizes, where the blocks near the boundary are chosen to have a much larger block size (and smaller block sizes close to potential break time segments). Still, given that the $p$ is relatively large, the parameter estimates can be  unstable, which leads to a smaller true positive rate (TPR) as shown in Table~\ref{table_sim_tpr}. 
As shown in Table~\ref{table_sim_D}, the results are robust to the change of the block size. 

\section{Comparison with Simulated Annealing (SA) Method}\label{sec:sa}
\setcounter{equation}{0}
In this section, we compared the TBFL method to the Simulated Annealing (SA) method  \citep{bybee2018change} in terms of detection accuracy. The simulation setting is provided as follows:

\textit{ Setting F (Gaussian Graphical Model)}. In the scenario setting F, the precision matrices are chosen to have a Toeplitz type sparse structure, similar as scenario D.2, with $n=6000$, $p = 20$.
We consider different setting of $m_0$ starting from 2 to 8.


\begin{figure}[h]
\begin{center}
\includegraphics[width=0.6\textwidth, clip=TRUE, trim=0mm 85mm 0mm 85mm]{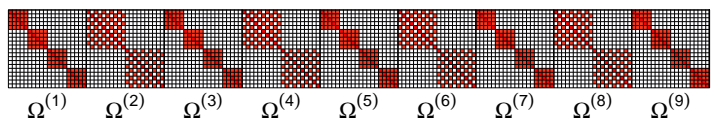}
\caption{True precision matrix $\Omega_x$ in Gaussian Graphical Model in simulation F ($m_0 = 8$).}
\label{fig:sim_ggm_F}
\end{center}
\end{figure}

\begin{figure}[!ht]
\begin{center}
\includegraphics[width=0.24\linewidth, clip=TRUE, trim=0mm 0mm 00mm 10mm]{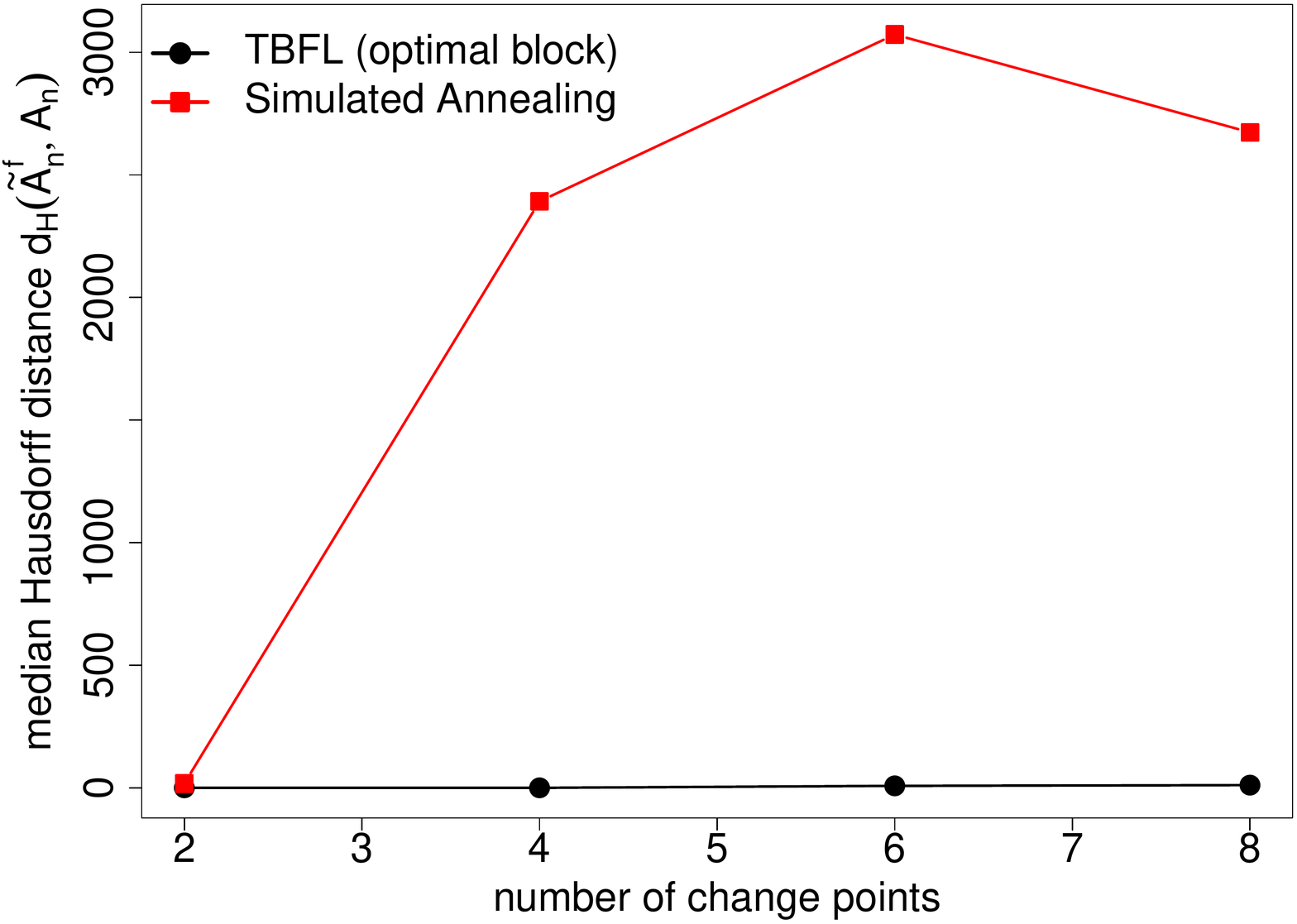}
\includegraphics[width=0.24\linewidth, clip=TRUE, trim=0mm 0mm 00mm 10mm]{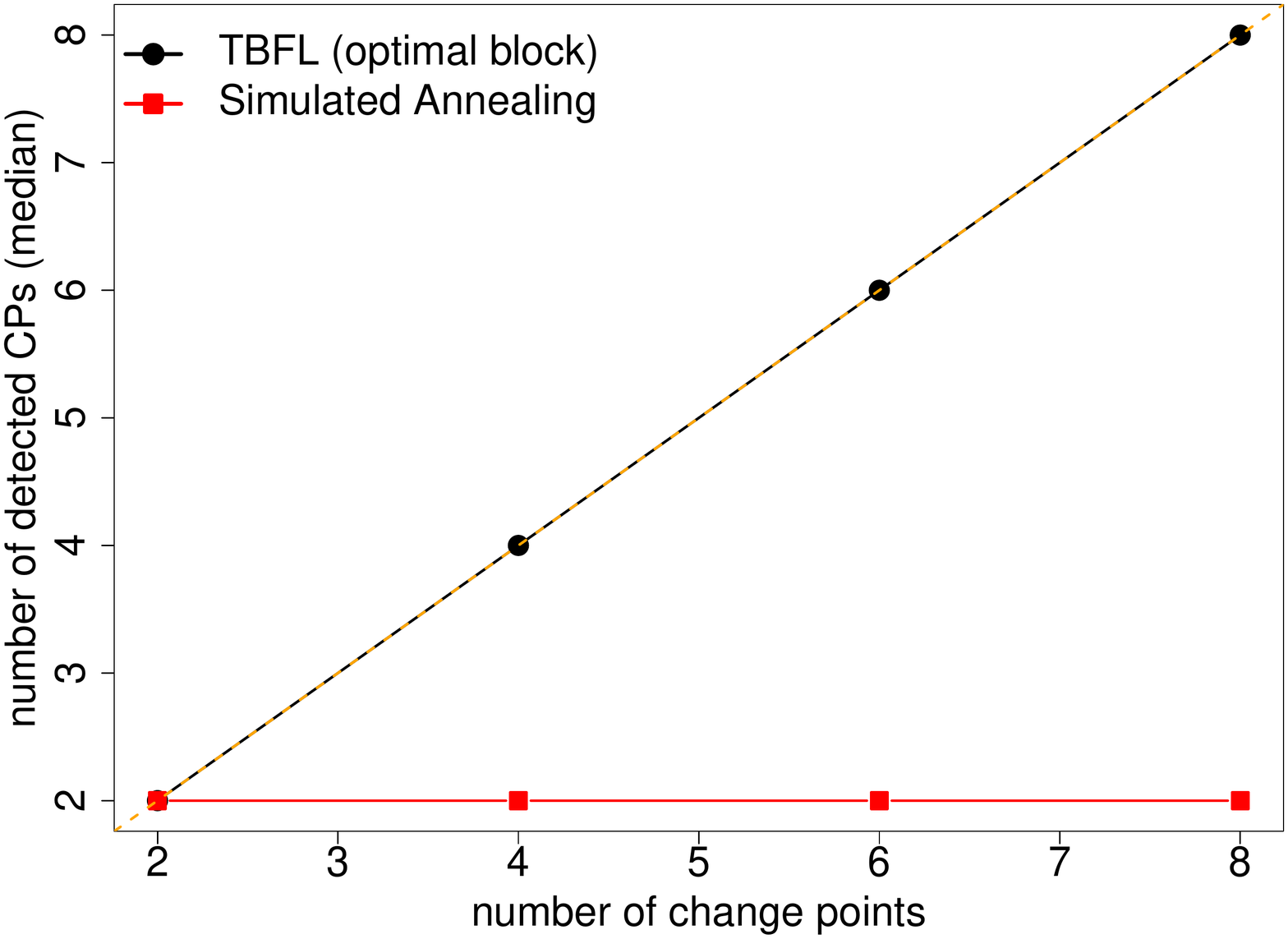}
\includegraphics[width=0.24\linewidth, clip=TRUE, trim=0mm 0mm 00mm 10mm]{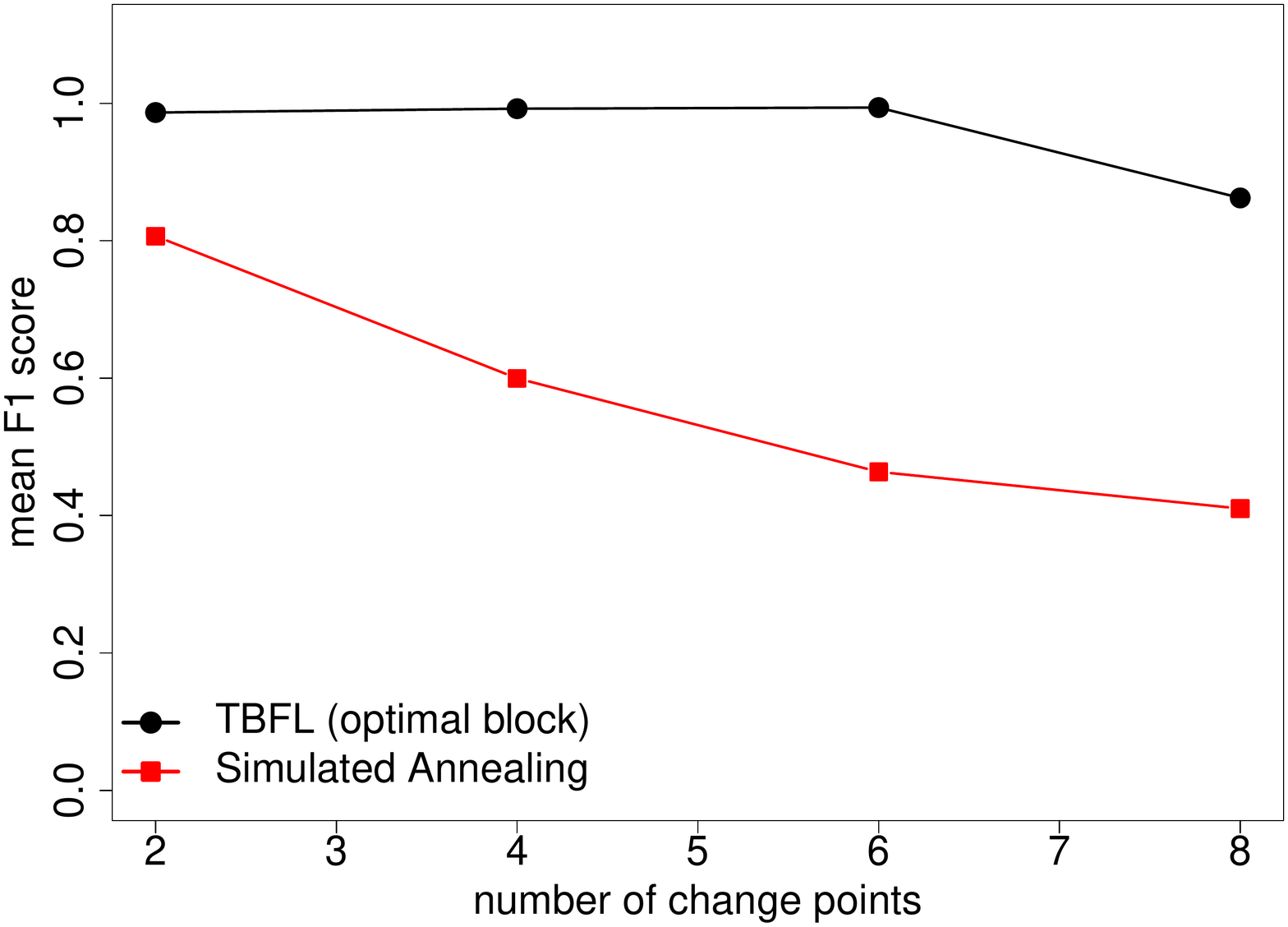}
\caption{
(a)  Hausdorff distance $d_H \left(  \widetilde{\mathcal{A}}_n^f, \mathcal{A}_n \right)$  for the TBFL and Simulated Annealing methods; 
(b) number of detected change points for the TBFL and Simulated Annealing methods. The 45 degree orange dashed line shows equality between the  number of detected change points on the vertical axis and the  number of true change points on the horizontal axis; 
(c) F1 score.
}
\label{fig:comput_performance_3}
\end{center}
\end{figure}

  \begin{table}[!ht]
\caption{\label{table_sim_F}Results of difference between $\widetilde{m}^f$ and $m_0$ for  TBFL and Simulated Annealing performance in simulation scenario F.}
\scriptsize
\centering
\begin{tabular}{l|cccc|cccc} 
  \hline
  & \multicolumn{4}{c|}{TBFL } & \multicolumn{4}{c}{SA}\\
  $\vert \widetilde{m}^f -m_0\vert$ & 0 &1 &2 & $>2$
  & 0 &1 &2 & $>2$\\ 
  \hline 
  $m_0 =2$&   94 & 5 & 1 & 0&
  50 &  46 &   2 &   2 \\ 
$m_0 =4$ & 93 & 7 & 0 & 0 &
13 &  25 &  28 &  34 \\ 
 $m_0 =6$&92 & 8 & 0 & 0 &
 2 &   7 &  10 &  81 \\ 
 $m_0 =8$ &66 & 15 & 1 & 18& 
 2 &  1 & 4 &  93 \\ 
  \hline
  \end{tabular}
  \end{table}
  
As shown in Table~\ref{table_sim_F}, among 100 replicates, our method can correctly estimate $m_0$ over 90\% replicates when $m_0 = 2$ to 6, while the SA method tends to underestimate the $m_0$. 
Note that our method also fails in estimating the number of the change points when the $m_0$ becomes larger (when $m_0 = 8$) due to the reason that the gap statistics and k-means method in block clustering step tend to fail when we have large cluster and different number of samples in each clusters.

Compared with the proposed Simulated Annealing (SA) method by \cite{bybee2018change}, our method have similar consistency rate but has better simulation results as shown in Table~\ref{table_sim_F}.  The Simulated Annealing (SA) method is not sensitive to the number of change point. 
\cite{bybee2018change} introduce majorize-minimize algorithm plus Simulated Annealing (SA) algorithm  for locating change points in large graphical models.
They also extend the method to multiple change-points by binary segmentation.
The rate of consistency for estimating the single change point locations is of order $O({d_n^\star}\log p)$ (see Theorem 9 in \cite{bybee2018change}).
Note that there is no consistency guarantee in terms of the number of change point for the method in \cite{bybee2018change} as the numerical results also verifies this.

\section{Details about EEG Data Pre-processing }\label{sec:eeg_more}

After removing the trend patterns in the raw EEG data, the pre-processing is accomplished in the following three steps: \textit{Step 1}: Partition the time series into sub-intervals with length $250$ and fit a VAR model for each of them; \textit{Step 2}:
Compute the residuals for each segment from the VAR model, scale the residuals and detect change points using TBFL assuming a Gaussian Graphical Model; 
\textit{Step 3}: Fit a VAR model again for each stationary segment from Step 2 and compute the residuals from the VAR model. The scaled residual data  in each stationary segment is then used to fit the Gaussian Graphical model for network connectivity analysis.

The Step 1 is to remove the possible temporal dependence in the EEG data. The Step 2 is to detect the change points in the Gaussian graphical model. Finally, the Step 3 is to get a more precise result in network connectivity.



\par



\bibhang=1.7pc
\bibsep=2pt
\fontsize{9}{14pt plus.8pt minus .6pt}\selectfont
\renewcommand\bibname{\large \bf References}
\expandafter\ifx\csname
natexlab\endcsname\relax\def\natexlab#1{#1}\fi
\expandafter\ifx\csname url\endcsname\relax
  \def\url#1{\texttt{#1}}\fi
\expandafter\ifx\csname urlprefix\endcsname\relax\def\urlprefix{URL}\fi

 \bibliographystyle{chicago}      
\bibliography{Reference}

\begin{thebibliography}{}

\bibitem[\protect\citeauthoryear{Aue and Horv{\'a}th}{Aue and
  Horv{\'a}th}{2013}]{aue2013structural}
Aue, A. and L.~Horv{\'a}th (2013).
\newblock Structural breaks in time series.
\newblock {\em Journal of Time Series Analysis\/}~{\em 34\/}(1), 1--16.

\bibitem[\protect\citeauthoryear{Aue, Rice, and S{\"o}nmez}{Aue
  et~al.}{2017}]{aue2017detecting}
Aue, A., G.~Rice, and O.~S{\"o}nmez (2017).
\newblock Detecting and dating structural breaks in functional data without
  dimension reduction.
\newblock {\em Journal of the Royal Statistical Society: Series B (Statistical
  Methodology)\/}.

\bibitem[\protect\citeauthoryear{Bai, Safikhani, and Michailidis}{Bai
  et~al.}{2020}]{bai2020multiple}
Bai, P., A.~Safikhani, and G.~Michailidis (2020).
\newblock Multiple change points detection in low rank and sparse high
  dimensional vector autoregressive models.
\newblock {\em IEEE Transactions on Signal Processing\/}~{\em 68}, 3074--3089.

\bibitem[\protect\citeauthoryear{Bai and Safikhani}{Bai and
  Safikhani}{2021}]{lineardetect}
Bai, Y. and A.~Safikhani (2021).
\newblock {\em LinearDetect: Change Point Detection in High-Dimensional Linear
  Regression Models}.
\newblock R package version 0.1.4.

\bibitem[\protect\citeauthoryear{Basseville and Nikiforov}{Basseville and
  Nikiforov}{1993}]{Basseville1993detection}
Basseville, M. and I.~V. Nikiforov (1993).
\newblock {\em Detection of abrupt changes: theory and application}, Volume
  104.
\newblock Prentice Hall Englewood Cliffs.

\bibitem[\protect\citeauthoryear{Basu and Michailidis}{Basu and
  Michailidis}{2015}]{Basu_2015}
Basu, S. and G.~Michailidis (2015).
\newblock Regularized estimation in sparse high-dimensional time series models.
\newblock {\em The Annals of Statistics\/}~{\em 43\/}(4), 1535--1567.

\bibitem[\protect\citeauthoryear{Bickel, Ritov, Tsybakov, et~al.}{Bickel
  et~al.}{2009}]{bickel2009simultaneous}
Bickel, P.~J., Y.~Ritov, A.~B. Tsybakov, et~al. (2009).
\newblock Simultaneous analysis of lasso and dantzig selector.
\newblock {\em The Annals of Statistics\/}~{\em 37\/}(4), 1705--1732.

\bibitem[\protect\citeauthoryear{Bleakley and Vert}{Bleakley and
  Vert}{2011}]{bleakley2011group}
Bleakley, K. and J.-P. Vert (2011).
\newblock The group fused lasso for multiple change-point detection.
\newblock {\em arXiv preprint arXiv:1106.4199\/}.

\bibitem[\protect\citeauthoryear{Bybee and Atchad{\'e}}{Bybee and
  Atchad{\'e}}{2018}]{bybee2018change}
Bybee, L. and Y.~Atchad{\'e} (2018).
\newblock Change-point computation for large graphical models: a scalable
  algorithm for gaussian graphical models with change-points.
\newblock {\em The Journal of Machine Learning Research\/}~{\em 19\/}(1),
  440--477.

\bibitem[\protect\citeauthoryear{Chan, Ng, and Yau}{Chan
  et~al.}{2021}]{chan2021self}
Chan, N.~H., W.~L. Ng, and C.~Y. Yau (2021).
\newblock A self-normalized approach to sequential change-point detection for
  time series.
\newblock {\em Statistica Sinica\/}~{\em 31\/}(1), 491--517.

\bibitem[\protect\citeauthoryear{Cho}{Cho}{2016}]{cho2016change}
Cho, H. (2016).
\newblock Change-point detection in panel data via double cusum statistic.
\newblock {\em Electronic Journal of Statistics\/}~{\em 10\/}(2), 2000--2038.

\bibitem[\protect\citeauthoryear{Cho and Fryzlewicz}{Cho and
  Fryzlewicz}{2015}]{cho2015multiple}
Cho, H. and P.~Fryzlewicz (2015).
\newblock Multiple-change-point detection for high dimensional time series via
  sparsified binary segmentation.
\newblock {\em Journal of the Royal Statistical Society: Series B (Statistical
  Methodology)\/}~{\em 77\/}(2), 475--507.

\bibitem[\protect\citeauthoryear{Cs{\"o}rg{\"o} and Horv{\'a}th}{Cs{\"o}rg{\"o}
  and Horv{\'a}th}{1997}]{csorgo1997limit}
Cs{\"o}rg{\"o}, M. and L.~Horv{\'a}th (1997).
\newblock {\em Limit theorems in change-point analysis}, Volume~18.
\newblock John Wiley \& Sons Inc.

\bibitem[\protect\citeauthoryear{Davis, Lee, and Rodriguez-Yam}{Davis
  et~al.}{2006}]{davis2006structural}
Davis, R.~A., T.~C.~M. Lee, and G.~A. Rodriguez-Yam (2006).
\newblock Structural break estimation for nonstationary time series models.
\newblock {\em Journal of the American Statistical Association\/}~{\em
  101\/}(473), 223--239.

\bibitem[\protect\citeauthoryear{Frick, Munk, and Sieling}{Frick
  et~al.}{2014}]{frick2014multiscale}
Frick, K., A.~Munk, and H.~Sieling (2014).
\newblock Multiscale change point inference.
\newblock {\em Journal of the Royal Statistical Society: Series B (Statistical
  Methodology)\/}~{\em 76\/}(3), 495--580.

\bibitem[\protect\citeauthoryear{Friedman, Hastie, and Tibshirani}{Friedman
  et~al.}{2010}]{friedman2010regularization}
Friedman, J., T.~Hastie, and R.~Tibshirani (2010).
\newblock Regularization paths for generalized linear models via coordinate
  descent.
\newblock {\em Journal of statistical software\/}~{\em 33\/}(1), 1.

\bibitem[\protect\citeauthoryear{Fris{\'e}n}{Fris{\'e}n}{2008}]{frisen2008financial}
Fris{\'e}n, M. (2008).
\newblock {\em Financial surveillance}, Volume~71.
\newblock John Wiley \& Sons.

\bibitem[\protect\citeauthoryear{Fryzlewicz}{Fryzlewicz}{2017}]{fryzlewicz2017tail}
Fryzlewicz, P. (2017).
\newblock Tail-greedy bottom-up data decompositions and fast mulitple
  change-point detection.
\newblock {\em Annals of Statistics\/}.

\bibitem[\protect\citeauthoryear{Gibberd and Roy}{Gibberd and
  Roy}{2017}]{gibberd2017multiple}
Gibberd, A.~J. and S.~Roy (2017).
\newblock Multiple changepoint estimation in high-dimensional gaussian
  graphical models.
\newblock {\em arXiv preprint arXiv:1712.05786\/}.

\bibitem[\protect\citeauthoryear{Harchaoui and L{\'e}vy-Leduc}{Harchaoui and
  L{\'e}vy-Leduc}{2010}]{harchaoui2010multiple}
Harchaoui, Z. and C.~L{\'e}vy-Leduc (2010).
\newblock Multiple change-point estimation with a total variation penalty.
\newblock {\em Journal of the American Statistical Association\/}~{\em
  105\/}(492), 1480--1493.

\bibitem[\protect\citeauthoryear{Hartigan and Wong}{Hartigan and
  Wong}{1979}]{hartigan1979algorithm}
Hartigan, J.~A. and M.~A. Wong (1979).
\newblock Algorithm as 136: A k-means clustering algorithm.
\newblock {\em Journal of the royal statistical society. series c (applied
  statistics)\/}~{\em 28\/}(1), 100--108.

\bibitem[\protect\citeauthoryear{Hastie, Tibshirani, and Friedman}{Hastie
  et~al.}{2009}]{hastie2009elements}
Hastie, T., R.~Tibshirani, and J.~Friedman (2009).
\newblock {\em The elements of statistical learning: data mining, inference,
  and prediction}.
\newblock Springer Science \& Business Media.

\bibitem[\protect\citeauthoryear{Hushchyn, Arzymatov, and Derkach}{Hushchyn
  et~al.}{2020}]{hushchyn2020online}
Hushchyn, M., K.~Arzymatov, and D.~Derkach (2020).
\newblock Online neural networks for change-point detection.
\newblock {\em arXiv preprint arXiv:2010.01388\/}.

\bibitem[\protect\citeauthoryear{Jackson, Scargle, Barnes, Arabhi, Alt,
  Gioumousis, Gwin, Sangtrakulcharoen, Tan, and Tsai}{Jackson
  et~al.}{2005}]{jackson2005algorithm}
Jackson, B., J.~D. Scargle, D.~Barnes, S.~Arabhi, A.~Alt, P.~Gioumousis,
  E.~Gwin, P.~Sangtrakulcharoen, L.~Tan, and T.~T. Tsai (2005).
\newblock An algorithm for optimal partitioning of data on an interval.
\newblock {\em IEEE Signal Processing Letters\/}~{\em 12\/}(2), 105--108.

\bibitem[\protect\citeauthoryear{Killick, Fearnhead, and Eckley}{Killick
  et~al.}{2012}]{killick2012optimal}
Killick, R., P.~Fearnhead, and I.~A. Eckley (2012).
\newblock Optimal detection of changepoints with a linear computational cost.
\newblock {\em Journal of the American Statistical Association\/}~{\em
  107\/}(500), 1590--1598.

\bibitem[\protect\citeauthoryear{Kolar and Xing}{Kolar and
  Xing}{2012}]{kolar2012estimating}
Kolar, M. and E.~P. Xing (2012).
\newblock Estimating networks with jumps.
\newblock {\em Electronic journal of statistics\/}~{\em 6}, 2069.

\bibitem[\protect\citeauthoryear{Leonardi and B{\"u}hlmann}{Leonardi and
  B{\"u}hlmann}{2016}]{leonardi2016computationally}
Leonardi, F. and P.~B{\"u}hlmann (2016).
\newblock Computationally efficient change point detection for high-dimensional
  regression.
\newblock {\em arXiv preprint arXiv:1601.03704\/}.

\bibitem[\protect\citeauthoryear{Liu, Zhou, Zhang, and Liu}{Liu
  et~al.}{2020}]{liu2020unified}
Liu, B., C.~Zhou, X.~Zhang, and Y.~Liu (2020).
\newblock A unified data-adaptive framework for high dimensional change point
  detection.
\newblock {\em Journal of the Royal Statistical Society: Series B (Statistical
  Methodology)\/}~{\em 82\/}(4), 933--963.

\bibitem[\protect\citeauthoryear{Loh and Wainwright}{Loh and
  Wainwright}{2012}]{loh2012}
Loh, P.-L. and M.~J. Wainwright (2012, 06).
\newblock High-dimensional regression with noisy and missing data: Provable
  guarantees with nonconvexity.
\newblock {\em Ann. Statist.\/}~{\em 40\/}(3), 1637--1664.

\bibitem[\protect\citeauthoryear{L{\"u}tkepohl}{L{\"u}tkepohl}{2005}]{lutkepohl2005new}
L{\"u}tkepohl, H. (2005).
\newblock {\em New introduction to multiple time series analysis}.
\newblock Springer Science \& Business Media.

\bibitem[\protect\citeauthoryear{Matteson and James}{Matteson and
  James}{2014}]{matteson2014nonparametric}
Matteson, D.~S. and N.~A. James (2014).
\newblock A nonparametric approach for multiple change point analysis of
  multivariate data.
\newblock {\em Journal of the American Statistical Association\/}~{\em
  109\/}(505), 334--345.

\bibitem[\protect\citeauthoryear{Meinshausen and B{\"u}hlmann}{Meinshausen and
  B{\"u}hlmann}{2006}]{meinshausen2006high}
Meinshausen, N. and P.~B{\"u}hlmann (2006).
\newblock High-dimensional graphs and variable selection with the lasso.
\newblock {\em Annals of statistics\/}~{\em 34\/}(3), 1436--1462.

\bibitem[\protect\citeauthoryear{Nezamfar, Orhan, Purwar, Hild, Oken, and
  Erdogmus}{Nezamfar et~al.}{2011}]{EEGvisual}
Nezamfar, H., U.~Orhan, S.~Purwar, K.~Hild, B.~Oken, and D.~Erdogmus (2011).
\newblock Decoding of multichannel eeg activity from the visual cortex in
  response to pseudorandom binary sequences of visual stimuli.
\newblock {\em International Journal of Imaging Systems and Technology\/}~{\em
  21\/}(2), 139--147.

\bibitem[\protect\citeauthoryear{Ombao, Von~Sachs, and Guo}{Ombao
  et~al.}{2005}]{OmbaoVonSachsGuo_2005}
Ombao, H., R.~Von~Sachs, and W.~Guo (2005).
\newblock Slex analysis of multivariate nonstationary time series.
\newblock {\em Journal of the American Statistical Association\/}~{\em
  100\/}(470), 519--531.

\bibitem[\protect\citeauthoryear{Qiu}{Qiu}{2013}]{qiu2013introduction}
Qiu, P. (2013).
\newblock {\em Introduction to statistical process control}.
\newblock CRC press.

\bibitem[\protect\citeauthoryear{Rinaldo}{Rinaldo}{2009}]{rinaldo2009properties}
Rinaldo, A. (2009).
\newblock Properties and refinements of the fused lasso.
\newblock {\em Annals of Statistics\/}~{\em 37\/}(5B), 2922--2952.

\bibitem[\protect\citeauthoryear{Rothman, Levina, and Zhu}{Rothman
  et~al.}{2010}]{rothman2010sparse}
Rothman, A.~J., E.~Levina, and J.~Zhu (2010).
\newblock Sparse multivariate regression with covariance estimation.
\newblock {\em Journal of Computational and Graphical Statistics\/}~{\em
  19\/}(4), 947--962.

\bibitem[\protect\citeauthoryear{Roy, Atchad{\'e}, and Michailidis}{Roy
  et~al.}{2017}]{roy2017change}
Roy, S., Y.~Atchad{\'e}, and G.~Michailidis (2017).
\newblock Change point estimation in high dimensional markov random-field
  models.
\newblock {\em Journal of the Royal Statistical Society: Series B (Statistical
  Methodology)\/}~{\em 79\/}(4), 1187--1206.

\bibitem[\protect\citeauthoryear{Safikhani, Bai, and Michailidis}{Safikhani
  et~al.}{2021}]{safikhani2021fast}
Safikhani, A., Y.~Bai, and G.~Michailidis (2021).
\newblock Fast and scalable algorithm for detection of structural breaks in big
  var models.
\newblock {\em Journal of Computational and Graphical Statistics\/}, 1--14.

\bibitem[\protect\citeauthoryear{Safikhani and Shojaie}{Safikhani and
  Shojaie}{2020}]{safikhani2020joint}
Safikhani, A. and A.~Shojaie (2020).
\newblock Joint structural break detection and parameter estimation in
  high-dimensional nonstationary var models.
\newblock {\em Journal of the American Statistical Association\/}, 1--14.

\bibitem[\protect\citeauthoryear{Savage, Zhang, Yu, Chou, and Wang}{Savage
  et~al.}{2014}]{savage2014anomaly}
Savage, D., X.~Zhang, X.~Yu, P.~Chou, and Q.~Wang (2014).
\newblock Anomaly detection in online social networks.
\newblock {\em Social Networks\/}~{\em 39}, 62--70.

\bibitem[\protect\citeauthoryear{Schwarz et~al.}{Schwarz
  et~al.}{1978}]{schwarz1978estimating}
Schwarz, G. et~al. (1978).
\newblock Estimating the dimension of a model.
\newblock {\em The annals of statistics\/}~{\em 6\/}(2), 461--464.

\bibitem[\protect\citeauthoryear{Tibshirani, Saunders, Rosset, Zhu, and
  Knight}{Tibshirani et~al.}{2005}]{tibshirani2005sparsity}
Tibshirani, R., M.~Saunders, S.~Rosset, J.~Zhu, and K.~Knight (2005).
\newblock Sparsity and smoothness via the fused lasso.
\newblock {\em Journal of the Royal Statistical Society: Series B (Statistical
  Methodology)\/}~{\em 67\/}(1), 91--108.

\bibitem[\protect\citeauthoryear{Tibshirani, Walther, and Hastie}{Tibshirani
  et~al.}{2001}]{tibshirani2001estimating}
Tibshirani, R., G.~Walther, and T.~Hastie (2001).
\newblock Estimating the number of clusters in a data set via the gap
  statistic.
\newblock {\em Journal of the Royal Statistical Society: Series B (Statistical
  Methodology)\/}~{\em 63\/}(2), 411--423.

\bibitem[\protect\citeauthoryear{Trujillo}{Trujillo}{2019}]{ANS9Q1_2019}
Trujillo, L. (2019).
\newblock {Raw Empirical EEG Data}.

\bibitem[\protect\citeauthoryear{van~de Geer, B{\"u}hlmann, Zhou,
  et~al.}{van~de Geer et~al.}{2011}]{van2011adaptive}
van~de Geer, S., P.~B{\"u}hlmann, S.~Zhou, et~al. (2011).
\newblock The adaptive and the thresholded lasso for potentially misspecified
  models (and a lower bound for the lasso).
\newblock {\em Electronic Journal of Statistics\/}~{\em 5}, 688--749.

\bibitem[\protect\citeauthoryear{Vershynin}{Vershynin}{2010}]{vershynin2010introduction}
Vershynin, R. (2010).
\newblock Introduction to the non-asymptotic analysis of random matrices.
\newblock {\em arXiv preprint arXiv:1011.3027\/}.

\bibitem[\protect\citeauthoryear{Wang, Lin, and Willett}{Wang
  et~al.}{2019}]{wang2019statistically}
Wang, D., K.~Lin, and R.~Willett (2019).
\newblock Statistically and computationally efficient change point localization
  in regression settings.
\newblock {\em arXiv preprint arXiv:1906.11364\/}.

\bibitem[\protect\citeauthoryear{Wang, Yu, Rinaldo, and Willett}{Wang
  et~al.}{2019}]{wang2019localizing}
Wang, D., Y.~Yu, A.~Rinaldo, and R.~Willett (2019).
\newblock Localizing changes in high-dimensional vector autoregressive
  processes.
\newblock {\em arXiv preprint arXiv:1909.06359\/}.

\bibitem[\protect\citeauthoryear{Wang, Ren, and Gu}{Wang
  et~al.}{2016}]{wang2016precision}
Wang, L., X.~Ren, and Q.~Gu (2016).
\newblock Precision matrix estimation in high dimensional gaussian graphical
  models with faster rates.
\newblock In {\em Artificial Intelligence and Statistics}, pp.\  177--185.

\bibitem[\protect\citeauthoryear{Wang and Samworth}{Wang and
  Samworth}{2016}]{wang2016high}
Wang, T. and R.~J. Samworth (2016).
\newblock High-dimensional changepoint estimation via sparse projection.
\newblock {\em arXiv preprint arXiv:1606.06246\/}.

\bibitem[\protect\citeauthoryear{Wang and Zhu}{Wang and
  Zhu}{2011}]{wang2011consistent}
Wang, T. and L.~Zhu (2011).
\newblock Consistent tuning parameter selection in high dimensional sparse
  linear regression.
\newblock {\em Journal of Multivariate Analysis\/}~{\em 102\/}(7), 1141--1151.

\bibitem[\protect\citeauthoryear{Wang and Mei}{Wang and
  Mei}{2015}]{wang2015large}
Wang, Y. and Y.~Mei (2015).
\newblock Large-scale multi-stream quickest change detection via shrinkage
  post-change estimation.
\newblock {\em IEEE Transactions on Information Theory\/}~{\em 61\/}(12),
  6926--6938.

\bibitem[\protect\citeauthoryear{Yu}{Yu}{2020}]{yu2020review}
Yu, Y. (2020).
\newblock A review on minimax rates in change point detection and localisation.
\newblock {\em arXiv preprint arXiv:2011.01857\/}.

\bibitem[\protect\citeauthoryear{Yuan and Lin}{Yuan and
  Lin}{2007}]{yuan2007model}
Yuan, M. and Y.~Lin (2007).
\newblock Model selection and estimation in the gaussian graphical model.
\newblock {\em Biometrika\/}~{\em 94\/}(1), 19--35.

\bibitem[\protect\citeauthoryear{Zhang and Lavitas}{Zhang and
  Lavitas}{2018}]{zhang2018unsupervised}
Zhang, T. and L.~Lavitas (2018).
\newblock Unsupervised self-normalized change-point testing for time series.
\newblock {\em Journal of the American Statistical Association\/}~{\em
  113\/}(522), 637--648.

\bibitem[\protect\citeauthoryear{Zhou, R{\"u}timann, Xu, and B{\"u}hlmann}{Zhou
  et~al.}{2011}]{zhou2011high}
Zhou, S., P.~R{\"u}timann, M.~Xu, and P.~B{\"u}hlmann (2011).
\newblock High-dimensional covariance estimation based on gaussian graphical
  models.
\newblock {\em The Journal of Machine Learning Research\/}~{\em 12},
  2975--3026.

\bibitem[\protect\citeauthoryear{Zhu, Pan, Li, Liu, and Wang}{Zhu
  et~al.}{2017}]{zhu2017network}
Zhu, X., R.~Pan, G.~Li, Y.~Liu, and H.~Wang (2017).
\newblock Network vector autoregression.
\newblock {\em The Annals of Statistics\/}~{\em 45\/}(3), 1096--1123.

\end{thebibliography}







\vskip .65cm
\noindent
Department of Statistics, University of Florida
\vskip 2pt
\noindent
E-mail: (baiyue@ufl.edu)
\vskip 2pt

\noindent
Department of Statistics, University of Florida
\vskip 2pt
\noindent
E-mail: (a.safikhani@ufl.edu)


\end{document}